\documentclass[11pt]{article}
\usepackage[utf8]{inputenc}

\usepackage{setspace}
\usepackage{mathtools}
\usepackage{hyperref}
\usepackage{color}
\usepackage{algorithm}
\usepackage{algpseudocode}
\usepackage{epsfig}
\usepackage{epstopdf}
\usepackage[section]{placeins}
\usepackage{graphicx}
\usepackage{amsmath, amssymb}
\usepackage{amsfonts}
\usepackage{url}
\usepackage{multirow}
\usepackage{array}
\usepackage[footnotesize]{subfigure}
\usepackage{latexsym}
\usepackage{natbib}
\usepackage{booktabs}
\usepackage{multicol}
\setcitestyle{authoryear,open={(},close={)}}
\usepackage{pgfplots}
\pgfplotsset{compat=1.18}
\usepackage{pgfplotstable}
\usepackage{stackengine}
\usetikzlibrary{calc,patterns}
\usepackage{enumitem}

\usepackage{bbm}

\allowdisplaybreaks

\usepackage{amsthm}
\usepackage{xpatch}
\makeatletter
\AtBeginDocument{\xpatchcmd{\@thm}{\thm@headpunct{.}}{\thm@headpunct{}}{}{}}
\makeatother
\newtheorem{theorem}{Theorem}[]
\newtheorem*{theorem*}{Theorem}
\newtheorem{lemma}{Lemma}
\newtheorem{proposition}{Proposition}

\theoremstyle{definition}
\newtheorem{definition}{Definition}
\newtheorem{assumption}{Assumption}

\newtheorem{remark}{Remark}

\newcommand{\edit}[1]{{\color{black}#1}}

\newcommand{\vs}{\vspace*{3mm}}
\newcommand{\lt}{\left}
\newcommand{\rt}{\right}

\def\be{{\bf e}}

\def\bI{{\bf I}}

\def\bp{{\bf p}}

\def\bR{{\bf R}}

\def\bS{{\bf S}}
\def\bs{{\bf s}}

\def\bU{{\bf U}}
\def\bu{{\bf u}}

\def\bX{{\bf X}}
\def\bx{{\bf x}}

\def\bZ{{\bf Z}}
\def\bz{{\bf z}}

\def\bLambda{{\bf \Lambda}}
\def\blambda{\boldsymbol\lambda}

\def\bPi{{\bf \Pi}}

\def\bsigma{\mbox{\boldmath $\sigma$}}

\def\btheta{{\boldsymbol\theta}}

\def\bxi{{\boldsymbol\xi}}

\def\cG{{\cal G}}

\def\cL{{\cal L}}

\def\cN{{\cal N}}

\def\sE{{\sf E}}

\def\sP{{\sf P}}

\def\bbR{\mathbb{R}}

\def\rd{{\rm d}}
\def\tc{{\tt c}}
\def\f{{\tt f}}

\def\sone{\mathbbm{1}}

\def\bzero{\mathbf{0}}

\def\rd{{\rm d}}

\def\blot{\quad {$\vcenter{\vbox{\hrule height.4pt
             \hbox{\vrule width.4pt height.9ex \kern.9ex \vrule
width.4pt}
             \hrule height.4pt}}$}}

\newenvironment{revision}{\begingroup\color{black}}{\endgroup}

\newcommand{\GIG}{\operatorname{GIG}}

\mathtoolsset{showonlyrefs}
\usepackage{lipsum}
\usepackage{rotating}
\usepackage{setspace}

\begin{document}

\title{Efficient Monte Carlo Valuation of Corporate Bonds in Financial Networks}
\author{Dohyun Ahn\thanks{Department of Systems Engineering and Engineering Management, E-mail: \href{mailto:dohyun.ahn@cuhk.edu.hk}{\tt dohyun.ahn@cuhk.edu.hk}}\\
{\small \it The Chinese University of Hong Kong}
\vspace{10pt}\\
Agostino Capponi\thanks{Department of Industrial Engineering and Operations Research and Columbia Business School, E-mail: \href{mailto:ac3827@columbia.edu}{\tt ac3827@columbia.edu}
}\\
{\small \it Columbia University}
}

\date{\today}
\maketitle
\begin{abstract}
\noindent

Valuing corporate bonds in systemic economies is challenging due to intricate webs of inter-institutional exposures. When a bank defaults, cascading losses propagate through the network, with payments determined by a system of fixed-point equations lacking closed-form solutions. Standard Monte Carlo methods cannot capture rare yet critical default events, while existing rare-event simulation techniques fail to account for higher-order network effects and scale poorly with network size.  To overcome these challenges, we propose a novel approach---\emph{Bi-Level Importance Sampling with Splitting}---and characterize individual bank defaults by decoupling them from the network's complex fixed-point dynamics. This separation enables a two-stage estimation process that directly generates samples from the banks' default events. We demonstrate theoretically that the method is both scalable and asymptotically optimal, and validate its effectiveness through numerical studies on \edit{both synthetic and} empirically observed networks.

\vs
\end{abstract}

\section{Introduction}\label{sec:intro}
In modern financial systems, institutions are intricately linked through common balance sheet holdings \citep{Capponi2015,GreenwoodEtAl2015,DuarteEisenbach2021FireSales} and interbank liabilities. This interconnectedness can be destabilizing and lead to the propagation of contagion and defaults during distressed periods. As demonstrated during financial crises, such as the 2008 global financial crisis, a financial institution can withstand a moderate shock, yet be forced into default by contagion due to the propagation of losses originating from other distressed firms within the network  \citep{GaiHaldaneKapadia2011}. 

The risk management literature has extensively analyzed financial contagion mechanisms and their impact on institutional and systemic stability, accounting for network exposures and default-triggered frictions. However, a critical gap remains: the pricing of corporate bonds subject to network-induced credit risk has been largely overlooked, despite its practical importance. This stands in contrast to the vast literature on corporate bond pricing, a central topic in financial engineering for decades. Our paper addresses this gap by proposing an efficient valuation framework for bonds of firms exposed to systemic credit risk.

The fundamental challenge in this valuation stems from the network effects themselves. Their complexity precludes any closed-form expression for the bond price, which is a key distinction from classical non-networked models, where analytical characterizations are typical. Arguably, the most tractable computational remedy would be the use of standard Monte Carlo simulation, where random shocks to each firm's assets are generated from the underlying distributions to simulate the resulting contagion, and the bond price is computed at an equilibrium where the cascade of losses and defaults stops. However, this method is notoriously inefficient in capturing rare but critical default events that are the key drivers of credit risk.

A judiciously designed Monte Carlo method can dramatically improve efficiency through variance reduction. The theoretically optimal (zero-variance) estimator, however, requires prior knowledge of the very quantity being estimated and is thus impractical to implement. Consequently, this paper focuses on designing a tractable Monte Carlo method that achieves near-optimal variance reduction by efficiently sampling extreme events that cause defaults. While Monte Carlo methods for rare-event simulation have been studied extensively, the literature on their application to financial networks remains sparse, and the few existing methods face three significant challenges:
\begin{enumerate}[label=(\roman*)]
    \item\label{item:network effects} \emph{Higher-order network effects.} Classical rare-event simulation techniques often rely on an explicit characterization of the event of interest. As previously noted, this is not possible in financial networks where contagion extends beyond direct, first-order exposures. The failure of one institution can trigger losses at its counterparties, which in turn propagate distress to their own counterparties, initiating indirect, higher-order cascades. These dynamics can cause shocks to circulate and amplify until the system settles into a new equilibrium, often described by a fixed-point equation. Consequently, the critical events we focus on are defined only implicitly, making traditional rare-event simulation methods inapplicable.

    \item\label{item:curse of dimensionality} \emph{The curse of dimensionality.} Despite the implicit nature of these events, a few Monte Carlo methods have been proposed for estimating default probabilities, particularly in the classical model of \cite{Eisenberg:01}.
    While these methods lead to accurate estimates for small networks (e.g., ten or fewer institutions), their computational time increases exponentially with network size.
    Given that real-world financial systems involve hundreds or thousands of interconnected institutions, a scalable method whose performance is robust to high dimensionality is imperative.

    \item\label{item:sensitivity} \emph{Network model sensitivity.} The existing methods mentioned above suffer from a further limitation: they are tailored to a specific model. They are built upon a linear programming reformulation of the equilibrium in the model of \cite{Eisenberg:01}, which only accounts for contagion through interbank liabilities. This rigid structure prevents the incorporation of other critical, non-linear contagion channels, such as fire-sale externalities and bankruptcy costs. As these non-linear effects have received significant recent attention for their realism, this motivates the design of a flexible Monte Carlo method applicable to a wide range of modern financial network models.
\end{enumerate}
To the best of our knowledge, our paper is the first to introduce an efficient Monte Carlo method for bond valuation that simultaneously addresses all three of these challenges. Our model incorporates common asset holdings and interbank liabilities, which are widely recognized as the two most prominent channels of financial contagion~\citep{ClercEtAl2016,Glasserman:16}. Specifically, distressed institutions may be forced to sell off some of their illiquid asset holdings to meet their debt obligations. This fire sale depresses asset prices and subsequently creates a negative externality, reducing the value of the same assets held on the balance sheets of all other banks.
Furthermore, if a distressed bank defaults on its obligations even after the asset liquidation, it distributes all its residual assets to pay its creditors. Due to the shortfall in payments, the creditors recover only a fraction of their claims, and these credit losses can then trigger further defaults, creating loss amplification among insolvent institutions.

In this model, we first explicitly characterize an institution's default condition conditional on the shock realizations of all other institutions. As discussed in ~\ref{item:network effects}, default events cannot be defined explicitly in the full-dimensional space of external shocks due to loss amplification and default cascade dynamics. This complexity makes it difficult to develop efficient simulation methods that sample directly from the default region. We overcome this by identifying a critical threshold for a target institution's external assets, below which it defaults given the shocks to the rest of the system. This is achieved by strategically decoupling the target institution from the network and analyzing the impact of amplified losses from the separated system back onto the target. This impact is what defines the critical threshold. \edit{This threshold can be computed by solving a fictitious clearing
system for the other banks. The fictitious default algorithm requires
at most one round per non-target bank, avoiding the exponential
enumeration of default configurations.} Furthermore, this decoupling strategy is not limited to our specific financial model but is broadly applicable. Accordingly, our characterization provides a unified solution to the challenges outlined in \ref{item:network effects} to \ref{item:sensitivity}.

We then leverage this default set characterization to design a novel and efficient Monte Carlo method for pricing institutional bonds: Bi-Level Importance Sampling with Splitting (BLISS). The proposed method employs a splitting technique that decomposes the default event into a sequence of two nested, more tractable sub-events. This bi-level framework aligns perfectly with the decoupling approach mentioned above: the outer-level event corresponds to shock realizations in non-target institutions, while the inner-level event represents the conditional default of the target institution itself. To achieve significant variance reduction, BLISS integrates this splitting framework with two distinct importance sampling schemes, one for each level. Specifically,  \edit{the outer layer samples the non-target return factors under their marginal distribution conditional on a conservative event containing the target bank's default event. We provide an exact sampler for this conditional distribution. Given the outer-layer variables, the inner layer samples the target-specific Gaussian factor conditional on actual default 
via optimal importance sampling, a provably ideal technique for this purpose \citep{Asmussen:07}.} 

Our key technical contribution lies in establishing theoretical performance guarantees for the BLISS estimator. \edit{In particular, we analyze the proposed method along a sequence of networks with increasing size and system-wide stability. We prove that BLISS is asymptotically optimal for Gaussian log-returns and achieves bounded relative error when the default probabilities in favorable and conservative scenarios remain comparable. We also quantify the asymptotic second-moment improvement contributed by the outer layer. Furthermore, the construction extends to normal mean-variance mixtures with a common generalized inverse Gaussian mixing variable. In this setting, the outer layer jointly samples the mixing variable and the non-target Gaussian factors conditional on the conservative event. 
The asymptotic analysis distinguishes semi-heavy left tails from regularly varying left tails and identifies the corresponding conditions for asymptotic optimality, bounded relative error, and outer-layer efficiency gains. This also explains why outer-layer importance sampling can remain useful with diagonal loadings, where a common mixing variable induces dependence among banks' returns.}

\edit{Notably, our methodology is sufficiently flexible to accommodate more general frameworks that incorporate market frictions alongside the two contagion channels. In our numerical study, we demonstrate its robustness and applicability when accounting for bankruptcy costs, the priority structure for external liabilities, as well as deterministic and random price impacts.}

\begin{revision}
The remainder of the paper is organized as follows. Section~\ref{sec:literature} reviews the related literature. Section~\ref{sec:problem} introduces the modeling framework and formulates the bond valuation problem. Section~\ref{sec:BLISS} develops the network-decoupling argument and the Gaussian BLISS algorithm. Section~\ref{sec:analysis} establishes its asymptotic efficiency properties. Section~\ref{sec:GH_extension} extends the construction and analysis to generalized hyperbolic log-returns. Section~\ref{sec:numerics} validates our theoretical results through synthetic experiments with increasing network sizes, and conducts robustness checks using a model calibrated to data from the European Banking Authority (EBA). We conclude in  Section~\ref{sec:conclusion}. Proofs are deferred to Appendix~\ref{app:proofs}.
\end{revision}

\section{Literature Review}\label{sec:literature}

The valuation of interbank claims in the presence of network contagion poses significant computational challenges. While theoretical frameworks for clearing payments and default cascades are well-developed, translating these models into tractable valuation tools requires addressing the curse of dimensionality that arises from network interactions and non-linear contagion channels. This section reviews three related strands of literature: network-based financial contagion models, existing approaches to debt valuation in networks that highlight computational limitations, and simulation methodologies for related problems that inform our proposed approach.

Network-based financial contagion has been extensively studied, beginning with the seminal framework of \cite{Eisenberg:01}, which characterizes equilibrium clearing payments through a fixed-point formulation that accounts for cascading defaults arising from failures in fulfilling interbank liabilities. Subsequent research has extended this framework along several dimensions.
Contagion channels have been enriched to include asset fire sales \citep{Cifuentes:05}, bankruptcy costs \citep[][]{Rogers:13,Glasserman:15}, and cross-holdings \citep{WeberWeske2017}. Other studies have incorporated institutional features such as central clearing counterparties \citep{Amini:16,Ahn:19}, multiple maturities \citep{KusnetsovMariaVeraart2019},  collateral requirements with contract termination \citep{GhamamiEtAl2022}, and contagious bank runs \citep{veraart2025}. The model we employ builds upon the framework of \cite{Cifuentes:05} to capture both liability and fire-sale contagion channels, which we extend to include external liabilities and stochastic shocks to external assets.

Building upon these models, other studies have investigated computational and policy questions that motivate our methodological contribution.
\cite{Elsinger:06} develop an empirical approach to assess systemic risk, \cite{Capponi:16} study the impact of the concentration of interbank liabilities on systemic losses, and \cite{Chen:16} analyze the interplay between the two contagion channels of asset liquidation and interbank liabilities. Other studies have characterized the sensitivity of clearing payments \citep{Liu:10,Feinstein:18} to non-interbank assets and network structure, intervention policies \citep{CapponiChen2015,AhnK:19,bernard2022}, systemic stress scenario selection \citep{AhnEtAl2023}, computation of systemic risk measures~\citep{ararat2023}, and worst-case network effects under partial information \citep{Ahn:23,Hu:24}. We refer to \cite{Capponi2016} and  \cite{Glasserman:16} for comprehensive surveys.

While the literature on clearing payments focuses on the ex-post consequences of shocks, computing the ex-ante value of interbank claims requires integrating over all possible realizations of external shocks. This problem is fundamentally constrained by the curse of dimensionality. The methodology of  \cite{gourieroux2012} decomposes the space of external shocks into $2^n$ default scenarios for a network of $n$ institutions and then evaluates expectations conditional on each scenario. However, this process is computationally impractical for large networks, meaning explicit analytical solutions are often restricted to highly simplified cases, such as those with few firms or no interbank liabilities. 
\cite{Barucca:20} develop a general network valuation framework that synthesizes both clearing and ex-ante valuation models, while also accounting for uncertainty in external assets.
However, analytical solutions remain elusive for complex, realistic systems, specifically when we take into account non-linear contagion channels such as asset liquidation and bankruptcy costs.

\edit{\cite{CossinSchellhorn2007} provide an analytical valuation framework for corporate debt in a network economy, in which a queueing-network model features strategic default and asymptotic steady-state price formulas require large cash reserves. By contrast, we model defaults that occur when a bank's asset value falls below its liabilities and provide an exact simulation framework that relies on neither steady-state nor large-cash assumptions.}


Despite its practical significance, the development of efficient simulation methods for financial networks remains an underexplored area. The few existing studies have focused on estimating firm-specific default probabilities via conditional Monte Carlo \citep{Ahn:18}, mixed importance sampling \citep{AhnZ2021}, and conditional importance sampling \citep{AhnZheng:23,AhnZ-WSC:23}. A similar problem has been explored in the context of distribution networks, where \cite{Blanchet:19} introduce importance sampling and conditional Monte Carlo estimators to compute systemic failure probabilities.\footnote{We refer the reader to \cite{Gandy:17} and \cite{Glasserman2023} for different approaches to sampling constrained networks and their application in financial networks.}
\edit{These approaches face challenges~\ref{item:curse of dimensionality} and~\ref{item:sensitivity} discussed in Section~\ref{sec:intro}: the first three methods rely on the exponential partitioning of the external asset space, in line with \cite{gourieroux2012}, whereas the latter solves a linear program in each iteration, and all of them require a target set described by linear constraints, which excludes fire sales.}
\edit{Moreover, all the aforementioned schemes except \cite{AhnZ-WSC:23} require elliptically distributed shocks to external assets, making them incompatible with standard log-normally distributed returns. While the method in \cite{AhnZ-WSC:23} accommodates such shocks, its applicability is limited to linear clearing structures, as discussed above.}
\edit{Our approach instead uses a network-derived conditional default threshold to construct a two-layer estimator. The resulting sampler avoids enumeration of default configurations and applies directly to the nonlinear clearing structure.} 
Our bi-level sampling approach is conceptually analogous to rare-event simulation techniques developed in the portfolio credit risk literature.
To be more specific, \cite{GlassermanL:05}, \cite{GlassermanKS:08}, \cite{kang2005}, and
\cite{Bassamboo:08} propose two-stage importance sampling methods that separate the sampling of common risk factors in a credit portfolio from that of idiosyncratic risk factors. This mirrors our approach in that we also split random variables into two groups and apply different importance sampling schemes to each. The key distinction, however, lies in how the target rare event and the corresponding sampling distributions are characterized. In our setting, these characterizations are driven by network contagion effects---a factor absent from the credit portfolio models discussed in the aforementioned papers. In those models, an obligor’s default does not directly affect another obligor’s financial status; instead, their dependence is captured through elliptical copulas. Hence, while the underlying two-stage concept is similar, the methods from the portfolio credit risk literature are not applicable to the problem of bond valuation in financial networks. 

\section{Model and Problem Formulation}
\label{sec:problem}
\subsection{Basic Framework}\label{subsec:model}

We build on the framework of \cite{Cifuentes:05}, extending it to   incorporate liabilities outside the financial network and stochastic external liquid asset values. We consider a network of $n$ banks indexed by $1,\ldots,n$. Each bank $i$'s initial balance sheet consists of the following items:\vspace{-8pt}
\begin{itemize}\itemsep=0pt
    \item $\bar p_{ij}\geq0$: the nominal liabilities (or payment obligation) to bank $j$, with $\bar p_{ii} = 0$;
    \item $\bar p_{i0}>0$: the nominal liabilities (or payment obligation) to external entities;
  \item $S_i^0>0$: the initial value of external liquid assets;
\item $e_i > 0$: the units of illiquid assets with a nominal unit price of $\bar q$.
\end{itemize}
Table~\ref{tab:init-bs} provides a decomposition of the balance sheet into assets and liabilities. Given this setup, bank $i$'s total liabilities are denoted by $\bar p_i = \bar p_{i0}+\sum_{j=1}^n \bar p_{ij}$, and its initial net worth is given by $w_i = S_i^0 + \sum_{j=1}^n\bar p_{ji} + e_i \bar q - \bar p_i$. We denote the total illiquid asset holdings across all banks in the system by $\bar e \coloneqq \sum_{i=1}^n e_i$.    In line with previous studies \citep[e.g.,][]{Cifuentes:05,Chen:16,bernard2022}, all banks are assumed to hold the same type of illiquid assets. This single representative asset can be viewed as an approximation of multiple correlated assets \citep{WeberWeske2017}.

\begin{table}[t]
\centering
\begin{tabular}{ccc}
    \toprule
    \textbf{Assets} & & \textbf{Liabilities} \\
    \midrule
    External Liquid Assets $S_i^0$ & & External Liabilities $\bar p_{i0}$ \\
    Interbank Assets $\bar p_{ji},\, j \neq i$ & & Interbank Liabilities $\bar p_{ij},\, j \neq i$ \\
    Illiquid Assets $e_i$ with price $\bar q$ & & \\
    \bottomrule
\end{tabular}
\caption{Bank $i$'s initial balance sheet}
\label{tab:init-bs}
\end{table}

Assuming that all liabilities share the same maturity, we model the value of bank $i$'s external liquid assets at maturity as
\begin{equation}\label{eq:external assets}
     S_i =  S_i^0\exp\lt(-\frac{\sigma_i^2}{2}+\sum_{k=1}^i\Lambda_{ik}Z_k\rt),
    \end{equation}
    where $(Z_1,\ldots,Z_n)\sim \cN({\bf 0}_n,\bI_n)$ \edit{and $\bLambda=(\Lambda_{ij})$ is the lower-triangular Cholesky factor of the log-return covariance matrix $\bLambda\bLambda^\top$, with $\Lambda_{ii}>0$ for all $i$. We write $\sigma_i^2:=\sum_{k=1}^{i}\Lambda_{ik}^2$. Thus, the external liquid assets are jointly lognormal and satisfy $\sE[S_i]=S_i^0$.\footnote{\edit{External liquid assets represent an aggregate portfolio that may include cash and other liquid instruments. The volatility is a proxy for the risk of this portfolio, not evidence that all its constituents are equities or that its log-returns are Gaussian.}} The risk-free rate is assumed to be zero without loss of generality. In Section~\ref{sec:GH_extension}, we replace this Gaussian specification with a normal mean-variance mixture. Therein,  $S_i^0$ remains the deterministic scale multiplying the exponential return and is not generally its mean.

    \begin{remark}
    The lower-triangular representation in \eqref{eq:external assets} facilitates the outer/inner decomposition, which will be used in our simulation methodology. We relabel the institution of interest as bank $n$ before factorizing the covariance matrix. This representation ensures that the non-target asset values depend only on ${\bf Z}_{-n}$, while the independent standard normal factor $Z_n$ affects only the target's return. Although such a separation can also be achieved via a conditional Gaussian representation or a suitable orthogonal transformation of the latent Gaussian coordinates, we adopt the lower-triangular representation for convenience and clarity. 
    \end{remark}}

    We next assume that all liabilities are of equal seniority; that is, when a bank's asset value at maturity falls short of the total liabilities, its debts are repaid to the creditors proportionally to the payment obligations. To illustrate this, we introduce the relative liability matrix $\bPi=(\pi_{ij})\in\bbR^{n\times n}_+$, where $\pi_{ij} \coloneqq (\bar p_{ij}/\bar p_i) {\bf 1}_{\{\bar p_i > 0\}}$ represents the proportion of bank~$i$'s liabilities to bank~$j$. Then, if $p_j$ represents bank $j$'s total debt payment for each $j=1,\ldots,n$, bank $i$'s interbank assets, i.e., the total payment received by bank $i$, become $\sum_{j=1}^n\pi_{ji}p_j$.

    For each realization $\bs$ of the external liquid asset vector $\bS$, banks with insufficient liquid assets need to sell off illiquid assets to fulfill their payment obligations. Specifically, banks first use their external liquid assets $s_i$ and interbank assets $(\sum_j \pi_{ji} p_j)$ to meet the payment obligation. When these assets are not sufficient (i.e., $s_i + \sum_j \pi_{ji} p_j<\bar p_i$) and the market price of illiquid assets at the maturity is given by $q$, bank $i$ either sells $(\bar p_i -  s_i - \sum_j \pi_{ji} p_j)/q$ units of illiquid assets if the residual asset value is enough to fulfill the obligation (i.e., $s_i +qe_i+  \sum_j \pi_{ji} p_j\geq\bar p_i$), or liquidates all its illiquid assets to distribute the residual assets to its creditors via the pro rata rule otherwise.

The price $q$ of illiquid assets is determined by the \textit{inverse demand function} $Q:[0, \bar e] \rightarrow \bbR_+$, which maps the total amount of illiquid assets in the market to the price.
Consequently, the equilibrium price $q(\bs)$ and the clearing payments $\{p_i(\bs)\}_{i=1}^n$ are characterized by a solution to the following fixed-point equations:
\begin{equation}\label{eq:market-clearing}
            \left\{~\begin{aligned}
                &q = Q\left(\sum_{i=1}^{n}\lt\{\frac{(\bar p_i -  s_i- \sum_j \pi_{ji} p_j)^+}{q} \wedge e_i \rt\}\right),		\\
                &p_i = \bar p_i \wedge \left(s_i+ q e_i +  \sum_{j=1}^n \pi_{ji} p_j\right),~i=1,\ldots,n,
            \end{aligned}\right.
        \end{equation}
where $x^+\coloneqq\max\{x,0\}$ and $x\wedge y\coloneqq \min\{x,y\}$ for any $x,y\in\bbR$. It follows from Theorem~2 of \cite{amini2016uniqueness} that the solution to the above equations is uniquely determined under some mild assumptions on the inverse demand function $Q$ listed below.
\begin{assumption}\label{ass:inverse_demand} We make the following assumptions:
\begin{enumerate}[label=(\alph*)]
    \item $\bar q=Q(0) > Q(\bar e) > 0$;
    \item $Q(x)$ is continuous and strictly decreasing in $x \in [0, \bar e]$;\label{as:Q-decreasing}
    \item $x Q(x)$ is strictly increasing in $x \in [0, \bar e]$.
\end{enumerate}
\end{assumption}
Given the fixed-point characterization in \eqref{eq:market-clearing}, we say that bank $i$ is solvent if $p_i(\bs) = \bar p_i$ and that it defaults otherwise. To rule out trivial cases where some banks are initially insolvent, we additionally assume that $p_i(\bS^0)=\bar p_i$ for all $i=1,\ldots,n$.
\subsection{Problem Specification}
Under the modeling framework in Section~\ref{subsec:model}, 
we aim to calculate the price of bank $n$'s bond with a face value of one, i.e.,
\begin{align}
\hat{p}_n\coloneqq \sE\bigg[\frac{p_n(\bS)}{\bar p_n}\bigg],
\end{align}
and this quantity can be decomposed as follows:
\begin{align}
\hat p_n= 1 -
\sP(p_n(\bS) <\bar p_n)+\sE\bigg[\frac{p_n(\bS)}{\bar p_n}\sone_{\{p_n(\bS) <\bar p_n\}}\bigg],\label{eq:decomp}
\end{align}
where $\{p_n(\bS)<\bar p_n\}$ denotes the event of bank $n$'s default and $\sone_A$ equals one if $A$ occurs and zero otherwise.
Based on the above decomposition of the bond price, we develop an efficient Monte Carlo method to estimate the expectation
\begin{equation}\label{eq:main}
\edit{\gamma_n}\coloneqq\sE\big[h(\bS)\sone_{\{p_n(\bS) <\bar p_n\}}\big]
\end{equation}
for any function $h(\cdot)$ uniformly positive and bounded on $\bbR_+^n$. This is because the second and last terms in \eqref{eq:decomp} are equal to the expectation~\eqref{eq:main} when $h(\cdot)\equiv 1$ and $h(\cdot) = p_n(\cdot)/\bar p_n$, respectively. Furthermore, as the random variables associated with these two terms are strongly correlated, an efficient Monte Carlo method for \eqref{eq:main} can be paired with common random numbers \citep{Glasserman2003-MCFE}. This approach effectively estimates both terms while reducing the overall variance in the calculation of the bond price $\hat p_n$.

    Estimating expectations of the form $\sE\big[g(\bX)\sone_{\{\bX\in A\}}\big]$, where $g$ is a function and $\bX$ is a random vector, has long been a central problem in the Monte Carlo simulation literature. Our objective \eqref{eq:main} falls within this class of problems. Designing an efficient Monte Carlo method for such an objective typically requires a careful geometric understanding of the target set $A$, which helps us generate more samples from this set to reduce variance.
In our setting, however, this poses a major challenge: the explicit representation of the target set is notoriously difficult due to the implicit fixed-point characterization of the equilibrium price and clearing payments in \eqref{eq:market-clearing}. To address this difficulty, \edit{prior Monte Carlo research has focused on financial networks without illiquid assets and developed simulation methods that exploit a linear programming reformulation of \eqref{eq:market-clearing} with $e_1=\cdots=e_n=0$ \citep{Ahn:18,AhnZ2021,AhnZheng:23}. As discussed in Section~\ref{sec:intro}, these approaches do not extend to our framework with illiquid assets.}

\section{Network Decoupling and Bi-Level Importance Sampling}\label{sec:BLISS}
We address the above-mentioned issue in our general setup by identifying a novel representation of the default event $\{p_n(\bs)<\bar p_n\}$. This finding enables us to develop our proposed algorithm for estimating $\edit{\gamma_n}$ in~\eqref{eq:main}, referred to as Bi-Level Importance Sampling with Splitting (BLISS).

\subsection{Reformulation of Default Events}\label{sec:prelim}

In this subsection, we focus on analyzing bank~$n$'s default event $\{p_n(\bs)<\bar p_n\}$ for any fixed $\bs\in\bbR_+^n$. We first note that a bank's solvency is fundamentally dependent on its external liquid assets. Therefore, to exclude the trivial possibility that a bank could remain solvent even with a complete loss of these assets, we impose the following condition, which is typically satisfied in practice. As will be shown in Section~\ref{subsec:numeric_2}, it also holds in the real-world example based on actual data.
\begin{assumption}\label{ass:balance}
For each bank $i=1,\ldots,n$, the value of its total nominal liabilities exceeds the combined nominal value of its illiquid and interbank assets by at least one monetary unit, i.e.,
$\bar p_i\geq1+\bar q e_i+\sum_{j\neq i} \pi_{ji} \bar p_j$.
\end{assumption}

The main challenge in simulating default events in networks is that a
bank's solvency depends implicitly on the entire system's fixed-point equations in \eqref{eq:market-clearing}. {This implicit characterization prevents direct
application of standard rare-event simulation techniques. We overcome this by decoupling the target bank from the network through a fictitious
system\edit{.}}
Using the above assumption, we establish a tractable representation
by considering a \emph{fictitious} system in which bank $n$ sells all its illiquid assets and satisfies its payment obligation regardless of its financial status. 
In this system, for each realization $\bs$ of $\bS$, the equilibrium asset price $\tilde q(\bs_{-n})$ and the clearing payment of banks $1$ through $n-1$, denoted by $\tilde p_1(\bs_{-n}),\ldots,\tilde p_{n-1}(\bs_{-n})$, are the solution to the following fixed-point equations:
\begin{equation}\label{eq:market-clearing2}
            \left\{~\begin{aligned}
                &q = Q\left(e_n+\sum_{i=1}^{n-1}\lt\{\frac{(\bar p_i -  s_i- \pi_{ni} \bar p_n-\sum_{j=1}^{n-1} \pi_{ji} p_j)^+}{q} \wedge e_i \rt\}\right),		\\
                &p_i = \bar p_i \wedge \left(s_i+ q e_i +\pi_{ni} \bar p_n+  \sum_{j=1}^{n-1} \pi_{ji} p_j\right),~i=1,\ldots,n-1.
                %
            \end{aligned}\right.
        \end{equation}
It is easy to check that the solution to these equations is also unique based on Assumption~\ref{ass:inverse_demand}. 
Furthermore, we define
\begin{equation}\label{eq:vn}
v_n(\bs_{-n})\coloneqq \bar p_n-\tilde q(\bs_{-n})e_n- \sum_{j=1}^{n-1} \pi_{jn} \tilde p_j(\bs_{-n}),
\end{equation}
where the second and third terms represent the values of bank $n$'s illiquid assets and interbank assets, respectively, in this fictitious system.
{The quantity $v_n(\bs_{-n})$ in equation~\eqref{eq:vn} represents the minimum external liquid asset value that bank $n$ requires to remain solvent, conditional on the external asset realizations $\bs_{-n}$ of all other banks. In the fictitious system, bank $n$'s net worth is therefore $s_n - v_n(\bs_{-n})$.\footnote{\edit{The threshold formulation easily accommodates a deterministic liquidity buffer. If bank $n$'s external liquid asset value is $c_n + S_n^{\mathtt{risky}}$, where $c_n \ge 0$ denotes its deterministic buffer, the formulation becomes $S_n^{\mathtt{risky}} < v_n(\bs_{-n}) - c_n$. Similarly, the deterministic buffers of other banks enter their external assets in the clearing system.}} The following proposition establishes that this threshold fully characterizes bank $n$'s default condition in the original system. Specifically, it shows that the $n$-dimensional implicit default event, defined through the fixed-point equations in~\eqref{eq:market-clearing}, can be equivalently represented as a simple one-dimensional threshold condition on bank $n$'s external assets.} 

\begin{proposition}[Default condition reformulation]\label{prop:fictitious}
Suppose that Assumptions~\ref{ass:inverse_demand} and~\ref{ass:balance} hold. Then, bank $n$ defaults in the original system if and only if it has a negative net worth in the above-mentioned fictitious system, i.e., $\{\bs\in\bbR_+^n:p_n(\bs)<\bar p_n\}=\{\bs\in\bbR_+^n:s_n<v_n(\bs_{-n})\}$.
\end{proposition}

\edit{Proposition~\ref{prop:fictitious}} is a nontrivial extension of Lemma~1 in \cite{Ahn:23} under the Eisenberg-Noe framework without fire sales to our generalized modeling framework in Section~\ref{subsec:model}. While their result is derived by using a linear program reformulation of the clearing payment characterization and its duality, such a tool is not applicable in our setup due to the nonlinearity introduced by asset liquidation in the clearing payment characterization; see the first equation in~\eqref{eq:market-clearing}. Instead, we achieve this generalization by decoupling bank $n$'s external assets from the fixed-point equations in~\eqref{eq:market-clearing} through the aforementioned fictitious system approach.
Without this approach, deriving the critical threshold $v_n(\bs_{-n})$ would be challenging because $s_n$ is inherently coupled with the fixed-point equations in~\eqref{eq:market-clearing}; consequently, one must resort to numerically solving the system of equations to determine the threshold. See Figure~\ref{fig:default set} for a description of a default event in a 2-bank system based on the result in Proposition~\ref{prop:fictitious}.

\begin{figure}[!t]
\centering
\begin{tikzpicture}[]
\begin{axis}[ticks=none,height=3in,width=5in,ymin=-1,ymax=5,xmax=5,xmin=-2/3,xticklabel=\empty,yticklabel=\empty,axis lines = middle,xlabel=$ s_1$,ylabel=$ s_2$]

\addplot[fill=black, fill opacity=0.1, draw=none, area legend] coordinates
{(5,0)  (5,4/3)  (4/3,4/3)  (0,4)   (0,0)}
	\closedcycle;

\coordinate (A) at (axis cs:4/3,4/3);
\coordinate (A2) at (axis cs:3,-2);
\coordinate (B) at (axis cs:0,4);
\coordinate (B2) at (axis cs:-1,6);
\draw[-,thick] (B) to (A);
\draw[dashed] (A) to (A2);
\draw[dashed] (B) to (B2);
\coordinate (C) at (axis cs:5,4/3);
\coordinate (C2) at (axis cs:-1,6);
\coordinate (D) at (axis cs:0,4/3);
\coordinate (D2) at (axis cs:-1,4/3);
\draw[-,thick] (A) to (C);
\draw[dashed] (A) to (D2);

\coordinate (E) at (axis cs:4/3,1.6);
\node[anchor=west,align=center] at (E.east) {$(\bar p_1-\pi_{21}\bar p_2,\bar p_2-\pi_{12}\bar p_1)$};
\node[anchor=west,align=center] at (B.east) {$(0,\bar p_2-\pi_{12}\pi_{21}\bar p_2)$};
\addplot+[black, mark=*, mark size=1.5, every mark/.append style={solid, fill=black}] coordinates {(0,4)};
\addplot+[black, mark=*, mark size=1.5, every mark/.append style={solid, fill=black}] coordinates {(4/3,4/3)};

\node[anchor=south,align=center] at (axis cs:2.5,0.4) {$ s_2<v_2(s_1)=\bar p_2-\pi_{12}\min\{\bar p_1,s_1+\pi_{21}\bar p_2\}$};
\end{axis}
\end{tikzpicture}
\caption{A graphical illustration of the default event $\{\bs\in\bbR^2_+:p_2(\bs)<\bar p_2\}$ in a 2-bank system without fire sales. The shaded area represents the set of liquid asset values that cause bank 2 to default. The thick solid line depicts the value of $v_2(s_1)$ for different $s_1$.\label{fig:default set}}
\end{figure}

\subsection{BLISS: Introduction}\label{subsec:bliss_new}
{Using the explicit default characterization from Proposition~\ref{prop:fictitious}, we now construct our simulation methodology. We leverage the fact that the  
threshold $v_n(\bS_{-n})$ suggests a natural splitting. First, at the outer level, we sample the 
external assets of banks 1 to $n-1$. Then, at the inner level, we conditionally 
sample bank $n$'s external assets.} 
In the special case of $h\equiv 1$, it can be easily checked that such a splitting (or conditional Monte Carlo) approach results in smaller estimation errors than the standard Monte Carlo method:
\begin{align}
    {\sf Var}\big(\sone_{\{p_n(\bS) <\bar p_n\}}\big)&=
    {\sf Var}\big(\sone_{\{S_n <v_n(\bS_{-n})\}}\big)\\
    &= {\sf Var}\Big(\sE\big[\sone_{\{S_n <v_n(\bS_{-n})\}}|\bS_{-n}\big]\Big)
    + \sE\Big[{\sf Var}\big(\sone_{\{S_n <v_n(\bS_{-n})\}}|\bS_{-n}\big)\Big]\\
    &> {\sf Var}\Big(\sE\big[\sone_{\{S_n <v_n(\bS_{-n})\}}|\bS_{-n}\big]\Big)\\
    &= {\sf Var}\big(\sP(S_n <v_n(\bS_{-n})|\bS_{-n})\big),\label{eq:var-cmc}
\end{align}
where the strict inequality holds because ${\sf Var}\big(\sone_{\{S_n <v_n(\bS_{-n})\}}|\bS_{-n}=\bx\big)>0$ for all $\bx\in\bbR_+^{n-1}$.

To generalize this splitting approach for more flexible choices of $h$ and achieve greater variance reduction, we propose a new algorithm, which we call \emph{Bi-Level Importance Sampling with Splitting (BLISS)}, for estimating an equivalent (split) version of \eqref{eq:main}:
\begin{equation}\label{eq:main-recast}
\edit{\gamma_n}=\sE_{\bS_{-n}}\Big[\sE_{S_n}\big[h(\bS)\sone_{\{S_n<v_n(\bS_{-n})\}}\big|\bS_{-n}\big]\Big].
\end{equation}
This algorithm applies two types of importance sampling methods to the outer-layer sampling of $\bS_{-n}$ and the inner-layer sampling of $S_n$ in the splitting framework, from which its name is derived. 

\begin{revision}
\textbf{Outer-layer sampling: Conservative conditional sampling.}
Based on the definition of $\{S_i\}_{i=1}^n$ in \eqref{eq:external assets}, we first generate a sample of $\bZ_{-n}=(Z_1,\ldots,Z_{n-1})^\top$ using the \emph{conservative conditional sampling}, introduced below.  Let
$
M_n\coloneqq S_n^0\exp(-{\sigma_n^2}/{2})$ and $
v_n^\tc\coloneqq v_n(\bzero_{n-1})$
represent the median of bank $n$'s external liquid asset value and the critical threshold $v_n(\cdot)$ in the conservative scenario when all non-target banks' external assets vanish (i.e., $\bs_{-n}=\bzero_{n-1}$), respectively. 
Define the associated conservative log-buffer as
\begin{equation}\label{eq:conservative log buffer}
B_n^\tc\coloneqq\log\frac{M_n}{v_n^\tc}.
\end{equation}
Since $v_n(\bs_{-n})\le v_n^\tc$ for any $\bs_{-n}\in\bbR_+^{n-1}$ by the monotonicity established in the proof of Proposition~\ref{prop:fictitious}, bank~$n$ remains solvent if $B_n^\tc+\sum_{i=1}^n\Lambda_{ni}Z_i\geq 0$, or equivalently $S_n\geq v_n^\tc$. Hence, letting $\blambda_n\coloneqq(\Lambda_{n1},\ldots,\Lambda_{n(n-1)})^\top$, we propose to generate samples of $\bZ_{-n}$ conditional on the conservative half-space
\begin{equation}\label{eq:gaussian_conservative_halfspace}
\{\blambda_n^\top\bZ_{-n}+\Lambda_{nn}Z_n<-B_n^\tc\},
\end{equation}
so that most generated samples fall into the default set. By Bayes' rule, the density for this new sampling distribution and the corresponding outer-layer likelihood ratio are given by
\begin{equation}\label{eq:new_outer_density}
 \varphi_{\tt out}(\bz)
 \coloneqq
 \frac{\phi_{n-1}(\bz)\kappa_n^\tc(\bz)}{\kappa_n^\tc}
 \quad \text{and}\quad \cL_{\tt out}(\bz)=\frac{\kappa_n^\tc}{\kappa_n^\tc(\bz)}\quad\text{for}~\bz\in\bbR^{n-1},
\end{equation}
respectively, where $\phi_d$ denotes the $d$-dimensional standard normal density, 
\begin{align}
\kappa_n^\tc(\bz)&\coloneqq\sP(S_n<v_n^\tc\,|\,\bZ_{-n}=\bz)=\bar\Phi\!\left(\frac{B_n^\tc+\blambda_n^\top\bz}{\Lambda_{nn}}\right)~~\text{for}~\bz\in\bbR^{n-1},\label{eq:gauss_conservative_prob}
\\
\kappa_n^\tc&\coloneqq\sP(S_n<v_n^\tc)=\sE[\kappa_n^\tc(\bZ_{-n})]=\bar\Phi\!\left(\frac{B_n^\tc}{\sigma_n}\right),\label{eq:gauss_conservative_unconditional}
\end{align}
and $\bar\Phi$ is the complementary cumulative distribution function of a standard normal random variable.

The following result characterizes the exact sampler for the conservative conditional sampling:
\begin{proposition}[Exact sampler for the outer-layer  sampling distribution]\label{prop:new_outer_sampler}
The following procedure generates $\bZ_{-n}$ from the density in~\eqref{eq:new_outer_density}:
\begin{enumerate}[label=(\roman*)]
\item draw $Y$ from $\cN(0,1)$ truncated to $(-\infty,-B_n^\tc/\sigma_n)$ and independently draw $\bxi$ from $\cN(\bzero_{n},\bI_{n})$;
\item return
\begin{equation}\label{eq:outer_rank_one}
 \bZ_{-n}
 \coloneqq
 \bxi_{-n}+\lt(\frac{Y}{\sigma_n}-\frac{\sum_{i=1}^n\Lambda_{ni}\xi_i}{\sigma_n^2}\rt)\blambda_n.
\end{equation}
\end{enumerate}
\end{proposition}
\end{revision}

\edit{\textbf{Inner-layer sampling: Near-optimal importance sampling.} Given a realization $\bz$ of $\bZ_{-n}$ and the associated external-asset vector $\bS_{-n}(\bz)$ for non-target banks generated from~\eqref{eq:external assets}, we} 
sample $Z_n$ using \emph{near-optimal importance sampling} to construct a sample of $S_n$. To illustrate this, we define 
\edit{\begin{equation}\label{eq:lnzn}
\ell_n(\bz)\coloneqq\frac{B_n(\bz)+\blambda_n^\top\bz}{\Lambda_{nn}},
\end{equation}
where $B_n(\bz)$ represents the realized log-buffer for bank~$n$ given by
$
B_n(\bz)\coloneqq\log({M_n}/{v_n(\bS_{-n}(\bz))}).
$}
Then, by \eqref{eq:external assets}, it is straightforward to see that the default event \edit{$\{S_n<v_n(\bS_{-n}(\bz))\}$} can be recast as $\{Z_n<-\ell_n(\bz)\}$. Furthermore, for any random variable $X$ and any given set $A$, the optimal importance sampling distribution for estimating $\sP(X\in A)$ is known as the conditional distribution of $X$ given $A$ \citep{Asmussen:07}. Hence, especially for the case of $h\equiv 1$ and when \edit{$\bZ_{-n}=\bz$}, the optimal importance sampling distribution of $Z_n$ for estimating the inner conditional expectation in \eqref{eq:main-recast} is characterized by the truncated standard normal density
\begin{equation}\label{eq:psi}
\edit{\varphi_{\tt in}}(x|\bz)\coloneqq\frac{\edit{\phi_1}(x)\sone_{\{x<-\ell_n(\bz)\}}}{\edit{\bar\Phi(\ell_n(\bz))}}\quad\text{for}~ x\in\bbR,
\end{equation}
where \edit{$\bar\Phi(\ell_n(\bz))$ denotes bank $n$'s default probability given $\bZ_{-n}=\bz$}. We use this sampling distribution for general cases of $h$, leading to the term \emph{near-optimal} importance sampling. The associated likelihood ratio is given by
\begin{equation}\label{eq:Lin}
\edit{\cL_{\tt in}(\bz)  = \bar\Phi(\ell_n(\bz))~~\text{for}~\bz\in\bbR^{n-1}.}
\end{equation}
Note that this likelihood ratio does not depend on the samples of $Z_n$ because the sampling distribution of $Z_n$ is defined on the interval $(-\infty,-\ell_n(\edit{\bz}))$, i.e., $\sone_{\{Z_n<-\ell_n(\edit{\bz})\}}=1$.

\edit{\textbf{BLISS estimator.}} Consequently, our proposed BLISS estimator is finally characterized as
\begin{equation}\label{eq:estimator}
\edit{\Gamma_n^*
 \coloneqq}~h\!\lt(S^*_1,\ldots,S^*_n\rt)\cL_{\tt in}\!\lt(\bZ_{-n}^*\rt)\cL_{\tt out}\!\lt(\bZ^*_{-n}\rt)\edit{=h\!\lt(S^*_1,\ldots,S^*_n\rt)\bar\Phi\!\lt(\ell_n(\bZ_{-n}^*)\rt)
 \frac{\kappa_n^\tc}{\kappa_n^\tc(\bZ_{-n}^*)},}
\end{equation}
where \edit{$\bZ^*_{-n}$ is generated by the exact sampler in Proposition~\ref{prop:new_outer_sampler}}, $Z^*_n=(Z_n|Z_n<-\ell_n(\bZ_{-n}^*))$, and $S^*_i = S_i^0\exp(-{\sigma_i^2}/{2}+\sum_{k=1}^i\Lambda_{ik}Z^*_k)$ for $i=1,\ldots,n$. Note that this estimator is unbiased\edit{, and the auxiliary Gaussian samples used in Proposition~\ref{prop:new_outer_sampler} are discarded after the outer-layer draw}. The complete description of our BLISS procedure is provided in Algorithm~\ref{alg:BLISS}.

\begin{revision}
\begin{algorithm}[t]
\caption{BLISS for Gaussian Log-returns}\label{alg:BLISS}
\edit{\begin{algorithmic}[1]
\State Compute $B_n^\tc$ and $\kappa_n^\tc$ from \eqref{eq:conservative log buffer} and \eqref{eq:gauss_conservative_unconditional}, respectively
\Procedure{outer-layer sampling}{}
\State Sample $Y$ from $\cN(0,1)$ truncated to $(-\infty,-B_n^\tc/\sigma_n)$
\State Sample $\bxi$ from $\cN(\bzero_{n},\bI_{n})$
\State {Reconstruct $\bZ_{-n}$ via}~\eqref{eq:outer_rank_one}
\State Set $S_i=S_i^0\exp(-\sigma_i^2/2+\sum_{k=1}^i\Lambda_{ik}Z_k)$ for $i=1,\ldots,n-1$
\State Compute $\cL_{\tt out}({\bZ_{-n}})$
\EndProcedure
\Procedure{inner-layer sampling}{${\bZ_{-n}}$}
\State Compute $\ell_n({\bZ_{-n}})$ from~\eqref{eq:lnzn}
\State Sample $Z_n$ from~\eqref{eq:psi}
\State Set $S_n=S_n^0\exp(-\sigma_n^2/2+\sum_{k=1}^n\Lambda_{nk}{Z_k})$
\State Compute $\cL_{\tt in}({\bZ_{-n}})$
\EndProcedure
\State Set $T=h(S_1,\ldots,S_n)\cL_{\tt in}({\bZ_{-n}})\cL_{\tt out}({\bZ_{-n}})$
\State \textbf{return} $T$
\end{algorithmic}}
\end{algorithm}

For each draw, BLISS generates the asset vector and computes the threshold $v_n(\bs_{-n})$ by solving the fictitious system, which takes at most $n-1$ rounds of the fictitious default algorithm~\citep{WeberWeske2017}. The cost per draw is therefore polynomial in $n$. By contrast, estimators that partition the shock space by default configuration must handle up to $2^n$ configurations, which is the source of the exponential growth in computational cost reported for existing methods (see Section~\ref{sec:intro} and Figure~\ref{fig:comparison}). Separately, the number of draws needed for a target relative error is governed by the rarity of default and the efficiency of the estimator, which Sections~\ref{sec:analysis} and~\ref{sec:GH_extension} analyze.
\end{revision}

\section{Analysis of the BLISS Estimator}\label{sec:analysis}

We provide theoretical performance guarantees for the BLISS algorithm introduced in the previous section. We mainly focus on situations where the default event is rare. 
In this case, 
the naive Monte Carlo estimator is significantly inefficient. For example, when $h\equiv 1$, the relative error of the naive Monte Carlo estimator of \edit{$\gamma_n$ in \eqref{eq:main} is $\sqrt{(1-\gamma_n)/(N\gamma_n)}$, where $N$ denotes the sample size. This implies that the relative error explodes as $\gamma_n\to0$.}
To mathematically evaluate the efficiency of the proposed method, we analyze its performance in \edit{an asymptotic regime governed by the network size $n$, in which the default probability decreases to $0$ as $n$ tends to $\infty$.} 

\edit{Specifically, we consider a sequence of networks
\begin{equation}\label{eq:network_sequence}
 \mathfrak{N}_n
 =\big(\bPi_n,\bar\bp_n,\be_n,h_n,Q_n,\bS^0_n,\bLambda_n\big),
 \qquad n=1,2,\ldots,
\end{equation}
where 
the $n$th network is of size $n$, every primitive may depend on $n$, and for fixed~$n$, $\bPi_n$, $\bar\bp_n$, $\be_n$, $h_n$, $Q_n$, $\bS^0_n$, and $\bLambda_n$ correspond to $\bPi$, $\bar\bp$, $\be,h,Q$, $\bS^{0}$, and $\bLambda$, respectively, used in the previous section. 
For each~$n$, the $n$-indexed primitives are assumed to satisfy Assumptions~\ref{ass:inverse_demand} and~\ref{ass:balance}.  Thus, Assumption~\ref{ass:inverse_demand} is imposed pointwise on $Q_n:[0,\bar e_n]\to\bbR_+$, where $\bar e_n=\sum_{i=1}^ne_{i,n}$, and Assumption~\ref{ass:balance} and $p_{i,n}(\bS_n^0)=\bar p_{i,n}$ hold for all banks in each finite network. 
The pointwise clearing conditions guarantee a unique clearing solution for each finite network; see \cite[Theorem~2]{amini2016uniqueness}. The uniform positivity and boundedness {of the payoff function stipulated in~\eqref{eq:main}} mean that fixed constants $0<h_L\le h_U<\infty$ exist so that $h_L\le h_n(\bs)\le h_U$ for every $n$ and every $\bs\in\bbR_+^n$. These are well-posedness conditions rather than rarity assumptions.

Let $v_{i,n}(\cdot)$ be the critical threshold in~\eqref{eq:vn} after relabeling bank $i$ as the target in the $n$th network $\mathfrak{N}_n$. Then, under the above conditions, it can be easily checked that
for any $i=1,\ldots,n$,
\begin{equation}\label{eq:threshold_sandwich}
 0<v_{i,n}^\f\le v_{i,n}(\cdot)\le v_{i,n}^\tc,
\end{equation}
where 
$v_{i,n}^\tc\coloneqq v_{i,n}(\bzero_{n-1})$ is the threshold under the conservative scenario discussed in Section~\ref{subsec:bliss_new}, and $v_{i,n}^\f\coloneqq
 \bar p_{i,n}-Q_n(e_{i,n})e_{i,n}
 -\sum_{j\neq i}\pi_{ji,n}\bar p_{j,n}$ denotes the threshold under the favorable scenario when all non-target banks hold sufficient external assets to satisfy their payment obligations without requiring liquidation. In what follows, the superscripts $\tc$ and $\f$ denote the conservative/stress and favorable/support benchmarks, respectively.

For each $i$, we let $D_{i,n}^\tc\coloneqq B_{i,n}^\tc/\sigma_{i,n}$, where $B_{i,n}^\tc$ is the conservative log-buffer obtained by applying~\eqref{eq:conservative log buffer} when bank $i$ is the target. Similarly, we define $D_{i,n}^\f$ by replacing $v_{i,n}^\tc$ in the construction of $D_{i,n}^\tc$ with $v_{i,n}^\f$. We call them the conservative and favorable distances-to-default, respectively, because the conservative and favorable benchmarks of the default probability are governed by these quantities; specifically, $\sP(S_i<v_{i,n}^\tc)=\bar\Phi(D_{i,n}^\tc)$ and $\sP(S_i<v_{i,n}^\f)=\bar\Phi(D_{i,n}^\f)$. Using these concepts, we construct the focal asymptotic regime for analyzing the BLISS estimator:
\begin{definition}[System-wide stability regime]\label{def:gaussian_stability_regime}
The network sequence $\{\mathfrak{N}_n\}$ is in the \emph{system-wide stability regime} if (i) $\min_{1\leq i\leq n}D_{i,n}^\tc\to\infty$ and (ii)  $\max_{1\leq i\leq n}|D_{i,n}^\f/D_{i,n}^\tc-1|\to0$ as $n$ grows.
\end{definition}
Under Assumption~\ref{ass:balance}, condition (i) implies that the maximum default probability among all banks vanishes as $n\to\infty$.
Accordingly, the network sequence represents a structural progression toward increasingly stable financial systems, rather than the mechanical addition of marginal banks to a static economy. 
Condition (ii) of the system-wide stability regime is both realistic and relatively mild: it indicates that shock propagation within the network plays a second-order role in default events, with first-order impacts driven by direct shocks to external liquid assets. Such second-order nature of the network effects is well established in the literature~\citep[see, e.g.,][]{Glasserman:15}.
A prototypical example of this regime lies in the case where the initial values of the external liquid assets, $\bS^0_n$, scale with $n$, while other balance sheet components, the functions $h_n,Q_n$, and the return volatility $\bsigma_n=\sqrt{\mathrm{diag}(\bLambda_n\bLambda_n^\top)}$ remain uniformly bounded away from zero and from above across~$n$.



Below we will show the asymptotic efficiency of the BLISS estimator in the system-wide stability regime. To this end, we introduce two standard notions of} efficiency in rare-event simulation, formally defined as follows:
\begin{definition}[Rare-event efficiency criteria]\label{def:efficiency}
Let \edit{$\Gamma_n$ be} an unbiased estimator of \edit{$\gamma_n$, with $\gamma_n\to0$}. Then, we say that the estimator \edit{$\Gamma_n$} is \emph{asymptotically optimal} if
$$
 \lim_{\edit{n}\to\infty}\frac{\log\sE[\edit{\Gamma_n^2}]}{2\log\edit{\gamma_n}}=1,
$$
\edit{and it has \emph{bounded relative error} if}
$$
\edit{\limsup_{n\to\infty}\frac{\sE[\Gamma_n^2]}{\gamma_n^2}<\infty.}
$$
\end{definition}
The asymptotic optimality of an estimator indicates that its $L^2$ norm and first moment share the same asymptotic rate of change, thereby achieving significant variance reduction for large \edit{$n$} compared to the naive Monte Carlo method. \edit{While challenging to attain, bounded relative error is considered the ideal property in simulation efficiency as it guarantees consistent accuracy with a fixed sample size, no matter how rare the target event becomes. Obviously, bounded relative error implies asymptotic optimality. See \cite{Asmussen:07} for more details on these efficiency notions}.

Given this setup, the asymptotic \edit{efficiency} of the BLISS estimator is demonstrated in the following theorem.
\begin{revision}
\begin{theorem}[Asymptotic efficiency of $\Gamma_n^*$]\label{thm:gaussian_optimality}
Under the system-wide stability regime in the sequence $\{\mathfrak{N}_n\}$ of networks, the BLISS estimator $\Gamma_n^*$ is asymptotically optimal.  Furthermore, if
\begin{equation}\label{eq:gaussian_tail_comparability}
 \limsup_{n\to\infty}
 \frac{\bar\Phi(D_{n,n}^\tc)}{\bar\Phi(D_{n,n}^\f)}<\infty,
\end{equation}
then $\Gamma_n^*$ has bounded relative error.
\end{theorem}
Recall that $\bar\Phi(D_{n,n}^\tc)$ and $\bar\Phi(D_{n,n}^\f)$ are bank~$n$'s default probabilities under the conservative and favorable scenarios, respectively, discussed earlier. Thus, \eqref{eq:gaussian_tail_comparability} strengthens condition~(ii) of Definition~\ref{def:gaussian_stability_regime} for the target bank, requiring these probabilities to remain asymptotically comparable  up to a finite multiplicative factor. Since the actual default probability lies between these two benchmarks, the conservative half-space probability used in the outer-layer sampling is consequently within a bounded factor of the true target default probability.
\end{revision}

Moreover, \edit{we quantify} the variance reduction achieved by using our \edit{outer-layer distribution in \eqref{eq:new_outer_density}}, relative to the approach in which the outer-layer importance sampling is not applied (i.e., sampling $\bZ_{-n}$ from $\cN(\bzero_{n-1},\bI_{n-1})$). 
\begin{revision}
To achieve this, let $\sE_n^*$ denote the expectation under the exact outer-layer distribution~\eqref{eq:new_outer_density} and inner-layer distribution~\eqref{eq:psi} for the $n$th network $\mathfrak{N}_n$, with every quantity interpreted using the $n$-indexed primitives.
Next, we formally define the inner-layer-only estimator 
 $\Gamma_n^0\coloneqq h_n(\bS)\bar\Phi(\ell_n(\bZ_{-n}))$ and the associated expectation $\sE_n^0$.



\begin{proposition}[Outer-layer efficiency gain in $\Gamma_n^*$]\label{prop:superiority}
{Let $\blambda_{n,n} = (\Lambda_{n1,n},\ldots,\Lambda_{n(n-1),n})^\top$.} Under the system-wide stability regime in the sequence $\{\mathfrak{N}_n\}$ of networks,  if $\lim_{n\to\infty}({\|\blambda_{n,n}\|^2}/{\sigma_{n,n}^2})$ exists, then we have
\begin{equation}\label{eq:gaussian_comparison_ratios}
\lim_{n\to\infty}\frac{\log\sE_n^*[(\Gamma_n^*)^2]}{\log\sE_n^0[(\Gamma_n^0)^2]}
 ={1+\lim_{n\to\infty}\frac{\|\blambda_{n,n}\|^2}{\sigma_{n,n}^2}.}
\end{equation}
\end{proposition}
\end{revision}
Proposition~\ref{prop:superiority} indicates that correlation among external assets magnifies the performance advantage of our proposed BLISS estimator, \edit{$\Gamma_n^*$}, over its counterpart without outer-layer importance sampling. Specifically, when assets are uncorrelated, the equality \edit{${\Lambda_{nn,n}^2}={\sigma_{n,n}^2}$} holds, and the second moments of both estimators decay at the same rate. In contrast, higher correlation decreases the ratio \edit{${\Lambda_{nn,n}^2}/{\sigma_{n,n}^2}$}, which causes the second moment of \edit{$\Gamma_n^*$} to decay exponentially faster.

\begin{revision}
\section{Extension to Generalized Hyperbolic Distributions}\label{sec:GH_extension}
This section retains the network sequence $\{\mathfrak{N}_n\}$, Assumptions~\ref{ass:inverse_demand} and~\ref{ass:balance}, clearing conditions, and deterministic thresholds $\{v_{i,n}^\tc,v_{i,n}^\f\}$ from Section~\ref{sec:analysis}, but replaces the Gaussian log-return by the normal mean-variance mixture with generalized inverse Gaussian (GIG) mixing variables.  This is the classical construction of \cite{BarndorffNielsenKentSorensen1982}, which yields the class of generalized hyperbolic distributions. This class encompasses Student's $t$, hyperbolic, normal-inverse Gaussian, normal-inverse gamma, variance-gamma, and Laplace distributions, and its use for flexible modeling of tail behavior and skewness in financial data has been extensively discussed in the related literature~\citep[see, e.g.,][]{McNeil:15,Cont01022001}.  

\subsection{ Preliminaries}\label{subsec:GH model}
It is well known that the proper GIG parameter set is given by
$$
 \cG\coloneqq
 \{(\lambda,\chi,\psi):\lambda\in\bbR,\ \chi>0,\ \psi>0\}
 \cup\{(\lambda,0,\psi):\lambda>0,\ \psi>0\}
 \cup\{(\lambda,\chi,0):\lambda<0,\ \chi>0\}.
$$
For $\chi,\psi>0$, the density of $V\sim\GIG(\lambda,\chi,\psi)$ is
\begin{equation}\label{eq:gig_density}
 f_V(v)
 =\frac{(\psi/\chi)^{\lambda/2}}{2K_\lambda(\sqrt{\chi\psi})}
 v^{\lambda-1}
 \exp\!\left[-\frac12\left(\frac{\chi}{v}+\psi v\right)\right],
 \qquad v>0,
\end{equation}
where $K_\lambda$ is the modified Bessel function of the third kind.  The cases $\chi=0$ and $\psi=0$ are the gamma and inverse-gamma boundaries, respectively. Given this setup, we assume a generalized hyperbolic log-return for external liquid assets as follows:

\begin{assumption}\label{ass:gh_model}
Fix $(\lambda,\chi,\psi)\in\cG$.  For each $n$, let $\bLambda_n$ be lower triangular with strictly positive diagonal, and let $\btheta_n\in\bbR^n$, 
$V\sim\GIG(\lambda,\chi,\psi)$ and
 $\bZ\sim\cN(\bzero_n,\bI_n)$ with $ V\perp\bZ$. Then, for the $n$th network $\mathfrak{N}_n$, bank $i$'s external liquid asset value is given by
\begin{equation}\label{eq:gh_latent_model}
 S_{i,n}\coloneqq S^0_{i,n}\exp(X_{i,n}),
\end{equation}
where $S^0_{i,n}>0$ is deterministic and $\bX_n\coloneqq V\btheta_n+\sqrt{V}\bLambda_n\bZ$. 
\end{assumption}

As discussed earlier, this assumption replaces the Gaussian log-return model in~\eqref{eq:external assets} with the generalized hyperbolic version.  The scale $S^0_{i,n}$ is not asserted to equal $\sE[S_{i,n}]$; this distinction is necessary on heavy-tailed GIG boundaries, where exponential moments may fail. 
Throughout this section, probabilities and expectations refer to the generalized hyperbolic distribution in Assumption~\ref{ass:gh_model}, and the system-wide stability regime in Definition~\ref{def:gaussian_stability_regime} is imposed under that distribution.
Specifically, we first let $\omega_{n,n}\coloneqq\sqrt{\sum_{i=1}^n\Lambda_{ni,n}^2}$ for the target bank~$n$ and then
transform the conservative and favorable thresholds $\{B_{n,n}^\tc,B_{n,n}^\f,D_{n,n}^\tc,D_{n,n}^\f\}$ in Section~\ref{sec:analysis} into
\begin{equation}\label{eq:gh_latent_thresholds}
 b_{n,n}^\tc\coloneqq\log\frac{S^0_{n,n}}{v_{n,n}^\tc},
 \quad
 b_{n,n}^\f\coloneqq\log\frac{S^0_{n,n}}{v_{n,n}^\f},
 \quad
 d_{n,n}^\tc\coloneqq\frac{b_{n,n}^\tc}{\omega_{n,n}},
 \quad\text{and}\quad
 d_{n,n}^\f\coloneqq\frac{b_{n,n}^\f}{\omega_{n,n}},
\end{equation}
respectively, and similarly define $\{d_{i,n}^\tc, d_{i,n}^\f\}$ for $i\neq n$ by relabeling bank $i$ as the target. 
For the target bank, the standardized return is given by
\begin{equation}\label{eq:gh_projection}
 T_n\coloneqq\frac{X_{n,n}}{\omega_{n,n}}
 =\eta_n V+\sqrt{V}R_n,
 ~~\text{where}~\eta_n\coloneqq\frac{\theta_{n,n}}{\omega_{n,n}}~~\text{and}~~R_n\coloneqq\frac{\sum_{i=1}^n\Lambda_{ni,n}Z_i}{\omega_{n,n}}\sim\cN(0,1).
\end{equation}
The parameter $\eta_n$ in~\eqref{eq:gh_projection} measures the sensitivity of the target bank’s standardized conditional mean to the common mixing variable $V$. Its sign determines the direction of the variance-dependent mean shift, while its magnitude measures the size of that shift relative to the target bank’s Gaussian shock scale. We impose the following regularity condition on this parameter for the asymptotic analysis: 
\begin{assumption}\label{ass:gh_fixed}
There exist $\eta\in\bbR$ and $N<\infty$ such that $\eta_n=\eta$ for any $n\ge N$.
\end{assumption}
Assumption~\ref{ass:gh_fixed} states that the above-mentioned standardized sensitivity to the common volatility state does not change along the network sequence $\{\mathfrak{N}_n\}$. 
Note that this assumption does not require the unstandardized coefficient $\theta_{n,n}$ or the scale $\omega_{n,n}$ to be separately constant. Also, it does not restrict the evolution of network topology, balance sheets, default thresholds, or factor scale. Its purpose is to ensure that rarity is driven by the divergence of the network-adjusted distance-to-default, rather than by an $n$-dependent change in the target shock distribution, and to permit the use of the tail asymptotics for the standardized marginal return distribution. 
Our finite-$n$ sampler uses the actual finite-$n$ parameters and ensures  unbiasedness without this assumption.

\subsection{Conditional GIG and the Proposed Sampling Algorithm}

\textbf{Outer-layer sampling.} To construct the exact outer-layer sampler for generalized hyperbolic cases, we extend the exact outer-layer sampler in Proposition~\ref{prop:new_outer_sampler} for the Gaussian case and propose to draw samples of $(V,\bZ_{-n})$ conditional on the conservative half-space
$$
\lt\{V\theta_{n,n}+\sqrt{V}\big(\blambda_{n,n}^\top\bZ_{-n}+\Lambda_{nn,n}Z_n\big)<-b_{n,n}^\tc\rt\},
$$
which aligns with \eqref{eq:gaussian_conservative_halfspace} for the Gaussian case. Again, by Bayes' rule, the associated density and outer-layer likelihood ratio are given by
\begin{equation}\label{eq:gh_outer_density}
 \widetilde\varphi_{\tt out}(v,\bz)
 \coloneqq
 \frac{f(v,\bz)\widetilde\kappa_n^\tc(v,\bz)}{\widetilde\kappa_n^\tc}
 \quad \text{and}\quad \widetilde\cL_{\tt out}(v,\bz)=\frac{\widetilde\kappa_n^\tc}{\widetilde\kappa_n^\tc(v,\bz)}\quad\text{for}~v\in\bbR_+,\bz\in\bbR^{n-1},
\end{equation}
respectively, where $f$ denotes the unconditional joint density of $(V,\bZ_{-n})$, 
\begin{align}
 \widetilde\kappa_n^\tc(v,\bz)&\coloneqq
 \bar\Phi\!\left(
 \frac{b_{n,n}^\tc+v\theta_{n,n}+\sqrt v\blambda_{n,n}^\top\bz}{\sqrt v\Lambda_{nn,n}} \right)~~\text{for}~v\in\bbR_+,\bz\in\bbR^{n-1},\label{eq:gh_conservative_prob}
\end{align}
and $\widetilde\kappa_n^\tc\coloneqq
\sP(V\theta_{n,n}+\sqrt{V}(\blambda_{n,n}^\top\bZ_{-n}+\Lambda_{nn,n}Z_n)<-b_{n,n}^\tc).$

\begin{proposition}[Outer-layer sampler for generalized hyperbolic cases]\label{prop:gh_sampler}
Suppose Assumption~\ref{ass:gh_model} holds. Let $F_{T_n}$ be the cumulative distribution function of $T_n$ in~\eqref{eq:gh_projection}.  The following procedure generates $(V,\bZ_{-n})$ from the density in~\eqref{eq:gh_outer_density}:
\begin{enumerate}[label=(\roman*)]
\item draw $U\sim\operatorname{Unif}(0,1)$ and set
\begin{equation}\label{eq:gh_truncated_projection}
 \widetilde T\coloneqq F_{T_n}^{-1}\!\big(UF_{T_n}(-d_{n,n}^\tc)\big);
\end{equation}
\item draw
$\widetilde V\sim\GIG(\lambda-1/2,\chi+\widetilde T^2,\psi+\eta_n^2)$ and set $\widetilde R\coloneqq(\widetilde T-\eta_n\widetilde V)/\sqrt{\widetilde V}$;
\item draw $\bxi\sim\cN(\bzero_{n},\bI_{n})$ and return
\begin{equation}\label{eq:gh_residual_law}
 \widetilde\bZ_{-n}=\bxi_{-n}+\left(\frac{\widetilde R}{\omega_{n,n}}-\frac{\sum_{i=1}^n\Lambda_{ni,n}\xi_i}{\omega_{n,n}^2}\right)\blambda_{n,n}.
\end{equation}
\end{enumerate}
\end{proposition}

\textbf{Inner-layer sampling and BLISS estimator.}  Given a realization $(v,\bz)$ of $(V,\bZ_{-n})$ and the associated external-asset vector $\bS_{-n}(v,\bz)$ for non-target banks generated from~\eqref{eq:gh_latent_model}, we define
\begin{equation}\label{eq:gh_actual_threshold}
 b_{n,n}(v,\bz)\coloneqq\log\frac{S^0_{n,n}}{v_n(\bS_{-n}(v,\bz))}~~~\text{and}~~~\widetilde\ell_n(v,\bz)\coloneqq
 \frac{b_{n,n}(v,\bz)+v\theta_{n,n}+\sqrt v\blambda_{n,n}^\top\bz}{\sqrt v\Lambda_{nn,n}}.
\end{equation}
Then, in line with the inner-layer sampling introduced in Section~\ref{subsec:bliss_new}, we generate samples of $Z_n$ from the following truncated standard normal density:
\begin{equation}\label{eq:gh_inner_density}
 \widetilde\varphi_{\tt in}(x| v,\bz)
 =\frac{\edit{\phi_1}(x)\sone_{\{x<-\widetilde\ell_n(v,\bz)\}}}{\bar\Phi(\widetilde\ell_n(v,\bz))}.
\end{equation}
The associated likelihood ratio is given by
\begin{equation}\label{eq:gh_Lin}
\widetilde\cL_{\tt in}(v,\bz)  = {\bar\Phi(\widetilde\ell_n(v,\bz))}.
\end{equation}
Accordingly, the BLISS estimator is finally constructed as follows:
\begin{equation}\label{eq:gh_bliss_estimator}
 \widetilde\Gamma_n^*
 \coloneqq
 h\big(\widetilde S_{1,n},\ldots,\widetilde S_{n,n}\big)\widetilde\cL_{\tt in}(v,\bz)\widetilde\cL_{\tt out}(v,\bz)
 =h\big(\widetilde S_{1,n},\ldots,\widetilde S_{n,n}\big)
 \bar\Phi\big(\widetilde\ell_n(\widetilde V,\widetilde\bZ_{-n})\big)\frac{\widetilde\kappa_n^\tc}{\widetilde\kappa_n^\tc(\widetilde V,\widetilde\bZ_{-n})},
\end{equation}
where $(\widetilde V,\widetilde\bZ_{-n})$ is generated by the sampler in Proposition~\ref{prop:gh_sampler}, $\widetilde Z_n=(Z_n|Z_n<-\widetilde\ell_n(\widetilde V,\widetilde\bZ_{-n}))$, and $\widetilde S_{i,n} = S^0_{i,n}\exp(\widetilde V\theta_{i,n}+{\widetilde V^{1/2}}\sum_{k=1}^i\Lambda_{ik,n}\widetilde Z_k)$ for $i=1,\ldots,n$. Note that the second equality in \eqref{eq:gh_bliss_estimator} holds by \eqref{eq:gh_outer_density} and \eqref{eq:gh_Lin}. It is straightforward that this estimator is unbiased, and the complete description of the BLISS procedure is provided in Algorithm~\ref{alg:GH_BLISS}.

\begin{algorithm}[t]
\caption{BLISS for Generalized Hyperbolic Log-returns}\label{alg:GH_BLISS}
\edit{\begin{algorithmic}[1]
\Procedure{outer-layer sampling}{}
\State Draw $\widetilde T$ via~\eqref{eq:gh_truncated_projection}
\State Draw $\widetilde V$ from the conditional GIG distribution in Proposition~\ref{prop:gh_sampler}
\State Set $\widetilde R=(\widetilde T-\eta_n\widetilde V)/\sqrt{\widetilde V}$
\State Reconstruct $\widetilde\bZ_{-n}$ via~\eqref{eq:gh_residual_law}
\State Set $\widetilde S_{i,n}=S^0_{i,n}\exp(\widetilde V\theta_{i,n}+{\widetilde V}^{1/2}\sum_{k=1}^i\Lambda_{ik,n}\widetilde Z_k)$ for $i=1,\ldots,n-1$

\State Compute $\widetilde\cL_{\tt out}(\widetilde V,\widetilde\bZ_{-n})$
\EndProcedure
\Procedure{inner-layer sampling}{$\widetilde V,\widetilde\bZ_{-n}$}
\State Compute $\widetilde\ell_n(\widetilde V,\widetilde\bZ_{-n})$ from~\eqref{eq:gh_actual_threshold}
\State Sample $\widetilde Z_n$ from~\eqref{eq:gh_inner_density}
\State Set $\widetilde S_{n,n}=S^0_{n,n}\exp(\widetilde V\theta_{n,n}+{\widetilde V}^{1/2}\sum_{k=1}^n\Lambda_{nk,n}\widetilde Z_k)$
\State Compute $\widetilde\cL_{\tt in}(\widetilde V,\widetilde\bZ_{-n})$
\EndProcedure
\State Set $T=h(\widetilde S_{1,n},\ldots,\widetilde S_{n,n})\widetilde\cL_{\tt in}(\widetilde V,\widetilde\bZ_{-n})\widetilde\cL_{\tt out}(\widetilde V,\widetilde\bZ_{-n})$
\State \textbf{return} $T$
\end{algorithmic}}
\end{algorithm}

\subsection{Asymptotic Efficiency}\label{subsec:gh_efficiency_and_gain}
In this subsection, we analyze the asymptotic efficiency of the BLISS estimator~\eqref{eq:gh_bliss_estimator} for generalized hyperbolic cases. Analogous to Section~\ref{sec:analysis}, we show the asymptotic optimality and bounded relative error of the proposed estimator and examine the value of outer-layer sampling.

\begin{proposition}[Left-tail asymptotics for the standardized return]\label{prop:gh_left_tail}
Suppose Assumptions~\ref{ass:gh_model} and~\ref{ass:gh_fixed} hold. Let $T=\eta V+\sqrt{V}R$, where $R\sim\cN(0,1)$, and write $F^-_T\!(x)=\sP(T<-x)$ for its left-tail function associated with~$(\lambda,\chi,\psi,\eta)$.
\begin{enumerate}[label=(\roman*)]
\item If $\psi>0$, or if $\psi=0$ and $\eta>0$, then
\begin{equation}\label{eq:gh_semih_tail}
 F^-_T\!(x)\sim C_*x^{\lambda-1}e^{-(\eta+\sqrt{\psi+\eta^2})x}
 \quad\text{for some constant}~C_*>0.
\end{equation}
\item If $\psi=0$, $\eta=0$, and $\lambda<0$, then
\begin{equation}\label{eq:gh_student_tail}
 F^-_T\!(x)\sim C_0x^{2\lambda}
 \quad\text{for some constant}~C_0>0.
\end{equation}
\item If $\psi=0$, $\eta<0$, and $\lambda<0$, then
\begin{equation}\label{eq:gh_negative_drift_tail}
 F^-_T\!(x)\sim C_\eta x^\lambda
 \quad\text{for some constant}~C_\eta>0.
\end{equation}
\end{enumerate}
\end{proposition}

Note that, by Assumption~\ref{ass:gh_fixed}, the standardized return $T_n$ has the same distribution as $T$ in the preceding proposition for all $n$ large enough. Accordingly, these results imply that in Case~(i), the standardized return exhibits the so-called ``semi-heavy'' left tail---an exponential decay rate modulated by a power-law or slowly varying function. In contrast, in Cases~(ii) and~(iii), the standardized return shows heavy left tails, specifically characterized by regular variation. Henceforth, we refer to Case (i) as the \emph{semi-heavy} branch and Cases (ii) and (iii) as the \emph{regularly varying} branch.
From these findings, one can intuit that the efficiency guarantees for the semi-heavy branch could resemble those for the Gaussian case discussed in Section~\ref{sec:analysis}, whereas the analysis for the regularly varying branch may depart from the Gaussian setting. This intuition will be formalized shortly.

In this subsection, the \emph{system-wide stability regime} means Definition~\ref{def:gaussian_stability_regime} with $D_{i,n}^\tc$ and $D_{i,n}^\f$ replaced by the distances-to-default $d_{i,n}^\tc$ and $d_{i,n}^\f$ defined in~\eqref{eq:gh_latent_thresholds}, as described in Section~\ref{subsec:GH model}; in particular, as $n$ grows,
$$
 \min_{i=1,\ldots,n}d_{i,n}^\tc\rightarrow\infty\quad\text{and}\quad
 \max_{i=1,\ldots,n}\!\big|{d_{i,n}^\f}/{d_{i,n}^\tc}-1\big|\to0.
$$
Based on this asymptotic regime, the following theorem demonstrates the efficiency of the estimator:

\begin{theorem}[Asymptotic efficiency of $\widetilde\Gamma_n^*$]\label{thm:gh_efficiency}
Suppose Assumptions~\ref{ass:gh_model} and~\ref{ass:gh_fixed} hold. Let $\widetilde\sE_n$ denote expectation under the outer-layer distribution~\eqref{eq:gh_outer_density} and the inner-layer distribution~\eqref{eq:gh_inner_density}. Then, under the system-wide stability regime in the sequence $\{\mathfrak{N}_n\}$ of networks, the following statements hold:
\begin{enumerate}[label=(\roman*)]
\item  In the semi-heavy branch, the BLISS estimator $\widetilde\Gamma_n^*$ is asymptotically optimal. Moreover, if
\begin{equation}\label{eq:tail_comparability}
 d_{n,n}^\f-d_{n,n}^\tc=O(1),
\end{equation}
then it has bounded relative error;
\item In the regularly varying branch, $\widetilde\Gamma_n^*$ has bounded relative error.
\end{enumerate}
\end{theorem}
As alluded to earlier, the two conclusions of this theorem reflect the different geometries of the underlying left-tails.
In the proof of the theorem, we show that the relative second moment of the estimator is controlled by the ratio of the conservative and favorable benchmark tail probabilities. 
In the semi-heavy branch, the left tail of the standardized return decays exponentially by~\eqref{eq:gh_semih_tail}. This implies that the said ratio grows as $\exp\{(\eta+\sqrt{\psi+\eta^2})(d_{n,n}^\f-d_{n,n}^\tc)\}$, which is not bounded by the system-wide stability regime alone. 
Hence, we additionally impose \eqref{eq:tail_comparability} for bounded relative error. 
By contrast, in the regularly varying branch described by \eqref{eq:gh_student_tail} and \eqref{eq:gh_negative_drift_tail}, it can be easily checked that the system-wide stability regime guarantees the boundedness of the ratio without further assumptions.
Namely, the same BLISS estimator achieves bounded relative error under weaker conditions in the regularly varying branch because its performance is less sensitive to the gap between the conservative and favorable distances-to-default.\footnote{\edit{This does not mean that heavy-tailed rare events are intrinsically easier to simulate. Rather, it indicates that the proposed conservative threshold approximation is more robust in the regularly varying branches.}}

To isolate the contribution of the outer layer, we define the inner-layer-only estimator $\widetilde\Gamma_n^0$ for generalized hyperbolic log-returns by sampling $(V,\bZ_{-n})$ from its original distribution and returning
$\widetilde\Gamma_n^0
 \coloneqq
 h_n(\bS)\bar\Phi\bigl(\widetilde\ell_n(V,\bZ_{-n})\bigr)$. We further let $\widetilde\sE_n^*$ and $\widetilde\sE_n^0$ denote the expectations under the sampling distributions of the BLISS and inner-layer-only estimators, respectively.
 Given this setup, the following proposition demonstrates that the BLISS estimator dominates the inner-layer-only estimator, suggesting that the proposed outer-layer sampling effectively eliminates far-from-default common-volatility states ($V$) and common-factor directions ($\bZ_{-n}$) that persist under inner-layer-only sampling.  

\begin{proposition}[Outer-layer efficiency gain in $\widetilde\Gamma_n^*$]\label{prop:gh_outer_gain}
Under the conditions of Theorem~\ref{thm:gh_efficiency}, the following statements hold:
\begin{enumerate}[label=(\roman*)]
\item In the semi-heavy branch, if $\rho_n\coloneqq{\|\blambda_{n,n}\|^2}/{\omega_{n,n}^2}
 \to\rho\in[0,1)$ as $n$ increases, then the limit
\begin{equation}
 \lim_{n\to\infty}\frac{\log\widetilde\sE_n^*[(\widetilde\Gamma_n^*)^2]}
 {\log\widetilde\sE_n^0[(\widetilde\Gamma_n^0)^2]}
\end{equation}
exists, and its explicit form is given in the proof. This limit is strictly greater than one if $\psi>0$ and equals $1+\rho$ otherwise.
\item In the regularly varying branch, 
\begin{equation}\label{eq:gh_student_ratio}
 \lim_{n\to\infty}\frac{\widetilde\sE_n^*[(\widetilde\Gamma_n^*)^2]}
 {\widetilde\sE_n^0[(\widetilde\Gamma_n^0)^2]}
 =0,
\end{equation}
without any convergence assumption on $\rho_n\in[0,1)$.
\end{enumerate}
\end{proposition}

\subsection{Distributional Examples and Algorithmic Implementation}\label{subsec:gh_examples}
We next illustrate the conditional sampling construction in Proposition~\ref{prop:gh_sampler} for several commonly used distributions within the generalized hyperbolic family. Since the remaining steps of Algorithm~\ref{alg:GH_BLISS} are unchanged, each example specifies only the distribution of $T_n$ in~\eqref{eq:gh_projection} and the corresponding conditional mixing distribution. Below, we discuss the cases of multivariate $t$ and variance-gamma distributions, while two additional cases are presented in Appendix~\ref{app:GH-ex}.

\textbf{Multivariate $t$.} Student’s $t$ distribution provides a symmetric model with polynomially decaying tails, in contrast to Gaussian and semi-heavy subclasses. Its degrees-of-freedom parameter controls tail thickness, making it a widely used alternative to the Gaussian model for financial returns \citep{BlattbergGonedes1974,McNeil:15}. To obtain this model, we set $\btheta_n=\bzero_n$ and let $V=\nu/G$, where $\nu>0$ and $G\sim\chi_\nu^2$ is independent of $\bZ$. Then,
$$
V\sim\operatorname{GIG}(-{\nu}/{2},\nu,0),
\qquad
\eta_n=0,
\quad\text{and}\quad
T_n=\sqrt{V}R_n\sim t_\nu.
$$
The conditional mixing distribution is
$$
\widetilde V\mid(\widetilde T=t)
\sim
\operatorname{GIG}\!\left(
-\frac{\nu+1}{2},\,\nu+t^2,\,0
\right).
$$
In this case, the conditional draw can be generated explicitly. Given $\widetilde T$, we draw $G^\ast\sim\chi_{\nu+1}^2$ independently and set
$
\widetilde V={(\nu+\widetilde T^2)}/{G^\ast}.
$

\textbf{Variance-gamma.} The variance-gamma distribution has semi-heavy exponential tails and accommodates skewness through the mean loading. It arises from Brownian motion with drift evaluated at a gamma time change and has been used to model financial returns and option prices \citep{MadanSeneta1990,MadanCarrChang1998}. In our framework, this distribution is obtained on the gamma boundary $\chi=0$, with $\psi>0$ and $\lambda>0$. Specifically,
$
V\sim\operatorname{Gamma}\!\left(
\text{shape}=\lambda,\,
\text{rate}={\psi}/{2}
\right),
$
and $T_n$ follows a univariate variance-gamma distribution. Given a draw $\widetilde T$ from its truncated distribution, the conditional mixing distribution is
$$
\widetilde V\mid(\widetilde T=t)
\sim
\operatorname{GIG}\!\left(
\lambda-\frac12,\,t^2,\,\psi+\eta_n^2
\right).
$$
Since $T_n$ has an absolutely continuous distribution, $\widetilde T\neq0$ almost surely. Consequently, this conditional GIG distribution is proper almost surely.
\end{revision}

\section{Computational Performance and Empirical Validation}\label{sec:numerics}
To provide numerical validation of the preceding theoretical analyses and discussions, we begin with a synthetic dataset that illustrates the computational efficiency of the BLISS estimator against the state-of-the-art as network size scales. We subsequently \edit{use a real-world dataset to evaluate the robustness of its simulation efficiency against changes in the modeling framework.} In all implementations, we use MATLAB R2023a on
a Linux workstation equipped with Dual Intel Xeon E5-2697 2.6GHz CPU
 and 128 GB RAM.

\subsection{\edit{Simulation Efficiency: A Synthetic Example}}\label{subsec:numeric_1}
We consider $n$-bank financial networks with $\bar p_i = 5$, $\bar p_{i0}=4$, $e_i = 0$, $S^0_i = S^0$, \edit{$\Lambda_{ii} = 0.1$}, and $\Lambda_{ij}=0$ for all $i=1,\ldots,n$ and $j\neq i$, where the values of $S^0$ will be specified later. We construct two relative liability matrices $\bPi^c$ and $\bPi^r$, representing complete and ring networks, respectively:
$$
\bPi^c = \left[\begin{array}{cccc}
      0 & \frac{.2}{n-1}&  \cdots & \frac{.2}{n-1}  \\
      \frac{.2}{n-1}& 0 & \ddots& \vdots \\
      \vdots& \ddots& \ddots& \frac{.2}{n-1}  \\
      \frac{.2}{n-1}&\cdots& \frac{.2}{n-1}& 0  \\
\end{array}\right]
\quad\text{and}\quad\bPi^r = \left[\begin{array}{ccccc}
     0 & .2& 0& \cdots & 0 \\[-3.5pt]
     0& 0 & .2& \ddots& \vdots  \\[-3.5pt]
     \vdots& \vdots& \ddots & \ddots& 0 \\[-3.5pt]
     0& 0& \cdots& 0 & .2  \\[-3.5pt]
     .2 & 0& \cdots& 0& 0  \\
\end{array}\right].
$$
Both matrices satisfy $\sum_{j=1}^n\pi_{ij}=(\bar p_i-\bar p_{i0})/\bar p_i = 0.2$ for all $i=1,\ldots,n$.
Under these configurations, we evaluate the relative error and \edit{runtime} of our BLISS method against \edit{the naive Monte Carlo (MC), the conditional importance sampling (CIS) method \citep{AhnZ-WSC:23}, and the cross-entropy (CE) method \citep{deBoer2005}, as network size $(n)$ increases.} Note that the relative error is defined as the standard error divided by the target estimate. We exclude illiquid assets in this example because, as alluded to earlier, the CIS method cannot be used once fire sales are incorporated.
For simplicity, we set $h\equiv1$ and aim to estimate bank $n$'s default probability in this example.

\begin{revision}
\begin{figure}[!t]
\centering
	\begin{tikzpicture}[font=\sffamily\small]
	\begin{semilogyaxis}[name=plot1,height=1.8in,width=0.38\textwidth,
    title = {Complete network},
	ylabel={Relative error},
    ytickten={-6,-5,-4,-3,-2,-1,-0},
    yticklabels={$10^{-6}$, , $10^{-4}$, , $10^{-2}$, , $10^0$},
	ymin=0.000001,	ymax=1,
    xmin=4, xmax=40,
	every tick label/.append style={font=\scriptsize},
	axis on top,
	scaled x ticks = false,
	xticklabel style={/pgf/number format/fixed},
	/pgf/number format/1000 sep={}]
	\addplot+[black, very thick, densely dotted, mark = none]	table[x index=0,y index=1, col sep=comma]{figures/ex1nc.csv};
	\addplot+[blue, very thick, dashed, mark = o, mark size = 1pt, mark options=solid]	table[x index=0,y index=3, col sep=comma]{figures/ex1nc.csv};
	\addplot+[gray, very thick, dashdotted, mark = none]	table[x index=0,y index=5, col sep=comma]{figures/ex1nc.csv};
	\addplot+[red, very thick, solid, mark = star, mark options=solid]	table[x index=0,y index=7, col sep=comma]{figures/ex1nc.csv};
	\end{semilogyaxis}
	\begin{semilogyaxis}[name=plot2,height=1.8in,width=0.38\textwidth,at={($(plot1.east)+(0.5in,0in)$)},anchor=west,
    title = {Ring network},
    legend style={legend pos=outer north east, font = \sffamily\small, legend entries={MC, CIS, CE, BLISS}},
    ytickten={-6,-5,-4,-3,-2,-1,-0},
    yticklabels={$10^{-6}$, , $10^{-4}$, , $10^{-2}$, , $10^0$},
	ymin=0.000001,	ymax=1,
    xmin=4, xmax=40,
	every tick label/.append style={font=\scriptsize},
	axis on top,
	scaled x ticks = false,
	xticklabel style={/pgf/number format/fixed},
	/pgf/number format/1000 sep={}]
	\addplot+[black, very thick, densely dotted, mark = none]	table[x index=0,y index=1, col sep=comma]{figures/ex1nr.csv};
	\addplot+[blue, very thick, dashed, mark = o, mark size = 1pt, mark options=solid]	table[x index=0,y index=3, col sep=comma]{figures/ex1nr.csv};
	\addplot+[gray, very thick, dashdotted, mark = none]	table[x index=0,y index=5, col sep=comma]{figures/ex1nr.csv};
	\addplot+[red, very thick, solid, mark = star, mark options=solid]	table[x index=0,y index=7, col sep=comma]{figures/ex1nr.csv};
	\end{semilogyaxis}
	\begin{semilogyaxis}[name=plot3,height=1.8in,width=0.38\textwidth,at={($(plot1.south)-(0in,0.2in)$)},anchor=north,
	ylabel={Runtime (sec)},
	every tick label/.append style={font=\scriptsize},
	axis on top,
	scaled x ticks = false,
	xticklabel style={/pgf/number format/fixed},
	/pgf/number format/1000 sep={},
        ytickten={-2,-1,-0,1,2,3,4},
        yticklabels={$10^{-2}$, , $10^{0}$, , $10^2$, , $10^4$},
	ymin=0.01,	ymax=11000,
    xmin=4, xmax=40]
	\addplot+[black, very thick, densely dotted, mark = none]	table[x index=0,y index=2, col sep=comma]{figures/ex1nc.csv};
	\addplot+[blue, very thick, dashed, mark = o, mark size = 1pt, mark options=solid]	table[x index=0,y index=4, col sep=comma]{figures/ex1nc.csv};
	\addplot+[gray, very thick, dashdotted, mark = none]	table[x index=0,y index=6, col sep=comma]{figures/ex1nc.csv};
	\addplot+[red, very thick, solid, mark = star, mark options=solid]	table[x index=0,y index=8, col sep=comma]{figures/ex1nc.csv};
	\end{semilogyaxis}
	\begin{semilogyaxis}[name=plot4,height=1.8in,width=0.38\textwidth,at={($(plot3.east)+(0.5in,0in)$)},anchor=west,
	every tick label/.append style={font=\scriptsize},
	axis on top,
	scaled x ticks = false,
	xticklabel style={/pgf/number format/fixed},
	/pgf/number format/1000 sep={},
        ytickten={-2,-1,-0,1,2,3,4},
        yticklabels={$10^{-2}$, , $10^{0}$, , $10^2$, , $10^4$},
	ymin=0.01,	ymax=11000,
    xmin=4, xmax=40]
	\addplot+[black, very thick, densely dotted, mark = none]	table[x index=0,y index=2, col sep=comma]{figures/ex1nr.csv};
	\addplot+[blue, very thick, dashed, mark = o, mark size = 1pt, mark options=solid]	table[x index=0,y index=4, col sep=comma]{figures/ex1nr.csv};
	\addplot+[gray, very thick, dashdotted, mark = none]	table[x index=0,y index=6, col sep=comma]{figures/ex1nr.csv};
	\addplot+[red, very thick, solid, mark = star, mark options=solid]	table[x index=0,y index=8, col sep=comma]{figures/ex1nr.csv};
	\end{semilogyaxis}
	\begin{semilogyaxis}[name=plot5,height=1.8in,width=0.38\textwidth,at={($(plot3.south)-(0in,0.2in)$)},anchor=north,
	xlabel={Network size ($n$)},
	ylabel={Efficiency ratio},
	every tick label/.append style={font=\scriptsize},
	axis on top,
	scaled x ticks = false,
	xticklabel style={/pgf/number format/fixed},
	/pgf/number format/1000 sep={},
	ymin=0.1,	ymax=1e11,
    xmin=4, xmax=40]
	\addplot+[black, very thick, densely dotted, mark = none]	table[x index=0,y expr={\thisrowno{1}^2*\thisrowno{2}/\thisrowno{1}^2/\thisrowno{2}}, col sep=comma]{figures/ex1nc.csv};
	\addplot+[blue, very thick, dashed, mark = o, mark size = 1pt, mark options=solid]	table[x index=0,y expr={\thisrowno{1}^2*\thisrowno{2}/\thisrowno{3}^2/\thisrowno{4}}, col sep=comma]{figures/ex1nc.csv};
	\addplot+[gray, very thick, dashdotted, mark = none]	table[x index=0,y expr={\thisrowno{1}^2*\thisrowno{2}/\thisrowno{5}^2/\thisrowno{6}}, col sep=comma]{figures/ex1nc.csv};
	\addplot+[red, very thick, solid, mark = star, mark options=solid]	table[x index=0,y expr={\thisrowno{1}^2*\thisrowno{2}/\thisrowno{7}^2/\thisrowno{8}}, col sep=comma]{figures/ex1nc.csv};
	\end{semilogyaxis}
	\begin{semilogyaxis}[name=plot6,height=1.8in,width=0.38\textwidth,at={($(plot5.east)+(0.5in,0in)$)},anchor=west,
	xlabel={Network size ($n$)},
	every tick label/.append style={font=\scriptsize},
	axis on top,
	scaled x ticks = false,
	xticklabel style={/pgf/number format/fixed},
	/pgf/number format/1000 sep={},
	ymin=0.1,	ymax=1e11,
    xmin=4, xmax=40]
	\addplot+[black, very thick, densely dotted, mark = none]	table[x index=0,y expr={\thisrowno{1}^2*\thisrowno{2}/\thisrowno{1}^2/\thisrowno{2}}, col sep=comma]{figures/ex1nr.csv};
	\addplot+[blue, very thick, dashed, mark = o, mark size = 1pt, mark options=solid]	table[x index=0,y expr={\thisrowno{1}^2*\thisrowno{2}/\thisrowno{3}^2/\thisrowno{4}}, col sep=comma]{figures/ex1nr.csv};
	\addplot+[gray, very thick, dashdotted, mark = none]	table[x index=0,y expr={\thisrowno{1}^2*\thisrowno{2}/\thisrowno{5}^2/\thisrowno{6}}, col sep=comma]{figures/ex1nr.csv};
	\addplot+[red, very thick, solid, mark = star, mark options=solid]	table[x index=0,y expr={\thisrowno{1}^2*\thisrowno{2}/\thisrowno{7}^2/\thisrowno{8}}, col sep=comma]{figures/ex1nr.csv};
	\end{semilogyaxis}
	\end{tikzpicture}
    \edit{\caption{\label{fig:comparison}Performance comparison between MC, CIS, CE, and BLISS for the synthetic example in Section~\ref{subsec:numeric_1} under the Gaussian log-return model for different network sizes. Performance metrics include relative error (top panels){, runtime (middle panels), and efficiency ratio (bottom panels)}. The left and right panels correspond to the complete-network and ring-network cases, respectively.}}
\end{figure}

\begin{figure}[!t]
\centering
	\begin{tikzpicture}[font=\sffamily\small]
	\begin{semilogyaxis}[name=plot1,height=1.8in,width=0.38\textwidth,
    title = {Complete network},
	ylabel={Relative error},
    ytickten={-3,-2,-1,-0},
	ymin=0.001,	ymax=1,
    xmin=4, xmax=40,
	every tick label/.append style={font=\scriptsize},
	axis on top,
	scaled x ticks = false,
	xticklabel style={/pgf/number format/fixed},
	/pgf/number format/1000 sep={}]
	\addplot+[black, very thick, densely dotted, mark = none]	table[x index=0,y index=1, col sep=comma]{figures/ex1tc.csv};
	\addplot+[blue, very thick, dashed, mark = o, mark size = 1pt, mark options=solid]	table[x index=0,y index=3, col sep=comma]{figures/ex1tc.csv};
	\addplot+[gray, very thick, dashdotted, mark = none]	table[x index=0,y index=5, col sep=comma]{figures/ex1tc.csv};
	\addplot+[red, very thick, solid, mark = star, mark options=solid]	table[x index=0,y index=7, col sep=comma]{figures/ex1tc.csv};
	\end{semilogyaxis}
	\begin{semilogyaxis}[name=plot2,height=1.8in,width=0.38\textwidth,at={($(plot1.east)+(0.5in,0in)$)},anchor=west,
    title = {Ring network},
    legend style={legend pos=outer north east, font = \sffamily\small, legend entries={MC, CIS, CE, BLISS}},
    ytickten={-3,-2,-1,-0},
	ymin=0.001,	ymax=1,
    xmin=4, xmax=40,
	every tick label/.append style={font=\scriptsize},
	axis on top,
	scaled x ticks = false,
	xticklabel style={/pgf/number format/fixed},
	/pgf/number format/1000 sep={}]
	\addplot+[black, very thick, densely dotted, mark = none]	table[x index=0,y index=1, col sep=comma]{figures/ex1tr.csv};
	\addplot+[blue, very thick, dashed, mark = o, mark size = 1pt, mark options=solid]	table[x index=0,y index=3, col sep=comma]{figures/ex1tr.csv};
	\addplot+[gray, very thick, dashdotted, mark = none]	table[x index=0,y index=5, col sep=comma]{figures/ex1tr.csv};
	\addplot+[red, very thick, solid, mark = star, mark options=solid]	table[x index=0,y index=7, col sep=comma]{figures/ex1tr.csv};
	\end{semilogyaxis}
	\begin{semilogyaxis}[name=plot3,height=1.8in,width=0.38\textwidth,at={($(plot1.south)-(0in,0.2in)$)},anchor=north,
	ylabel={Runtime (sec)},
	every tick label/.append style={font=\scriptsize},
	axis on top,
	scaled x ticks = false,
	xticklabel style={/pgf/number format/fixed},
	/pgf/number format/1000 sep={},
        ytickten={-2,-1,-0,1,2,3,4},
        yticklabels={$10^{-2}$, , $10^{0}$, , $10^2$, , $10^4$},
	ymin=0.01,	ymax=10000,
    xmin=4, xmax=40]
	\addplot+[black, very thick, densely dotted, mark = none]	table[x index=0,y index=2, col sep=comma]{figures/ex1tc.csv};
	\addplot+[blue, very thick, dashed, mark = o, mark size = 1pt, mark options=solid]	table[x index=0,y index=4, col sep=comma]{figures/ex1tc.csv};
	\addplot+[gray, very thick, dashdotted, mark = none]	table[x index=0,y index=6, col sep=comma]{figures/ex1tc.csv};
	\addplot+[red, very thick,     solid, mark = star, mark options=solid]	table[x index=0,y index=8, col sep=comma]{figures/ex1tc.csv};
	\end{semilogyaxis}
	\begin{semilogyaxis}[name=plot4,height=1.8in,width=0.38\textwidth,at={($(plot3.east)+(0.5in,0in)$)},anchor=west,
	every tick label/.append style={font=\scriptsize},
	axis on top,
	scaled x ticks = false,
	xticklabel style={/pgf/number format/fixed},
	/pgf/number format/1000 sep={},
        ytickten={-2,-1,-0,1,2,3,4},
        yticklabels={$10^{-2}$, , $10^{0}$, , $10^2$, , $10^4$},
	ymin=0.01,	ymax=10000,
    xmin=4, xmax=40]
	\addplot+[black, very thick, densely dotted, mark = none]	table[x index=0,y index=2, col sep=comma]{figures/ex1tr.csv};
	\addplot+[blue, very thick, dashed, mark = o, mark size = 1pt, mark options=solid]	table[x index=0,y index=4, col sep=comma]{figures/ex1tr.csv};
	\addplot+[gray, very thick, dashdotted, mark = none]	table[x index=0,y index=6, col sep=comma]{figures/ex1tr.csv};
	\addplot+[red, very thick, solid, mark = star, mark options=solid]	table[x index=0,y index=8, col sep=comma]{figures/ex1tr.csv};
	\end{semilogyaxis}
	\begin{semilogyaxis}[name=plot5,height=1.8in,width=0.38\textwidth,at={($(plot3.south)-(0in,0.2in)$)},anchor=north,
	xlabel={Network size ($n$)},
	ylabel={Efficiency ratio},
	every tick label/.append style={font=\scriptsize},
	axis on top,
	scaled x ticks = false,
	xticklabel style={/pgf/number format/fixed},
	/pgf/number format/1000 sep={},
	ymin=0.00001,	ymax=15000,
    xmin=4, xmax=40]
	\addplot+[black, very thick, densely dotted, mark = none]	table[x index=0,y expr={\thisrowno{1}^2*\thisrowno{2}/\thisrowno{1}^2/\thisrowno{2}}, col sep=comma]{figures/ex1tc.csv};
	\addplot+[blue, very thick, dashed, mark = o, mark size = 1pt, mark options=solid]	table[x index=0,y expr={\thisrowno{1}^2*\thisrowno{2}/\thisrowno{3}^2/\thisrowno{4}}, col sep=comma]{figures/ex1tc.csv};
	\addplot+[gray, very thick, dashdotted, mark = none]	table[x index=0,y expr={\thisrowno{1}^2*\thisrowno{2}/\thisrowno{5}^2/\thisrowno{6}}, col sep=comma]{figures/ex1tc.csv};
	\addplot+[red, very thick, solid, mark = star, mark options=solid]	table[x index=0,y expr={\thisrowno{1}^2*\thisrowno{2}/\thisrowno{7}^2/\thisrowno{8}}, col sep=comma]{figures/ex1tc.csv};
	\end{semilogyaxis}
	\begin{semilogyaxis}[name=plot6,height=1.8in,width=0.38\textwidth,at={($(plot5.east)+(0.5in,0in)$)},anchor=west,
	xlabel={Network size ($n$)},
	every tick label/.append style={font=\scriptsize},
	axis on top,
	scaled x ticks = false,
	xticklabel style={/pgf/number format/fixed},
	/pgf/number format/1000 sep={},
	ymin=0.00001,	ymax=15000,
    xmin=4, xmax=40]
	\addplot+[black, very thick, densely dotted, mark = none]	table[x index=0,y expr={\thisrowno{1}^2*\thisrowno{2}/\thisrowno{1}^2/\thisrowno{2}}, col sep=comma]{figures/ex1tr.csv};
	\addplot+[blue, very thick, dashed, mark = o, mark size = 1pt, mark options=solid]	table[x index=0,y expr={\thisrowno{1}^2*\thisrowno{2}/\thisrowno{3}^2/\thisrowno{4}}, col sep=comma]{figures/ex1tr.csv};
	\addplot+[gray, very thick, dashdotted, mark = none]	table[x index=0,y expr={\thisrowno{1}^2*\thisrowno{2}/\thisrowno{5}^2/\thisrowno{6}}, col sep=comma]{figures/ex1tr.csv};
	\addplot+[red, very thick, solid, mark = star, mark options=solid]	table[x index=0,y expr={\thisrowno{1}^2*\thisrowno{2}/\thisrowno{7}^2/\thisrowno{8}}, col sep=comma]{figures/ex1tr.csv};
	\end{semilogyaxis}
	\end{tikzpicture}
    \edit{\caption{\label{fig:comparison_t}Performance comparison between MC, CIS, CE, and BLISS for the synthetic example in Section~\ref{subsec:numeric_1} under the $t$ log-return model for different network sizes. Performance metrics include relative error (top panels){, runtime (middle panels), and efficiency ratio (bottom panels)}. The left and right panels correspond to the complete-network and ring-network cases, respectively.}}
    \end{figure}
    
\end{revision}

 \edit{Figure~\ref{fig:comparison} presents the associated numerical results \edit{under the Gaussian log-return model in Section~\ref{subsec:model}}, where we use \edit{100,000} simulation trials to implement each of the four methods and we set $S^0 = 5+0.025(n-4)$ for $n=4,8,12,\ldots,40$ to ensure the system-wide stability regime. {In these scenarios, the probability estimates range from order $10^{-2}$ to $10^{-5}$ as $n$ increases from $4$ to $40$. The left column corresponds to the results with the complete network case, while the right column depicts those with the ring network configuration. The top, middle, and bottom rows compare the relative error, runtime, and efficiency ratio, respectively, of the four methods. For method $m$, the efficiency ratio is defined as
$$
\mathrm{ER}_m
=\frac{\mathrm{SE}_{\mathrm{MC}}^2\,t_{\mathrm{MC}}}
       {\mathrm{SE}_m^2\,t_m},
$$
where $\mathrm{SE}_m$ and $t_m$ denote the estimated standard error produced by method $m$ and its measured runtime, respectively. Thus, values above one indicate an improvement over MC in terms of the variance--runtime product \citep{GlassermanKS:08,AhnZheng:23}.} 
 In addition, we conduct similar numerical experiments under the $t$ log-return model in Section~\ref{subsec:gh_examples} with $\nu=3$ degrees of freedom. In this case, we set $S^0 = 5+0.25(n-4)$ and $\gamma_n$ ranges from order $10^{-2}$ to $10^{-4}$ as $n$ increases from $4$ to $40$. The corresponding results are reported in Figure~\ref{fig:comparison_t}. {In these experiments, due to its heavy computational burden, CIS is implemented only up to the network size where the runtime remains under $10^4$ seconds.}

 The two figures demonstrate that the BLISS method completely dominates the state-of-the-art CIS method in terms of efficiency ratios. Specifically, BLISS generally achieves relative errors smaller than those of CIS, and this advantage becomes more pronounced as $n$ increases. For instance, the relative error of CIS increases exponentially with the network size under the $t$ log-return model, whereas BLISS exhibits consistently declining relative errors across all cases in both figures. Furthermore, as alluded to earlier, BLISS exhibits only a marginal increase in runtime as network size grows, whereas CIS suffers exponential growth in runtime as the network expands. {As a result, BLISS demonstrates significantly higher efficiency ratios than CIS, with the gap increasing dramatically as the network size grows.} Although the MC and CE methods are comparable to BLISS in speed, they suffer from substantially higher relative errors{, resulting in efficiency ratios far inferior to those of BLISS}. Finally, the above observations turn out to be consistent across different network structures.

Since $\Lambda_{ij}=0$ for all $i=1,\ldots,n$ and $j\neq i$ in this experiment, the efficiency gains of BLISS in Figure~\ref{fig:comparison} stem exclusively from its inner-layer sampling (Propositions~\ref{prop:new_outer_sampler} and~\ref{prop:superiority}), whereas those in Figure~\ref{fig:comparison_t} are driven by both outer- and inner-layer sampling (Propositions~\ref{prop:gh_sampler} and~\ref{prop:gh_outer_gain}). The specific contribution of the outer-layer sampling under various scenarios is examined in Section~\ref{subsec:numeric_2}.}

\subsection{{Robustness Analysis}}\label{subsec:numeric_2}
In this subsection, we draw on the European Banking Authority
(EBA) stress test dataset \citep{EBA:18}, which covers major European banks in 2017, to assess \edit{how robust the performance of the BLISS method is to changes in model primitives}. We first describe our data collection and model calibration process and then illustrate the numerical experiments and their outcomes.

\subsubsection{Data Selection and Calibration}\label{subsubsec:calibration}
In the EBA stress test dataset, we use its Internal Ratings-Based (IRB) total exposure value and Tier 1 capital, both sourced from the dataset, as proxies for initial total assets and initial net worth $(w_i)$, respectively. Each bank~$i$'s total liabilities $(\bar p_i)$ are defined as the difference between the above two values.
We further define its total interbank assets, $(\sum_{j\neq i}\bar p_{ji})$, as the IRB exposures to other institutions reported in the dataset. While the dataset covers 48 banks, we exclude 12 institutions with no interbank exposures, resulting in a final sample of 36 banks $(n=36)$. The raw data on total assets, net worth, and interbank assets for these 36 banks (in millions of euros) are provided in Table~\ref{tab:eba} \edit{of Appendix~\ref{app:numeric}}; the same dataset is used in \cite{bernard2022}. We use the numbers in the first column of the table for the subscripts of variables; e.g., HSBC's net worth is denoted by $w_1 = 125,976$.
Additionally, the final column of Table~\ref{tab:eba} presents the return volatilities of the corresponding equities in 2017, computed from daily stock price data provided by MarketWatch. For the ten banks not listed on stock exchanges, equity volatilities are set equal to the 2017 volatility of the EURO STOXX Banks Index and marked with *. We note that throughout 2017, European risk-free rates remained negligible, consistent with our framework. We therefore maintain the assumption of a zero risk-free rate.


Using the raw data in Table~\ref{tab:eba}, we construct external liabilities $(\bar p_{i0})$ by imposing a common ratio $(\alpha)$ of external liabilities to total liabilities across all banks; that is, $\bar p_{i0}=\alpha\bar p_i$ for all $i$. This ratio is calibrated such that the aggregate interbank liabilities match the aggregate interbank assets within the network, i.e., $\sum_{i=1}^n(\bar p_i - \bar p_{i0})=\sum_{i=1}^n\sum_{j=1}^n\bar p_{ji}$, or equivalently,
$\alpha = 1-(\sum_{i=1}^n\sum_{j=1}^n\bar p_{ji})/(\sum_{i=1}^n\bar p_i)$. Moreover, the initial value of external liquid assets $(S_i^0)$ is set to a fixed portion $(\beta)$ of its nominal non-interbank assets: $S_i^0 = \beta(w_i+\bar p_i-\sum_{j=1}^n\bar p_{ji})$. The remainder of the non-interbank assets defines the nominal value of illiquid assets: $e_i =(1-\beta)(w_i+\bar p_i-\sum_{j=1}^n\bar p_{ji})$. We set $\beta = 40\%$ in our experiments. This choice is informed by T02.03.1 in \cite{ECB:18}, which shows that liquid assets accounted for approximately 40\% of the total assets of major European banks consistently throughout 2017. \edit{We use two functional forms for the inverse demand function: $Q(x)=\exp(-\zeta x)$ and $Q(x)=1-\zeta x$ as in~\cite{Cifuentes:05}, \cite{Chen:16}, and~\cite{bernard2022}, where the values of $\zeta$ will be specified later.} One can easily check that Assumptions~\ref{ass:inverse_demand} and~\ref{ass:balance} are satisfied under these specifications.

Moreover, we calculate the volatility $\sigma_i$ of each bank $i$ based on the Merton model~\citep{Merton:74}; specifically, we solve a fixed-point equation:
$\sigma_i = (w_i/A_i)\sigma_i^{\tt E}/\Delta(\sigma_i)$,
where $w_i$, $A_i$, and $\sigma_i^{\tt E}$ denote bank $i$'s net worth, total assets, and equity volatility, respectively, reported in Table~\ref{tab:eba}. The term $\Delta(\sigma_i)$ represents the Black–Scholes delta of a call option on the bank’s total asset value, with the strike price set to its total liabilities, a zero risk‑free rate, and a one‑year maturity, and depends on the asset volatility $\sigma_i$. While the solution to the above fixed-point equation represents the volatility of the total assets, we use it as a proxy for the volatility of the external liquid assets. With this, we define the matrix $\bLambda$ as $\bLambda = \texttt{diag}(\sigma_1,\ldots,\sigma_{36})\bR$, where $\bR$ is the matrix in Table~\ref{tab:correlation} of Appendix~\ref{app:numeric} for the \edit{non-diagonal-loading} case and the identity matrix for the diagonal-loading case.\footnote{\edit{In Gaussian scenarios, non-diagonal-loading and diagonal-loading cases correspond to correlated and uncorrelated assets, respectively.}} In both cases, the diagonal elements of the matrix $\bLambda\bLambda^\top$ correspond to $\sigma_1^2,\ldots,\sigma_{36}^2$.

\begin{figure}[!t]
\centering
\begin{tikzpicture}
    \def\n{36}             
    \def\highlight{10}     
    \def\radius{2.1cm}       
    \def\nodeSize{1mm}     
    \def\labelOffset{3mm}  

    \foreach \i in {1,...,\n} {
        \pgfmathsetmacro{\angle}{90 - (\i - 1) * (360 / \n)}
        \coordinate (n\i) at (\angle:\radius);
        \coordinate (t\i) at (\angle:\radius + \labelOffset);
    }


    \foreach \i in {1,...,\n} {
        \foreach \j in {\i,...,\n} {
            \ifnum \i=\j \else
                \ifnum \i > \highlight
                    \ifnum \j > \highlight
                    \else
                        \draw[gray!90, opacity=0.9, thin] (n\i) -- (n\j);
                    \fi
                \else
                    \draw[gray!90, opacity=0.9, thin] (n\i) -- (n\j);
                \fi
            \fi
        }
    }

    \foreach \i in {1,...,\n} {
        \ifnum \i > \highlight
            \node[
                rectangle, draw=gray!80, fill=white,
                inner sep=0pt, minimum size=\nodeSize
            ] at (n\i) {};
        \node[font=\sffamily\tiny, text=black] at (t\i) {\i};
        \else
            \node[
                rectangle, draw=red!80, fill=red!10,
                inner sep=0pt, minimum size=\nodeSize
            ] at (n\i) {};
        \node[font=\sffamily\tiny, text=red] at (t\i) {\i};
        \fi
    }

\end{tikzpicture}
    \caption{The recovered interbank network structure from the 2018 EBA stress test data. Nodes represent individual banks, and edges stand for their interbank exposures. The red-colored nodes correspond to the core banks (banks 1 to 10).}
    \label{fig:topology}
\end{figure}
Since bilateral interbank exposure data are unavailable, we reconstruct \edit{the core-periphery network} used in \cite{Chen:16}; see Figure~\ref{fig:topology} for \edit{its} graphical illustration. This reconstruction is based on an entropy-minimization method in \cite{Upper:04}, defining the core as the ten largest banks by total assets (Banks 1 to 10).
The core-periphery network is a relatively sparse, non-extreme structure widely observed in real-world financial systems~\citep{Craig:14}. 


\subsubsection{\edit{Experimental Setup}}
\edit{In the numerical experiments,} we compute the bond price defined in~\eqref{eq:decomp} in the setup described above. We select Allied Irish Banks Group plc as the target institution. As one of the largest banks in the Republic of Ireland, it ranks 36th in the EBA dataset. This choice is motivated by the fact that it exhibits the highest volatility among the banks under consideration, making the assessment of its bond price particularly meaningful.  

We use standard Monte Carlo (MC) as a benchmark to evaluate the performance of our BLISS methodology. In addition, to highlight the importance of outer-layer importance sampling, we compare our approach against a version of BLISS that excludes this step, which we refer to as Inner-Layer Importance Sampling (ILIS). For all numerical experiments in this subsection, we use $10^6$ simulation trials to implement MC, ILIS, and BLISS.

\edit{For external liquid assets, we set $S_i=S_i^0\exp(X_i)$ and consider the following
three distributions:
\begin{itemize}
    \item Gaussian: $X_i=-{\sigma_i^2}/{2}+\sum_{k=1}^i\Lambda_{ik}Z_k$ with $Z_1,\ldots,Z_n\overset{\text{i.i.d.}}{\sim}\cN(0,1)$;
    \item Multivariate $t$: $X_i=\sqrt{V_t}\sum_{k=1}^i\Lambda_{ik}Z_k$ with $V_t = 3/\chi_3^2$;
    \item Variance-gamma: $X_i=(-\sigma_i^2/2)V_{vg} + \sqrt{V_{vg}}\sum_{k=1}^i\Lambda_{ik}Z_k$ with $V_{vg} \sim \operatorname{Gamma}(1,1)$.
\end{itemize}
{Here, $\sE[S_i]=S_i^0$ under the Gaussian and variance-gamma log-return models, whereas $\sE[S_i]=\infty$ under the $t$ model, in which $S_i^0$ serves only as a scale parameter, as discussed in Section~\ref{subsec:GH model}.} Under the diagonal-loading case, Gaussian shocks are independent, whereas multivariate $t$ and variance-gamma log-returns remain nonlinearly dependent through their respective mixing variables $V_t$ and $V_{vg}$, and multivariate $t$ log-returns additionally exhibit asymptotic tail dependence.

To demonstrate the robustness of our methodology across different model parameters and modeling frameworks, we consider the following four directions:

\begin{enumerate}[label=(\roman*)]
\item\label{item:nu} \emph{Varying price impact.} We first evaluate the performance of MC, ILIS, and BLISS across different values of $\zeta$
under both inverse demand functions.

\item\label{item:delta}  \emph{Varying bankruptcy cost.} While our theoretical analysis does not take into account bankruptcy costs, we can seamlessly incorporate the bankruptcy cost modeling of \cite{Glasserman:15} into our framework and numerically check the performance of BLISS. Specifically, under the said bankruptcy cost modeling, bank $i$'s clearing payment becomes
$$
p_i
    =\min\!\left\{\bar p_i,
      \big(r_i-\delta(\bar p_i-r_i)\big)^+\right\},
      $$
where $r_i\coloneqq s_i+qe_i+\sum_j\pi_{ji}p_j$ denotes bank $i$'s residual assets after counting liquidation and
incoming interbank payments.
Namely, in the event of insolvency, the term
$\delta(\bar p_i-r_i)$ represents the bankruptcy cost deducted from bank $i$'s residual assets. In this setup, we assess the performance of the three methodologies across different values of $\delta$. The case $\delta=0$ recovers the original model. Note that even when bankruptcy costs are incorporated, bank~$n$'s default event can still be represented as $\{s_n<v_n(\bs_{-n})\}$ since the clearing payment and price are monotone in $\bs$.

\item\label{item:random_nu}  \emph{Random price impact.} While our modeling framework employs a deterministic inverse demand function, we numerically observe the impact of stochastic inverse demand functions on the performance of BLISS. On each simulation path, one
market-wide $\zeta$ is drawn independently of the asset shock and applied to
all banks. We use beta and uniform distributions for $\zeta$:
$\zeta\sim(2\times10^{-8})\times\operatorname{Beta}(2,2)$ and 
$\zeta\sim U(0,2\times10^{-8})$.
Both distributions have the same support and mean. The beta distribution
concentrates more mass around its mean, whereas the uniform
distribution places relatively more mass near the endpoints.

\item\label{item:priority}  \emph{External-liability priority.} Our basic model adopts the pro rata rule to repay liabilities. However, in reality, external liabilities (e.g., deposits) may have a higher priority than interbank liabilities. Hence, in this direction, we replace the equal seniority with a waterfall in
which external liabilities are paid before interbank liabilities. Let $\bar p_i^\iota \coloneqq \bar p_i-\bar p_{i0}$ and $\pi_{ij}^\iota \coloneqq \frac{\bar p_{ij}}{\bar p_i^\iota}$ denote bank $i$'s total interbank liabilities and the proportion of its interbank liabilities owed to bank~$j$, respectively. Then, bank $i$'s clearing payments---$p_{i0}$ for external liabilities and $p_i^\iota$ for interbank liabilities---satisfy
$$
p_{i0}=\min\!\left\{\bar p_{i0},\Big(s_i+qe_i+\sum\nolimits_j\pi_{ji}^\iota p_j^\iota\Big)^+\right\}~\text{and}~
    p_i^\iota
    =\min\!\left\{\bar p_i^\iota,\Big(s_i+qe_i+\sum\nolimits_j\pi_{ji}^\iota p_j^\iota-\bar p_{i0}\Big)^+\right\}.
$$
The fire-sale rule remains unchanged; in particular,
$$
    q = Q\left(\sum_{i=1}^{n}\lt\{\frac{(\bar p_i -  s_i- \sum_j \pi_{ji}^\iota p_j^\iota)^+}{q} \wedge e_i \rt\}\right).
$$
In this setup, we estimate the price of the external liability claim, i.e., $\sE[p_{n0}/\bar p_{n0}]$. To that end, BLISS computes the threshold $v_n(\bs_{-n}) =\bar p_{n0}-\check qe_n-\sum_{i\neq n}\pi_{in}^\iota\check p_i^\iota$, where $\check q$ and $\{\check p_i^\iota\}_{i=1}^n$ are the equilibrium asset price and clearing payments in a fictitious system with $\bS_{-n}=\bs_{-n}$, in which bank $n$ sells all its illiquid assets and makes no interbank payments.\footnote{\edit{The fictitious system in Section~\ref{sec:prelim} fixes bank $n$'s interbank payments at their full nominal amounts, whereas here we fix them at zero. This distinction reflects the different default boundaries. In the former case, the boundary $p_n=\bar p_n$ corresponds to full repayment of all liabilities, including interbank ones. In the latter case, when the senior external liabilities are just fully repaid, no payments are made for the junior interbank liabilities due to priority.}}
\end{enumerate}
Unless a quantity is varied explicitly, we use 
$\zeta=10^{-8}$, impose equal seniority, and set the bankruptcy-cost factor $(\delta)$
to zero.}

\begin{revision}
\subsubsection{Numerical Results}
We find that all our numerical results {share the same qualitative conclusions} under both exponential and linear inverse demand functions. Hence, we report the exponential case in the main text and defer the linear case to Appendix~\ref{app:numeric}.  These findings demonstrate the application and robustness of BLISS across variations in price impact, bankruptcy costs, random price impact, and external-liability priority. In all Gaussian, $t$, and variance-gamma (VG) specifications we consider, {BLISS improves on MC in terms of efficiency ratios}. The additional gains of BLISS over ILIS depend on the loading structure and the return distribution, with both methods exhibiting the expected equivalence in the Gaussian diagonal-loading case.

For performance metrics, we use the standard error of the estimated normalized bond prices and the efficiency ratio that accounts for both precision and computational cost.  Because the normalized bond price is estimated as one minus expected fractional loss, its standard error is also the standard error of the corresponding expected-loss estimate. We also report the bond yield implied by the estimated price $\hat p_n$ in~\eqref{eq:decomp}, defined as $y_n\coloneqq-\log\hat p_n$ and expressed in basis points. Since all liabilities mature in one year and the risk-free rate is set to zero, $y_n$ coincides with the credit spread of bank~$n$'s bond. Detailed discussions on estimator performance and robustness follow below.

\textbf{Results for \ref{item:nu} and \ref{item:delta}.} Figures~\ref{fig:nu} and~\ref{fig:delta} illustrate the performance of MC, ILIS, and BLISS under changes in the price-impact factor $\zeta$ and the bankruptcy-cost factor $\delta$, respectively. As $\zeta$ or $\delta$ decreases to zero, standard errors decline for all three methods, while BLISS commonly becomes more efficient relative to MC.
BLISS also outperforms ILIS in all \edit{non-diagonal-loading} panels and in the $t$ and VG diagonal-loading panels. The distinction between variance reduction and computational efficiency is especially clear under VG and $t$ returns: ILIS produces substantially smaller standard errors than MC, but its efficiency ratio is often below one because the reduction in variance does not offset its additional runtime. BLISS achieves a more favorable balance between these two quantities.

The comparison between loading structures helps explain the role of the outer-layer sampling in BLISS. In Gaussian diagonal-loading scenarios, the standard errors of ILIS and BLISS coincide, while small differences in measured efficiency ratios can arise from runtime variation. This is because the outer-layer importance sampling in BLISS is effective only under correlation (specifically, $\bZ_{-n}=\bxi_{-n}$ in~\eqref{eq:outer_rank_one} for uncorrelated cases).   With correlated Gaussian returns, outer-layer conditioning changes the distribution of the non-target factors and yields additional variance reduction. Diagonal loading has a different interpretation under the $t$ and VG models: the shared mixing variable continues to couple bank returns even when the Gaussian loading matrix is diagonal. BLISS conditions this common source of variation as well, and the gains observed in those panels demonstrate the value of doing so in these experiments. These observations are consistent with the role of outer-layer conditioning identified in the theoretical analysis (Propositions~\ref{prop:superiority} and~\ref{prop:gh_outer_gain}).

Table~\ref{tab:12} of Appendix~\ref{app:numeric} provides the corresponding yield estimates for the price-impact and bankruptcy-cost experiments, from which we confirm that lower factors result in smaller yields. We interpret these yield reductions as lower expected credit losses: weaker price impact mitigates losses from asset liquidation, while lower bankruptcy costs preserve resources available for repayment. {Yields remain positive when either factor is zero, reflecting the credit risk from direct shocks and interbank contagion; the increase in yields with each factor quantifies the amplification due to fire sales and bankruptcy costs. Yields also differ across return models, reflecting differences in both tail behavior and return variance.}

\textbf{Results for \ref{item:random_nu} and \ref{item:priority}.} Figure~\ref{fig:scenario} extends the comparison to random price impact and external-liability priority. In the Beta and Uniform random price-impact scenarios, BLISS remains more efficient than MC and retains its advantage over ILIS, apart from their Gaussian diagonal-loading equivalence. The external-liability priority scenario also preserves this pattern: BLISS continues to improve on MC and, where the outer layer is nontrivial, on ILIS. Its decent efficiency ratio shows that randomizing the price impact or a change in the repayment priority does not necessarily weaken its computational advantage.

Table~\ref{tab:34} of Appendix~\ref{app:numeric} presents the corresponding yields. Compared to the baseline scenario (i.e., the yields with $\delta=0$ in Table~\ref{tab:12}), these findings indicate higher expected losses under the random-price-impact specifications. Yields under uniformly distributed $\zeta$ are slightly higher than those under beta-distributed $\zeta$, but the differences are limited. Under external-liability priority scenarios, senior claim yields exhibit greater variability across model specifications. 
At a given level of repayment resources, seniority improves the external creditor's recovery. However, preferential external payments can reduce interbank receipts, affecting other banks' repayment capacity and liquidation decisions. The resulting changes in incoming payments and fire-sale prices can offset the direct protection. The balance between direct protection and changes in repayments can therefore vary with the return model and loading structure.

\begin{figure}[!t]
\centering
	\begin{tikzpicture}[font=\sffamily\small]
	\begin{semilogyaxis}[name=plot1,height=1.8in,width=0.315\textwidth,
    title = {Gaussian},
	ylabel={Standard error},
	every tick label/.append style={font=\scriptsize},
	axis on top,
 	scaled x ticks = false,    x dir=reverse,
	xticklabel style={/pgf/number format/fixed},
	/pgf/number format/1000 sep={}]
	\addplot+[black, densely dotted, thick, mark = o, mark size = 2pt, mark options=solid]	table[x index=0,y index=2, col sep=comma]{figures/ex2GexpC.csv};
	\addplot+[blue, thick, dashed, mark = square, mark size = 2pt, mark options=solid]	table[x index=0,y index=5, col sep=comma]{figures/ex2GexpC.csv};
	\addplot+[red, thick, solid, mark = star, mark options=solid]	table[x index=0,y index=8, col sep=comma]{figures/ex2GexpC.csv};
	\end{semilogyaxis}
	\begin{semilogyaxis}[name=plot2,height=1.8in,width=0.315\textwidth,at={($(plot1.east)+(0.4in,0in)$)},anchor=west,
    title = {$t$},
	every tick label/.append style={font=\scriptsize},
	axis on top,
 	scaled x ticks = false,    x dir=reverse,
	xticklabel style={/pgf/number format/fixed},
	/pgf/number format/1000 sep={}]
	\addplot+[black, densely dotted, thick, mark = o, mark size = 2pt, mark options=solid]	table[x index=0,y index=2, col sep=comma]{figures/ex2TexpC.csv};
	\addplot+[blue, thick, dashed, mark = square, mark size = 2pt, mark options=solid]	table[x index=0,y index=5, col sep=comma]{figures/ex2TexpC.csv};
	\addplot+[red, thick, solid, mark = star, mark options=solid]	table[x index=0,y index=8, col sep=comma]{figures/ex2TexpC.csv};
	\end{semilogyaxis}
	\begin{semilogyaxis}[name=plot3,height=1.8in,width=0.315\textwidth,at={($(plot2.east)+(0.4in,0in)$)},anchor=west,
    title = {Variance-gamma},
	every tick label/.append style={font=\scriptsize},
	axis on top,
 	scaled x ticks = false,    x dir=reverse,
	xticklabel style={/pgf/number format/fixed},
	/pgf/number format/1000 sep={},
    legend pos=outer north east, legend entries={MC, ILIS, BLISS}]
	\addplot+[black, densely dotted, thick, mark = o, mark size = 2pt, mark options=solid]	table[x index=0,y index=2, col sep=comma]{figures/ex2VGexpC.csv};
	\addplot+[blue, thick, dashed, mark = square, mark size = 2pt, mark options=solid]	table[x index=0,y index=5, col sep=comma]{figures/ex2VGexpC.csv};
	\addplot+[red, thick, solid, mark = star, mark options=solid]	table[x index=0,y index=8, col sep=comma]{figures/ex2VGexpC.csv};
	\end{semilogyaxis}
	\begin{semilogyaxis}[name=plot4,height=1.8in,width=0.315\textwidth,at={($(plot1.south)-(0in,0.2in)$)},anchor=north,
	xlabel={$\zeta~(\times10^{-8})$},
	ylabel={Efficiency ratio},
	every tick label/.append style={font=\scriptsize},
	axis on top,
 	scaled x ticks = false,    x dir=reverse,
	xticklabel style={/pgf/number format/fixed},
	/pgf/number format/1000 sep={}]
	\addplot+[black, densely dotted, thick, mark = o, mark size = 2pt, mark options=solid]	table[x index=0,y index=3, col sep=comma]{figures/ex2GexpC.csv};
	\addplot+[blue, thick, dashed, mark = square, mark size = 2pt, mark options=solid]	table[x index=0,y index=6, col sep=comma]{figures/ex2GexpC.csv};
	\addplot+[red, thick, solid, mark = star, mark options=solid]	table[x index=0,y index=9, col sep=comma]{figures/ex2GexpC.csv};
	\end{semilogyaxis}
	\begin{semilogyaxis}[name=plot5,height=1.8in,width=0.315\textwidth,at={($(plot4.east)+(0.4in,0in)$)},anchor=west,
	xlabel={$\zeta~(\times10^{-8})$},
	every tick label/.append style={font=\scriptsize},
	axis on top,
 	scaled x ticks = false,    x dir=reverse,
	xticklabel style={/pgf/number format/fixed},
	/pgf/number format/1000 sep={}]
	\addplot+[black, densely dotted, thick, mark = o, mark size = 2pt, mark options=solid]	table[x index=0,y index=3, col sep=comma]{figures/ex2TexpC.csv};
	\addplot+[blue, thick, dashed, mark = square, mark size = 2pt, mark options=solid]	table[x index=0,y index=6, col sep=comma]{figures/ex2TexpC.csv};
	\addplot+[red, thick, solid, mark = star, mark options=solid]	table[x index=0,y index=9, col sep=comma]{figures/ex2TexpC.csv};
	\end{semilogyaxis}
	\begin{semilogyaxis}[name=plot6,height=1.8in,width=0.315\textwidth,at={($(plot5.east)+(0.4in,0in)$)},anchor=west,
	xlabel={$\zeta~(\times10^{-8})$},
	every tick label/.append style={font=\scriptsize},
	axis on top,
 	scaled x ticks = false,    x dir=reverse,
	xticklabel style={/pgf/number format/fixed},
	/pgf/number format/1000 sep={}]
	\addplot+[black, densely dotted, thick, mark = o, mark size = 2pt, mark options=solid]	table[x index=0,y index=3, col sep=comma]{figures/ex2VGexpC.csv};
	\addplot+[blue, thick, dashed, mark = square, mark size = 2pt, mark options=solid]	table[x index=0,y index=6, col sep=comma]{figures/ex2VGexpC.csv};
	\addplot+[red, thick, solid, mark = star, mark options=solid]	table[x index=0,y index=9, col sep=comma]{figures/ex2VGexpC.csv};
	\end{semilogyaxis}
    \node[yshift=29pt] 
    at ($(plot1.north west)!0.5!(plot3.north east)$) 
    {(a) Non-diagonal-loading case};
    
	\begin{semilogyaxis}[name=plot7,height=1.8in,width=0.315\textwidth,at={($(plot4.south)-(0in,1in)$)},anchor=north,
    title = {Gaussian},
	ylabel={Standard error},
	every tick label/.append style={font=\scriptsize},
	axis on top,
 	scaled x ticks = false,    x dir=reverse,
	xticklabel style={/pgf/number format/fixed},
	/pgf/number format/1000 sep={}]
	\addplot+[black, densely dotted, thick, mark = o, mark size = 2pt, mark options=solid]	table[x index=0,y index=2, col sep=comma]{figures/ex2GexpD.csv};
	\addplot+[blue, thick, dashed, mark = square, mark size = 2pt, mark options=solid]	table[x index=0,y index=5, col sep=comma]{figures/ex2GexpD.csv};
	\addplot+[red, thick, solid, mark = star, mark options=solid]	table[x index=0,y index=8, col sep=comma]{figures/ex2GexpD.csv};
	\end{semilogyaxis}
	\begin{semilogyaxis}[name=plot8,height=1.8in,width=0.315\textwidth,at={($(plot7.east)+(0.4in,0in)$)},anchor=west,
    title = {$t$},
	every tick label/.append style={font=\scriptsize},
	axis on top,
 	scaled x ticks = false,    x dir=reverse,
	xticklabel style={/pgf/number format/fixed},
	/pgf/number format/1000 sep={}]
	\addplot+[black, densely dotted, thick, mark = o, mark size = 2pt, mark options=solid]	table[x index=0,y index=2, col sep=comma]{figures/ex2TexpD.csv};
	\addplot+[blue, thick, dashed, mark = square, mark size = 2pt, mark options=solid]	table[x index=0,y index=5, col sep=comma]{figures/ex2TexpD.csv};
	\addplot+[red, thick, solid, mark = star, mark options=solid]	table[x index=0,y index=8, col sep=comma]{figures/ex2TexpD.csv};
	\end{semilogyaxis}
	\begin{semilogyaxis}[name=plot9,height=1.8in,width=0.315\textwidth,at={($(plot8.east)+(0.4in,0in)$)},anchor=west,
    title = {Variance-gamma},
	every tick label/.append style={font=\scriptsize},
	axis on top,
 	scaled x ticks = false,    x dir=reverse,
	xticklabel style={/pgf/number format/fixed},
	/pgf/number format/1000 sep={},
    legend pos=outer north east, legend entries={MC, ILIS, BLISS}]
	\addplot+[black, densely dotted, thick, mark = o, mark size = 2pt, mark options=solid]	table[x index=0,y index=2, col sep=comma]{figures/ex2VGexpD.csv};
	\addplot+[blue, thick, dashed, mark = square, mark size = 2pt, mark options=solid]	table[x index=0,y index=5, col sep=comma]{figures/ex2VGexpD.csv};
	\addplot+[red, thick, solid, mark = star, mark options=solid]	table[x index=0,y index=8, col sep=comma]{figures/ex2VGexpD.csv};
	\end{semilogyaxis}
	\begin{semilogyaxis}[name=plot10,height=1.8in,width=0.315\textwidth,at={($(plot7.south)-(0in,0.2in)$)},anchor=north,
	xlabel={$\zeta~(\times10^{-8})$},
	ylabel={Efficiency ratio},
	every tick label/.append style={font=\scriptsize},
	axis on top,
 	scaled x ticks = false,    x dir=reverse,
	xticklabel style={/pgf/number format/fixed},
	/pgf/number format/1000 sep={}]
	\addplot+[black, densely dotted, thick, mark = o, mark size = 2pt, mark options=solid]	table[x index=0,y index=3, col sep=comma]{figures/ex2GexpD.csv};
	\addplot+[blue, thick, dashed, mark = square, mark size = 2pt, mark options=solid]	table[x index=0,y index=6, col sep=comma]{figures/ex2GexpD.csv};
	\addplot+[red, thick, solid, mark = star, mark options=solid]	table[x index=0,y index=9, col sep=comma]{figures/ex2GexpD.csv};
	\end{semilogyaxis}
	\begin{semilogyaxis}[name=plot11,height=1.8in,width=0.315\textwidth,at={($(plot10.east)+(0.4in,0in)$)},anchor=west,
	xlabel={$\zeta~(\times10^{-8})$},
	every tick label/.append style={font=\scriptsize},
	axis on top,
 	scaled x ticks = false,    x dir=reverse,
	xticklabel style={/pgf/number format/fixed},
	/pgf/number format/1000 sep={}]
	\addplot+[black, densely dotted, thick, mark = o, mark size = 2pt, mark options=solid]	table[x index=0,y index=3, col sep=comma]{figures/ex2TexpD.csv};
	\addplot+[blue, thick, dashed, mark = square, mark size = 2pt, mark options=solid]	table[x index=0,y index=6, col sep=comma]{figures/ex2TexpD.csv};
	\addplot+[red, thick, solid, mark = star, mark options=solid]	table[x index=0,y index=9, col sep=comma]{figures/ex2TexpD.csv};
	\end{semilogyaxis}
	\begin{semilogyaxis}[name=plot12,height=1.8in,width=0.315\textwidth,at={($(plot11.east)+(0.4in,0in)$)},anchor=west,
	xlabel={$\zeta~(\times10^{-8})$},
	every tick label/.append style={font=\scriptsize},
	axis on top,
 	scaled x ticks = false,    x dir=reverse,
	xticklabel style={/pgf/number format/fixed},
	/pgf/number format/1000 sep={}]
	\addplot+[black, densely dotted, thick, mark = o, mark size = 2pt, mark options=solid]	table[x index=0,y index=3, col sep=comma]{figures/ex2VGexpD.csv};
	\addplot+[blue, thick, dashed, mark = square, mark size = 2pt, mark options=solid]	table[x index=0,y index=6, col sep=comma]{figures/ex2VGexpD.csv};
	\addplot+[red, thick, solid, mark = star, mark options=solid]	table[x index=0,y index=9, col sep=comma]{figures/ex2VGexpD.csv};
	\end{semilogyaxis}
    \node[yshift=29pt] 
    at ($(plot7.north west)!0.5!(plot9.north east)$) 
    {(b) Diagonal-loading case};
	\end{tikzpicture}
    \edit{\caption{\label{fig:nu} Performance comparison between MC, ILIS, and BLISS for the EBA example in Section~\ref{subsec:numeric_2} across different values of the price-impact factor $(\zeta)$ with the exponential inverse demand function. The three columns present the results with Gaussian, $t$, and variance-gamma log-return models, respectively. The upper six panels correspond to the \edit{non-diagonal-loading} case, while the lower six panels describe the diagonal-loading case. In each of the two cases, the top three panels compare standard errors of the three methods, and the bottom three panels exhibit their efficiency ratios.}}
\end{figure}
\begin{figure}[!t]
\centering
	\begin{tikzpicture}[font=\sffamily\small]
	\begin{semilogyaxis}[name=plot1,height=1.8in,width=0.315\textwidth,
    title = {Gaussian},
	ylabel={Standard error},
	every tick label/.append style={font=\scriptsize},
	axis on top,
 	scaled x ticks = false,    x dir=reverse,
	xticklabel style={/pgf/number format/fixed},
	/pgf/number format/1000 sep={}]
	\addplot+[black, densely dotted, thick, mark = o, mark size = 2pt, mark options=solid]	table[x index=10,y index=12, col sep=comma]{figures/ex2GexpC.csv};
	\addplot+[blue, thick, dashed, mark = square, mark size = 2pt, mark options=solid]	table[x index=10,y index=15, col sep=comma]{figures/ex2GexpC.csv};
	\addplot+[red, thick, solid, mark = star, mark options=solid]	table[x index=10,y index=18, col sep=comma]{figures/ex2GexpC.csv};
	\end{semilogyaxis}
	\begin{semilogyaxis}[name=plot2,height=1.8in,width=0.315\textwidth,at={($(plot1.east)+(0.4in,0in)$)},anchor=west,
    title = {$t$},
	every tick label/.append style={font=\scriptsize},
	axis on top,
 	scaled x ticks = false,    x dir=reverse,
	xticklabel style={/pgf/number format/fixed},
	/pgf/number format/1000 sep={}]
	\addplot+[black, densely dotted, thick, mark = o, mark size = 2pt, mark options=solid]	table[x index=10,y index=12, col sep=comma]{figures/ex2TexpC.csv};
	\addplot+[blue, thick, dashed, mark = square, mark size = 2pt, mark options=solid]	table[x index=10,y index=15, col sep=comma]{figures/ex2TexpC.csv};
	\addplot+[red, thick, solid, mark = star, mark options=solid]	table[x index=10,y index=18, col sep=comma]{figures/ex2TexpC.csv};
	\end{semilogyaxis}
	\begin{semilogyaxis}[name=plot3,height=1.8in,width=0.315\textwidth,at={($(plot2.east)+(0.4in,0in)$)},anchor=west,
    title = {Variance-gamma},
	every tick label/.append style={font=\scriptsize},
	axis on top,
 	scaled x ticks = false,    x dir=reverse,
	xticklabel style={/pgf/number format/fixed},
	/pgf/number format/1000 sep={},
    legend pos=outer north east, legend entries={MC, ILIS, BLISS}]
	\addplot+[black, densely dotted, thick, mark = o, mark size = 2pt, mark options=solid]	table[x index=10,y index=12, col sep=comma]{figures/ex2VGexpC.csv};
	\addplot+[blue, thick, dashed, mark = square, mark size = 2pt, mark options=solid]	table[x index=10,y index=15, col sep=comma]{figures/ex2VGexpC.csv};
	\addplot+[red, thick, solid, mark = star, mark options=solid]	table[x index=10,y index=18, col sep=comma]{figures/ex2VGexpC.csv};
	\end{semilogyaxis}
	\begin{semilogyaxis}[name=plot4,height=1.8in,width=0.315\textwidth,at={($(plot1.south)-(0in,0.2in)$)},anchor=north,
	xlabel={$\delta$},
	ylabel={Efficiency ratio},
	every tick label/.append style={font=\scriptsize},
	axis on top,
 	scaled x ticks = false,    x dir=reverse,
	xticklabel style={/pgf/number format/fixed},
	/pgf/number format/1000 sep={}]
	\addplot+[black, densely dotted, thick, mark = o, mark size = 2pt, mark options=solid]	table[x index=10,y index=13, col sep=comma]{figures/ex2GexpC.csv};
	\addplot+[blue, thick, dashed, mark = square, mark size = 2pt, mark options=solid]	table[x index=10,y index=16, col sep=comma]{figures/ex2GexpC.csv};
	\addplot+[red, thick, solid, mark = star, mark options=solid]	table[x index=10,y index=19, col sep=comma]{figures/ex2GexpC.csv};
	\end{semilogyaxis}
	\begin{semilogyaxis}[name=plot5,height=1.8in,width=0.315\textwidth,at={($(plot4.east)+(0.4in,0in)$)},anchor=west,
	xlabel={$\delta$},
	every tick label/.append style={font=\scriptsize},
	axis on top,
 	scaled x ticks = false,    x dir=reverse,
	xticklabel style={/pgf/number format/fixed},
	/pgf/number format/1000 sep={}]
	\addplot+[black, densely dotted, thick, mark = o, mark size = 2pt, mark options=solid]	table[x index=10,y index=13, col sep=comma]{figures/ex2TexpC.csv};
	\addplot+[blue, thick, dashed, mark = square, mark size = 2pt, mark options=solid]	table[x index=10,y index=16, col sep=comma]{figures/ex2TexpC.csv};
	\addplot+[red, thick, solid, mark = star, mark options=solid]	table[x index=10,y index=19, col sep=comma]{figures/ex2TexpC.csv};
	\end{semilogyaxis}
	\begin{semilogyaxis}[name=plot6,height=1.8in,width=0.315\textwidth,at={($(plot5.east)+(0.4in,0in)$)},anchor=west,
	xlabel={$\delta$},
	every tick label/.append style={font=\scriptsize},
	axis on top,
 	scaled x ticks = false,    x dir=reverse,
	xticklabel style={/pgf/number format/fixed},
	/pgf/number format/1000 sep={}]
	\addplot+[black, densely dotted, thick, mark = o, mark size = 2pt, mark options=solid]	table[x index=10,y index=13, col sep=comma]{figures/ex2VGexpC.csv};
	\addplot+[blue, thick, dashed, mark = square, mark size = 2pt, mark options=solid]	table[x index=10,y index=16, col sep=comma]{figures/ex2VGexpC.csv};
	\addplot+[red, thick, solid, mark = star, mark options=solid]	table[x index=10,y index=19, col sep=comma]{figures/ex2VGexpC.csv};
	\end{semilogyaxis}
    \node[yshift=29pt] 
    at ($(plot1.north west)!0.5!(plot3.north east)$) 
    {(a) Non-diagonal-loading case};
    
	\begin{semilogyaxis}[name=plot7,height=1.8in,width=0.315\textwidth,at={($(plot4.south)-(0in,1in)$)},anchor=north,
    title = {Gaussian},
	ylabel={Standard error},
	every tick label/.append style={font=\scriptsize},
	axis on top,
 	scaled x ticks = false,    x dir=reverse,
	xticklabel style={/pgf/number format/fixed},
	/pgf/number format/1000 sep={}]
	\addplot+[black, densely dotted, thick, mark = o, mark size = 2pt, mark options=solid]	table[x index=10,y index=12, col sep=comma]{figures/ex2GexpD.csv};
	\addplot+[blue, thick, dashed, mark = square, mark size = 2pt, mark options=solid]	table[x index=10,y index=15, col sep=comma]{figures/ex2GexpD.csv};
	\addplot+[red, thick, solid, mark = star, mark options=solid]	table[x index=10,y index=18, col sep=comma]{figures/ex2GexpD.csv};
	\end{semilogyaxis}
	\begin{semilogyaxis}[name=plot8,height=1.8in,width=0.315\textwidth,at={($(plot7.east)+(0.4in,0in)$)},anchor=west,
    title = {$t$},
	every tick label/.append style={font=\scriptsize},
	axis on top,
 	scaled x ticks = false,    x dir=reverse,
	xticklabel style={/pgf/number format/fixed},
	/pgf/number format/1000 sep={}]
	\addplot+[black, densely dotted, thick, mark = o, mark size = 2pt, mark options=solid]	table[x index=10,y index=12, col sep=comma]{figures/ex2TexpD.csv};
	\addplot+[blue, thick, dashed, mark = square, mark size = 2pt, mark options=solid]	table[x index=10,y index=15, col sep=comma]{figures/ex2TexpD.csv};
	\addplot+[red, thick, solid, mark = star, mark options=solid]	table[x index=10,y index=18, col sep=comma]{figures/ex2TexpD.csv};
	\end{semilogyaxis}
	\begin{semilogyaxis}[name=plot9,height=1.8in,width=0.315\textwidth,at={($(plot8.east)+(0.4in,0in)$)},anchor=west,
    title = {Variance-gamma},
	every tick label/.append style={font=\scriptsize},
	axis on top,
 	scaled x ticks = false,    x dir=reverse,
	xticklabel style={/pgf/number format/fixed},
	/pgf/number format/1000 sep={},
    legend pos=outer north east, legend entries={MC, ILIS, BLISS}]
	\addplot+[black, densely dotted, thick, mark = o, mark size = 2pt, mark options=solid]	table[x index=10,y index=12, col sep=comma]{figures/ex2VGexpD.csv};
	\addplot+[blue, thick, dashed, mark = square, mark size = 2pt, mark options=solid]	table[x index=10,y index=15, col sep=comma]{figures/ex2VGexpD.csv};
	\addplot+[red, thick, solid, mark = star, mark options=solid]	table[x index=10,y index=18, col sep=comma]{figures/ex2VGexpD.csv};
	\end{semilogyaxis}
	\begin{semilogyaxis}[name=plot10,height=1.8in,width=0.315\textwidth,at={($(plot7.south)-(0in,0.2in)$)},anchor=north,
	xlabel={$\delta$},
	ylabel={Efficiency ratio},
	every tick label/.append style={font=\scriptsize},
	axis on top,
 	scaled x ticks = false,    x dir=reverse,
	xticklabel style={/pgf/number format/fixed},
	/pgf/number format/1000 sep={}]
	\addplot+[black, densely dotted, thick, mark = o, mark size = 2pt, mark options=solid]	table[x index=10,y index=13, col sep=comma]{figures/ex2GexpD.csv};
	\addplot+[blue, thick, dashed, mark = square, mark size = 2pt, mark options=solid]	table[x index=10,y index=16, col sep=comma]{figures/ex2GexpD.csv};
	\addplot+[red, thick, solid, mark = star, mark options=solid]	table[x index=10,y index=19, col sep=comma]{figures/ex2GexpD.csv};
	\end{semilogyaxis}
	\begin{semilogyaxis}[name=plot11,height=1.8in,width=0.315\textwidth,at={($(plot10.east)+(0.4in,0in)$)},anchor=west,
	xlabel={$\delta$},
	every tick label/.append style={font=\scriptsize},
	axis on top,
 	scaled x ticks = false,    x dir=reverse,
	xticklabel style={/pgf/number format/fixed},
	/pgf/number format/1000 sep={}]
	\addplot+[black, densely dotted, thick, mark = o, mark size = 2pt, mark options=solid]	table[x index=10,y index=13, col sep=comma]{figures/ex2TexpD.csv};
	\addplot+[blue, thick, dashed, mark = square, mark size = 2pt, mark options=solid]	table[x index=10,y index=16, col sep=comma]{figures/ex2TexpD.csv};
	\addplot+[red, thick, solid, mark = star, mark options=solid]	table[x index=10,y index=19, col sep=comma]{figures/ex2TexpD.csv};
	\end{semilogyaxis}
	\begin{semilogyaxis}[name=plot12,height=1.8in,width=0.315\textwidth,at={($(plot11.east)+(0.4in,0in)$)},anchor=west,
	xlabel={$\delta$},
	every tick label/.append style={font=\scriptsize},
	axis on top,
 	scaled x ticks = false,    x dir=reverse,
	xticklabel style={/pgf/number format/fixed},
	/pgf/number format/1000 sep={}]
	\addplot+[black, densely dotted, thick, mark = o, mark size = 2pt, mark options=solid]	table[x index=10,y index=13, col sep=comma]{figures/ex2VGexpD.csv};
	\addplot+[blue, thick, dashed, mark = square, mark size = 2pt, mark options=solid]	table[x index=10,y index=16, col sep=comma]{figures/ex2VGexpD.csv};
	\addplot+[red, thick, solid, mark = star, mark options=solid]	table[x index=10,y index=19, col sep=comma]{figures/ex2VGexpD.csv};
	\end{semilogyaxis}
    \node[yshift=29pt] 
    at ($(plot7.north west)!0.5!(plot9.north east)$) 
    {(b) Diagonal-loading case};
	\end{tikzpicture}
    \vspace{-12pt}
    \edit{\caption{\label{fig:delta} Performance comparison between MC, ILIS, and BLISS for the EBA example in Section~\ref{subsec:numeric_2} across different values of the bankruptcy-cost factor $(\delta)$ with the exponential inverse demand function. The three columns present the results with Gaussian, $t$, and variance-gamma log-return models, respectively. The upper six panels correspond to the \edit{non-diagonal-loading} case, while the lower six panels describe the diagonal-loading case. In each of the two cases, the top three panels compare standard errors of the three methods, and the bottom three panels exhibit their efficiency ratios.}}
\end{figure}

\begin{figure}[!t]
\centering
	\begin{tikzpicture}[font=\sffamily\small]
	\begin{semilogyaxis}[name=plot1,height=1.8in,width=0.315\textwidth,
    title = {Gaussian},
	ylabel={Standard error},
    xtick={0,1,2,3},
    xticklabels={Beta,Uniform,Priority},
	every tick label/.append style={font=\scriptsize},
	axis on top]
	\addplot+[black, only marks, thick, mark = o, mark size = 2pt, mark options=solid]	table[x expr=\coordindex,y index=2, col sep=comma, skip first n=1]{figures/ex2GexpC2.csv};
	\addplot+[blue, only marks, thick, mark = square, mark size = 2pt, mark options=solid]	table[x expr=\coordindex,y index=5, col sep=comma, skip first n=1]{figures/ex2GexpC2.csv};
	\addplot+[red, only marks, thick, mark = star, mark options=solid]	table[x expr=\coordindex,y index=8, col sep=comma, skip first n=1]{figures/ex2GexpC2.csv};
	\end{semilogyaxis}
	\begin{semilogyaxis}[name=plot2,height=1.8in,width=0.315\textwidth,at={($(plot1.east)+(0.4in,0in)$)},anchor=west,
    title = {$t$},
    xtick={0,1,2,3},
    xticklabels={Beta,Uniform,Priority},
	every tick label/.append style={font=\scriptsize},
	axis on top]
	\addplot+[black, only marks, thick, mark = o, mark size = 2pt, mark options=solid]	table[x expr=\coordindex,y index=2, col sep=comma, skip first n=1]{figures/ex2TexpC2.csv};
	\addplot+[blue, only marks, thick, mark = square, mark size = 2pt, mark options=solid]	table[x expr=\coordindex,y index=5, col sep=comma, skip first n=1]{figures/ex2TexpC2.csv};
	\addplot+[red, only marks, thick, mark = star, mark options=solid]	table[x expr=\coordindex,y index=8, col sep=comma, skip first n=1]{figures/ex2TexpC2.csv};
	\end{semilogyaxis}
	\begin{semilogyaxis}[name=plot3,height=1.8in,width=0.315\textwidth,at={($(plot2.east)+(0.4in,0in)$)},anchor=west,
    title = {Variance-gamma},
    xtick={0,1,2,3},
    xticklabels={Beta,Uniform,Priority},
	every tick label/.append style={font=\scriptsize},
	axis on top,
    legend pos=outer north east, legend entries={MC, ILIS, BLISS}]
	\addplot+[black, only marks, thick, mark = o, mark size = 2pt, mark options=solid]	table[x expr=\coordindex,y index=2, col sep=comma, skip first n=1]{figures/ex2VGexpC2.csv};
	\addplot+[blue, only marks, thick, mark = square, mark size = 2pt, mark options=solid]	table[x expr=\coordindex,y index=5, col sep=comma, skip first n=1]{figures/ex2VGexpC2.csv};
	\addplot+[red, only marks, thick, mark = star, mark options=solid]	table[x expr=\coordindex,y index=8, col sep=comma, skip first n=1]{figures/ex2VGexpC2.csv};
	\end{semilogyaxis}
	\begin{semilogyaxis}[name=plot4,height=1.8in,width=0.315\textwidth,at={($(plot1.south)-(0in,0.2in)$)},anchor=north,
	ylabel={Efficiency ratio},
    xtick={0,1,2,3},
    xticklabels={Beta,Uniform,Priority},
	every tick label/.append style={font=\scriptsize},
	axis on top]
	\addplot+[black, only marks, thick, mark = o, mark size = 2pt, mark options=solid]	table[x expr=\coordindex,y index=3, col sep=comma, skip first n=1]{figures/ex2GexpC2.csv};
	\addplot+[blue, only marks, thick, mark = square, mark size = 2pt, mark options=solid]	table[x expr=\coordindex,y index=6, col sep=comma, skip first n=1]{figures/ex2GexpC2.csv};
	\addplot+[red, only marks, thick, mark = star, mark options=solid]	table[x expr=\coordindex,y index=9, col sep=comma, skip first n=1]{figures/ex2GexpC2.csv};
	\end{semilogyaxis}
	\begin{semilogyaxis}[name=plot5,height=1.8in,width=0.315\textwidth,at={($(plot4.east)+(0.4in,0in)$)},anchor=west,
    xtick={0,1,2,3},
    xticklabels={Beta,Uniform,Priority},
	every tick label/.append style={font=\scriptsize},
	axis on top]
	\addplot+[black, only marks, thick, mark = o, mark size = 2pt, mark options=solid]	table[x expr=\coordindex,y index=3, col sep=comma, skip first n=1]{figures/ex2TexpC2.csv};
	\addplot+[blue, only marks, thick, mark = square, mark size = 2pt, mark options=solid]	table[x expr=\coordindex,y index=6, col sep=comma, skip first n=1]{figures/ex2TexpC2.csv};
	\addplot+[red, only marks, thick, mark = star, mark options=solid]	table[x expr=\coordindex,y index=9, col sep=comma, skip first n=1]{figures/ex2TexpC2.csv};
	\end{semilogyaxis}
	\begin{semilogyaxis}[name=plot6,height=1.8in,width=0.315\textwidth,at={($(plot5.east)+(0.4in,0in)$)},anchor=west,
    xtick={0,1,2,3},
    xticklabels={Beta,Uniform,Priority},
	every tick label/.append style={font=\scriptsize},
	axis on top]
	\addplot+[black, only marks, thick, mark = o, mark size = 2pt, mark options=solid]	table[x expr=\coordindex,y index=3, col sep=comma, skip first n=1]{figures/ex2VGexpC2.csv};
	\addplot+[blue, only marks, thick, mark = square, mark size = 2pt, mark options=solid]	table[x expr=\coordindex,y index=6, col sep=comma, skip first n=1]{figures/ex2VGexpC2.csv};
	\addplot+[red, only marks, thick, mark = star, mark options=solid]	table[x expr=\coordindex,y index=9, col sep=comma, skip first n=1]{figures/ex2VGexpC2.csv};
	\end{semilogyaxis}
    \node[yshift=29pt] 
    at ($(plot1.north west)!0.5!(plot3.north east)$) 
    {(a) Non-diagonal-loading case};
    
	\begin{semilogyaxis}[name=plot7,height=1.8in,width=0.315\textwidth,at={($(plot4.south)-(0in,1in)$)},anchor=north,
    title = {Gaussian},
	ylabel={Standard error},
    xtick={0,1,2,3},
    xticklabels={Beta,Uniform,Priority},
	every tick label/.append style={font=\scriptsize},
	axis on top]
	\addplot+[black, only marks, thick, mark = o, mark size = 2pt, mark options=solid]	table[x expr=\coordindex,y index=2, col sep=comma, skip first n=1]{figures/ex2GexpD2.csv};
	\addplot+[blue, only marks, thick, mark = square, mark size = 2pt, mark options=solid]	table[x expr=\coordindex,y index=5, col sep=comma, skip first n=1]{figures/ex2GexpD2.csv};
	\addplot+[red, only marks, thick, mark = star, mark options=solid]	table[x expr=\coordindex,y index=8, col sep=comma, skip first n=1]{figures/ex2GexpD2.csv};
	\end{semilogyaxis}
	\begin{semilogyaxis}[name=plot8,height=1.8in,width=0.315\textwidth,at={($(plot7.east)+(0.4in,0in)$)},anchor=west,
    title = {$t$},
    xtick={0,1,2,3},
    xticklabels={Beta,Uniform,Priority},
	every tick label/.append style={font=\scriptsize},
	axis on top]
	\addplot+[black, only marks, thick, mark = o, mark size = 2pt, mark options=solid]	table[x expr=\coordindex,y index=2, col sep=comma, skip first n=1]{figures/ex2TexpD2.csv};
	\addplot+[blue, only marks, thick, mark = square, mark size = 2pt, mark options=solid]	table[x expr=\coordindex,y index=5, col sep=comma, skip first n=1]{figures/ex2TexpD2.csv};
	\addplot+[red, only marks, thick, mark = star, mark options=solid]	table[x expr=\coordindex,y index=8, col sep=comma, skip first n=1]{figures/ex2TexpD2.csv};
	\end{semilogyaxis}
	\begin{semilogyaxis}[name=plot9,height=1.8in,width=0.315\textwidth,at={($(plot8.east)+(0.4in,0in)$)},anchor=west,
    title = {Variance-gamma},
    xtick={0,1,2,3},
    xticklabels={Beta,Uniform,Priority},
	every tick label/.append style={font=\scriptsize},
	axis on top,
    legend pos=outer north east, legend entries={MC, ILIS, BLISS}]
	\addplot+[black, only marks, thick, mark = o, mark size = 2pt, mark options=solid]	table[x expr=\coordindex,y index=2, col sep=comma, skip first n=1]{figures/ex2VGexpD2.csv};
	\addplot+[blue, only marks, thick, mark = square, mark size = 2pt, mark options=solid]	table[x expr=\coordindex,y index=5, col sep=comma, skip first n=1]{figures/ex2VGexpD2.csv};
	\addplot+[red, only marks, thick, mark = star, mark options=solid]	table[x expr=\coordindex,y index=8, col sep=comma, skip first n=1]{figures/ex2VGexpD2.csv};
	\end{semilogyaxis}
	\begin{semilogyaxis}[name=plot10,height=1.8in,width=0.315\textwidth,at={($(plot7.south)-(0in,0.2in)$)},anchor=north,
	ylabel={Efficiency ratio},
    xtick={0,1,2,3},
    xticklabels={Beta,Uniform,Priority},
	every tick label/.append style={font=\scriptsize},
	axis on top]
	\addplot+[black, only marks, thick, mark = o, mark size = 2pt, mark options=solid]	table[x expr=\coordindex,y index=3, col sep=comma, skip first n=1]{figures/ex2GexpD2.csv};
	\addplot+[blue, only marks, thick, mark = square, mark size = 2pt, mark options=solid]	table[x expr=\coordindex,y index=6, col sep=comma, skip first n=1]{figures/ex2GexpD2.csv};
	\addplot+[red, only marks, thick, mark = star, mark options=solid]	table[x expr=\coordindex,y index=9, col sep=comma, skip first n=1]{figures/ex2GexpD2.csv};
	\end{semilogyaxis}
	\begin{semilogyaxis}[name=plot11,height=1.8in,width=0.315\textwidth,at={($(plot10.east)+(0.4in,0in)$)},anchor=west,
    xtick={0,1,2,3},
    xticklabels={Beta,Uniform,Priority},
	every tick label/.append style={font=\scriptsize},
	axis on top]
	\addplot+[black, only marks, thick, mark = o, mark size = 2pt, mark options=solid]	table[x expr=\coordindex,y index=3, col sep=comma, skip first n=1]{figures/ex2TexpD2.csv};
	\addplot+[blue, only marks, thick, mark = square, mark size = 2pt, mark options=solid]	table[x expr=\coordindex,y index=6, col sep=comma, skip first n=1]{figures/ex2TexpD2.csv};
	\addplot+[red, only marks, thick, mark = star, mark options=solid]	table[x expr=\coordindex,y index=9, col sep=comma, skip first n=1]{figures/ex2TexpD2.csv};
	\end{semilogyaxis}
	\begin{semilogyaxis}[name=plot12,height=1.8in,width=0.315\textwidth,at={($(plot11.east)+(0.4in,0in)$)},anchor=west,
    xtick={0,1,2,3},
    xticklabels={Beta,Uniform,Priority},
	every tick label/.append style={font=\scriptsize},
	axis on top]
	\addplot+[black, only marks, thick, mark = o, mark size = 2pt, mark options=solid]	table[x expr=\coordindex,y index=3, col sep=comma, skip first n=1]{figures/ex2VGexpD2.csv};
	\addplot+[blue, only marks, thick, mark = square, mark size = 2pt, mark options=solid]	table[x expr=\coordindex,y index=6, col sep=comma, skip first n=1]{figures/ex2VGexpD2.csv};
	\addplot+[red, only marks, thick, mark = star, mark options=solid]	table[x expr=\coordindex,y index=9, col sep=comma, skip first n=1]{figures/ex2VGexpD2.csv};
	\end{semilogyaxis}
    \node[yshift=29pt] 
    at ($(plot7.north west)!0.5!(plot9.north east)$) 
    {(b) Diagonal-loading case};
	\end{tikzpicture}
    \edit{\caption{\label{fig:scenario} Performance comparison between MC, ILIS, and BLISS for the EBA example in Section~\ref{subsec:numeric_2} across different random price impact scenarios and external-liability priority scenarios with the exponential inverse demand function. The three columns present the results with Gaussian, $t$, and variance-gamma log-return models, respectively. The upper six panels correspond to the \edit{non-diagonal-loading} case, while the lower six panels describe the diagonal-loading case. In each of the two cases, the top three panels compare standard errors of the three methods, and the bottom three panels exhibit their efficiency ratios.}}
\end{figure}
\end{revision}

\section{Conclusion}\label{sec:conclusion}
In this paper, we have developed an efficient simulation framework for pricing corporate bonds of institutions, which account for network-induced credit risk. The valuation of corporate securities within these complex, interconnected systems is challenging because of the lack of closed-form solutions and the inefficiency of standard simulation methods in capturing rare default events.

To overcome these challenges, we proposed a novel computational framework---Bi-Level Importance Sampling with Splitting (BLISS)---to estimate bond prices in a network setting with simultaneous clearing of interbank liabilities which additionally accounts for fire sales. Specifically, we introduced a judiciously designed fictitious system that helps detach the target institution from the intertwined system, enabling us to reformulate its default condition in a tractable way. This allowed us to decompose the valuation problem into a tractable two-layer sampling process and \edit{combine exact outer-layer conditioning on a conservative event containing the target default event with} optimal importance sampling for the inner-layer sampling to achieve significant variance reduction.

\edit{Our theoretical analysis considers network sequences along which the target institution becomes increasingly stable. We establish asymptotic optimality for Gaussian log-returns and identify an additional tail-comparability condition for bounded relative error. Furthermore, we verify strong performance guarantees under normal mean-variance mixtures, including Student’s $t$ and variance-gamma models. We quantify the asymptotic contribution of outer-layer sampling through second-moment comparisons. 
Numerical experiments complement this analysis; synthetic examples examine performance as network size increases, while experiments based on European Banking Authority data demonstrate robustness against model variations.}

Several promising directions for future research emerge from this work. First, the BLISS framework could be extended to value more complex securities, including credit derivatives, contingent convertible bonds, or collateralized debt obligations, where network effects play an increasingly important role. Second, incorporating dynamic network structures, where interbank exposures evolve over time, would enhance the realism of the model and enable term-structure pricing of corporate bonds. Third, the methodology could be adapted to stress testing and regulatory capital assessment, where financial institutions must evaluate tail-risk scenarios across large portfolios efficiently.

\appendix
\setcounter{figure}{0}
\setcounter{table}{0}
\renewcommand{\thefigure}{A\arabic{figure}}
\renewcommand{\thetable}{A\arabic{table}}

\section{Proofs of the Main Theoretical Results}\label{app:proofs}
\begin{proof}[Proof of Proposition~\ref{prop:fictitious}]
For any $\bs,\bs'\in\bbR_+^n$ satisfying $\bs\leq\bs'$, let
$$
\left\{~\begin{aligned}
    &q^1 = Q\left(\sum_{i=1}^{n}\lt\{\frac{\big(\bar p_i -  s_i'- \sum_j \pi_{ji} p_j(\bs)\big)^+}{q(\bs)} \wedge e_i \rt\}\right),		\\
    &\bp^1 = \bar \bp \wedge \left(\bs'+ q(\bs) \be +  \bPi^\top\bp(\bs)\right).
    \end{aligned}\right.
 $$
and
$$
\left\{~\begin{aligned}
    &q^{k+1} = Q\left(\sum_{i=1}^{n}\lt\{\frac{\big(\bar p_i -  s_i'- \sum_j \pi_{ji} p_j^k\big)^+}{q^k} \wedge e_i \rt\}\right),		\\
    &\bp^{k+1} = \bar \bp \wedge \left(\bs'+ q^k \be +  \bPi^\top\bp^k\right)
    \end{aligned}\right.
$$
for all $k\geq1$. Then, by Assumption~\ref{ass:inverse_demand}\ref{as:Q-decreasing}, it is easy to check that $q(\bs)\leq q^1\leq q^2\leq \cdots$ and $\bp(\bs)\leq\bp^1\leq\bp^2\leq\cdots$. Furthermore, we have $q^k\to q(\bs')$ and $\bp^k\to\bp(\bs')$ as $k\to\infty$. Thus, we have
\begin{equation}\label{eq:qp-inc}
q(\bs)\leq q(\bs')~~\text{and}~~\bp(\bs)\leq\bp(\bs')~~\text{for any}~\bs,\bs'\in\bbR_+^n~\text{satisfying}~\bs\leq\bs'.
\end{equation}


Fix $\tilde s_1,\ldots,\tilde s_{n-1}\in\bbR_+$.
Define a real-valued function $\edit{\varrho}(\cdot)$ by
$$
\begin{aligned}
    \edit{\varrho}(s_n) = s_n +q(\tilde s_1,\ldots,\tilde s_{n-1},s_n)e_n + \sum_{j=1}^{n-1} \pi_{jn} p_j(\tilde s_1,\ldots,\tilde s_{n-1},s_n)-\bar p_n,
\end{aligned}
$$
which is strictly increasing by \eqref{eq:qp-inc}. \edit{For any sequence
${\bs}^{(k)}\to \bs$, the clearing solutions belong to the compact set
$
 [Q(\bar e),Q(0)]\times\prod_{i=1}^{n}[0,\bar p_i].
$
Every convergent subsequence has a limit satisfying the clearing equations
at $\bs$, by continuity of the fixed-point map and the strict lower bound
$Q(\bar e)>0$. By uniqueness, that limit equals $(q(\bs),\bp(\bs))$.
Consequently, the entire sequence converges to $(q(\bs),\bp(\bs))$, and thus, the function $\edit{\varrho}$ is continuous.}

Observe that $\edit{\varrho}(s_n)<0$ for all $s_n$ small enough since we have $\bar p_n>Q(0)e_n+\sum_{j=1}^{n-1} \pi_{jn} \bar p_j$, whereas $\edit{\varrho}(s_n)>0$ for all $s_n$ large enough. Hence, it is straightforward that there exists a unique value $s_n^*$ satisfying $\edit{\varrho}(s_n^*)=0$.
Then, we have
$$
\frac{\bar p_n - s_n^*- \sum_j \pi_{jn} p_j(\tilde s_1,\ldots,\tilde s_{n-1},s_n^*)}{q(\tilde s_1,\ldots,\tilde s_{n-1},s_n^*)} = e_n~~\text{and}~~p_n(\tilde s_1,\ldots,\tilde s_{n-1},s_n^*)=\bar p_n.
$$
This suggests that $q(\tilde s_1,\ldots,\tilde s_{n-1},s_n^*)=\tilde q(\tilde s_1,\ldots,\tilde s_{n-1})$ and $p_i(\tilde s_1,\ldots,\tilde s_{n-1},s_n^*)=\tilde p_i(\tilde s_1,\ldots,\tilde s_{n-1})$ for $i=1,2,\ldots,n-1$ by the uniqueness of $q(\cdot)$, $\tilde q(\cdot)$, $\bp(\cdot)$, and $\tilde\bp(\cdot)$. Therefore, $s_n^*=v_n(\tilde s_1,\ldots,\tilde s_{n-1})$.
Accordingly, we have
\begin{equation}\label{eq:set-equiv}
\begin{aligned}
    p_n(\tilde s_1,\ldots,\tilde s_{n-1}, s_n)<\bar p_n
    ~~&\Leftrightarrow~~\edit{\varrho}(s_n)<0\\
    &\Leftrightarrow~~s_n<s_n^*=v_n(\tilde s_1,\ldots,\tilde s_{n-1}).
\end{aligned}
\end{equation}
This completes the proof since $\tilde s_1,\ldots,\tilde s_{n-1}$ are chosen arbitrarily.
\end{proof}

\begin{revision}
\begin{proof}[Proof of Proposition~\ref{prop:new_outer_sampler}]
Let
$$
\bu=\frac{1}{\sigma_n}(\Lambda_{n1},\ldots,\Lambda_{nn})^\top,
$$
so that $\|\bu\|=1$. For $\bxi\sim N(0,\bI_n)$, write
$R_0=\bu^\top\bxi$ and $\bU=(\bI_n-\bu\bu^\top)\bxi$.
The variables $R_0$ and $\bU$ are jointly Gaussian and 
independent, with $R_0\sim N(0,1)$ and
$\bU\sim N(0,\bI_n-\bu\bu^\top)$.
Thus, the condition $\bu^\top\bxi<-B_n^\tc/\sigma_n$ is equivalent to $R_0<-B_n^\tc/\sigma_n$. Since $\bxi=\bu R_0+\bU$, the random variable $(\bxi\mid\bu^\top\bxi<-B_n^\tc/\sigma_n)$ has the same distribution with
$\bu Y+\bU$, where $Y$ is a standard normal variable truncated to this interval
and is independent of $\bU$. Its first $n-1$ coordinates are 
$$
 \bxi_{-n}+\left(\frac{Y}{\sigma_n}
       -\frac{\sum_{i=1}^{n}\Lambda_{ni}\xi_i}{\sigma_n^2}\right)\blambda_n,
$$
which is exactly the return value in the proposition. Consequently, applying Bayes' rule results in the density in~\eqref{eq:new_outer_density}.
\end{proof}



Throughout the following proofs, all relevant quantities refer to the $n$th network;  $v_n^\tc=v_{n,n}^\tc$, $D_n^\tc=D_{n,n}^\tc$, $d_n^\tc=d_{n,n}^\tc$, and $\sigma_n=\sigma_{n,n}$. The following two lemmas are used for the proof of Theorem~\ref{thm:gaussian_optimality}.
\begin{lemma}\label{lem:gaussian_target_rate}
Under the system-wide stability regime in the sequence $\{\mathfrak{N}_n\}$ of networks constructed above,
\begin{equation}\label{eq:gaussian_gamma_rate}
 \log\gamma_n\sim-\frac{(D_n^\tc)^2}{2}~~\text{as}~n\to\infty.
\end{equation}
\end{lemma}
\begin{proof}[Proof of Lemma~\ref{lem:gaussian_target_rate}]
Let $h_L,h_U>0$ denote the uniform lower and upper bounds of $h$, respectively. It is straightforward to check that
\begin{equation}\label{eq:gaussian_probability_sandwich_appendix}
 h_L\bar\Phi(D_n^\f)\le\gamma_n\le h_U\bar\Phi(D_n^\tc).
\end{equation}
The stability conditions imply $D_n^\tc\to\infty$ and
$D_n^\f/D_n^\tc\to1$. Furthermore, for $x>0$, the standard Mills bounds
$$
 \frac{x}{1+x^2}\phi_1(x)\le\bar\Phi(x)\le\frac{\phi_1(x)}{x}
$$
imply $\log\bar\Phi(x)=-x^2/2+O(\log x)$ as $x\to\infty$.  Applying this to the logarithm of both endpoints of~\eqref{eq:gaussian_probability_sandwich_appendix} proves~\eqref{eq:gaussian_gamma_rate}.
\end{proof}

\begin{lemma}\label{lem:bivariate_gaussian_tail}
Let $(X_1,X_2)$ be bivariate standard normal with correlation
$\rho\in[0,1)$. As $x\to\infty$,
$$
 \log\sP(X_1>x,X_2>x)
    =-\frac{x^2}{1+\rho}+O(\log x),
$$
where the constant implicit in $O(\log x)$ is independent of $\rho$.
In particular, the result applies to arbitrary sequences
$x_n\to\infty$ and $\rho_n\in[0,1)$, including $\rho_n\to1$.
\end{lemma}
\begin{proof}
Define independent standard normal variables
$$
 U=\frac{X_1+X_2}{\sqrt{2(1+\rho)}}\quad\text{and}\quad
 W=\frac{X_1-X_2}{\sqrt{2(1-\rho)}},
$$
and put $a_\rho=\sqrt{2/(1+\rho)}$ and
$b_\rho=\sqrt{(1-\rho)/(1+\rho)}$. Then,
$$
 \{X_1>x,X_2>x\}
   =\{U>a_\rho x+b_\rho|W|\}.
$$
For $x\ge1$, independence and $0\le b_\rho\le1$ yield
$$
 \sP(|W|\le x^{-1})\,
 \bar\Phi(a_\rho x+b_\rho/x)
 \le\sP(X_1>x,X_2>x)
 \le\bar\Phi(a_\rho x).
$$
Since $\phi_1$ is decreasing on $[0,\infty)$,
$\sP(|W|\le x^{-1})\ge 2\phi_1(1)/x$.
Also,  $1\le a_\rho\le\sqrt2$, and
$$
 (a_\rho x+b_\rho/x)^2/2
 =\frac{x^2}{1+\rho}+a_\rho b_\rho
          +\frac{b_\rho^2}{2x^2}.
$$
The last two terms are bounded uniformly in $\rho$ and $x\ge1$.
Applying the Mills bounds to both normal tails results in constants
$c_1,c_2>0$, independent of $x\ge1$ and $\rho\in[0,1)$, such that
$$
 c_1 x^{-2}\exp\!\left(-\frac{x^2}{1+\rho}\right)
 \le\sP(X_1>x,X_2>x)
 \le c_2 x^{-1}\exp\!\left(-\frac{x^2}{1+\rho}\right).
$$
Taking logarithms proves the lemma.
\end{proof}

\begin{proof}[Proof of Theorem~\ref{thm:gaussian_optimality}]
For the $n$th network, define
$$
 B_n(\bz)=\log\frac{M_n}{v_n(\bS_{-n}(\bz))}\quad\text{and}\quad\bar\Phi(\ell_n(\bz))=\bar\Phi\!\left(\frac{B_n(\bz)+\blambda_n^\top\bz}{\Lambda_{nn}}\right),
$$
and define $\kappa_n^\tc(\bz)$ as in~\eqref{eq:gauss_conservative_prob}.  Let
$$
 \kappa_n^\f(\bz)
 \coloneqq
 \bar\Phi\!\left(\frac{B_n^\f+\blambda_n^\top\bz}{\Lambda_{nn}}\right).
$$
Then, we have
 $0<\kappa_n^\f(\bz)\le \bar\Phi(\ell_n(\bz))\le \kappa_n^\tc(\bz)$, and its unconditional value is $\kappa_n^\f=\bar\Phi(D_n^\f)$.

The outer and inner proposal densities are~\eqref{eq:new_outer_density} and~\eqref{eq:psi}.  Since the integrand below is nonnegative, Tonelli's theorem implies that
\begin{align*}
 \sE_n^*[\Gamma_n^*]
 &=\int_{\bbR^{n-1}}\int_{\bbR}
 \phi_{n-1}(\bz)\edit{\phi_1}(u)h_n(\bS(\bz,u))
 \sone_{\{u<-\ell_n(\bz)\}}\,\rd u\,\rd\bz
 =\gamma_n,
\end{align*}
where $\bS(\bz,u)$ denotes the asset values $S_1,\ldots,S_n$ in \eqref{eq:external assets} when $\bZ_{-n}=\bz$ and $Z_n=u$.
Furthermore, since
 $0\le\Gamma_n^*\le h_U\kappa_n^\tc$, we have  $(\Gamma_n^*)^2\le h_U\kappa_n^\tc\Gamma_n^*$, and thus, 
\begin{equation}\label{eq:gaussian_master_second_moment}
 \gamma_n^2
 \le\sE_n^*[(\Gamma_n^*)^2]
 \le h_U\kappa_n^\tc\gamma_n,
\end{equation}
where the lower bound is Jensen's inequality.  By Lemma~\ref{lem:gaussian_target_rate}
 $\log\gamma_n\sim\log \kappa_n^\tc\sim\log \kappa_n^\f
 \sim-{(D_n^\tc)^2}/{2}$.
The logarithms of both bounds in~\eqref{eq:gaussian_master_second_moment} are therefore $-(D_n^\tc)^2+o((D_n^\tc)^2)$, which proves
\begin{equation}\label{eq:gaussian_optimal_second_moment}
 \log\sE_n^*[(\Gamma_n^*)^2]\sim-(D_n^\tc)^2~~\text{as}~n\to\infty.
\end{equation}
Hence, the asymptotic optimality of $\Gamma_n^*$ holds by \eqref{eq:gaussian_gamma_rate}.
Finally, $\gamma_n\ge h_L\kappa_n^\f$ and~\eqref{eq:gaussian_master_second_moment} imply
$$
 \frac{\sE_n^*[(\Gamma_n^*)^2]}{\gamma_n^2}
 \le\frac{h_U}{h_L}\frac{\kappa_n^\tc}{\kappa_n^\f}.
$$
Taking the limit superior and using~\eqref{eq:gaussian_tail_comparability} completes the proof.
\end{proof}

\begin{proof}[Proof of Proposition~\ref{prop:superiority}]
For $c>0$, define
 $$J_n^0(c)
 \coloneqq
 \sE\!\left[
 \bar\Phi\!\left(\frac{c+\blambda_n^\top\bZ}{\Lambda_{nn}}\right)^2
 \right],
 \qquad
 \bZ\sim\cN(\bzero_{n-1},\bI_{n-1}).$$
 Since $\kappa_n^\f(\bz)\le \bar\Phi(\ell_n(\bz))\le \kappa_n^\tc(\bz)$ for all $\bz\in\bR^{n-1}$, we observe that
\begin{equation}\label{eq:gaussian_no_outer_sandwich}
 h_L^2J_n^0(B_n^\f)
 \le\sE_n^0[(\Gamma_n^0)^2]
 \le h_U^2J_n^0(B_n^\tc).
\end{equation}
Let $X_1,X_2$ be independent standard normals, independent of $\bZ$, and define $Y_{j,n}\coloneqq\Lambda_{nn}X_j-\blambda_n^\top\bZ$ for $j=1,2$.
Then, $J_n^0(c)=\sP(Y_{1,n}>c,Y_{2,n}>c)$,
and $(Y_{1,n}/\sigma_n,Y_{2,n}/\sigma_n)$ is bivariate standard normal with correlation $\rho_n\coloneqq {\|\blambda_n\|^2}/{\sigma_n^2}$.  Since $D_n^\tc,D_n^\f\to\infty$, Lemma~\ref{lem:bivariate_gaussian_tail} therefore implies
$$
 \log J_n^0(B_n^\tc)
 \sim-\frac{(D_n^\tc)^2}{1+\rho_n}\quad\text{and}\quad \log J_n^0(B_n^\f)
 \sim-\frac{(D_n^\f)^2}{1+\rho_n}.
$$
Accordingly, taking logarithms in~\eqref{eq:gaussian_no_outer_sandwich} proves
$$
 \log \sE_n^0[(\Gamma_n^0)^2]
 \sim-\frac{(D_n^\tc)^2}{1+\rho_n}.
$$
Combining this with the BLISS second-moment rate in the proof of Theorem~\ref{thm:gaussian_optimality}
yields
$$
 \frac{\log\sE_n^*[(\Gamma_n^*)^2]}
      {\log\sE_n^0[(\Gamma_n^0)^2]}
       =1+\rho_n+o(1).
$$
Consequently, \eqref{eq:gaussian_comparison_ratios} follows.
\end{proof}

We use the following lemma and its proof to prove Propositions~\ref{prop:gh_sampler} and~\ref{prop:gh_left_tail}.
\begin{lemma}\label{lem:conditional_gig}

Let $T=\eta V+\sqrt{V}R$,
 where $V\sim\GIG(\lambda,\chi,\psi)$,
 $R\sim\cN(0,1)$ and $V\perp R$, and set
$\alpha\coloneqq\sqrt{\psi+\eta^2}$.
For interior GIG parameters,
\begin{align}\label{eq:gh_density}
 f_T(t)
 &=\frac{e^{\eta t}(\psi/\chi)^{\lambda/2}}
 {\sqrt{2\pi}K_\lambda(\sqrt{\chi\psi})}
 \left(\frac{\chi+t^2}{\psi+\eta^2}\right)^{(\lambda-1/2)/2}
 K_{\lambda-1/2}\!\left(\alpha\sqrt{\chi+t^2}\right),
\end{align}
and
\begin{equation}\label{eq:gig_posterior}
 V|(T=t)
 \sim\GIG\!\left(\lambda-\frac12,\chi+t^2,\psi+\eta^2\right).
\end{equation}
The gamma and inverse-gamma boundaries are obtained by the corresponding limits.
\end{lemma}
\begin{proof}[Proof of Lemma~\ref{lem:conditional_gig}]By \eqref{eq:gig_density}, 
$$
 f_V(v)=c_Gv^{\lambda-1}
             \exp\left\{-\frac12\left(\frac\chi v+\psi v\right)\right\},
 \qquad v>0,
$$
where $c_G>0$ is the normalizing constant, which equals $(\psi/2)^\lambda/\Gamma(\lambda)$ and
$(\chi/2)^{-\lambda}/\Gamma(-\lambda)$ on the gamma and inverse-gamma
boundaries, respectively.
Conditionally on $V=v$, $T$ has density
$$
 f_{T\mid V}(t\mid v)
 =\frac{e^{\eta t}}{\sqrt{2\pi}}v^{-1/2}
       \exp\left\{-\frac12\left(\frac{t^2}{v}+\eta^2v\right)\right\}.
$$
Consequently, its joint density with $V$ is
\begin{equation}\label{eq:gh-joint-kernel-repaired}
 f_{T,V}(t,v)=\frac{c_Ge^{\eta t}}{\sqrt{2\pi}}
 v^{\lambda-3/2}
 \exp\left\{-\frac12\left(\frac{\chi+t^2}{v}
                              +(\psi+\eta^2)v\right)\right\}.
\end{equation}
As a function of $v$, this is the GIG kernel with parameters
$(\lambda-1/2,\chi+t^2,\psi+\eta^2)$.  It is proper for
$T$-almost every $t$, and hence
$$
 V\mid(T=t)\sim
 \operatorname{GIG}(\lambda-1/2,\chi+t^2,\psi+\eta^2)
 \quad\text{for $T$-almost every $t$}.
$$
For interior parameters, integrating
\eqref{eq:gh-joint-kernel-repaired} with
$$
 \int_0^\infty v^{q-1}
 \exp\!\left[-\frac12\left(\frac{a}{v}+bv\right)\right]\rd v
 =2\left(\frac{a}{b}\right)^{q/2}K_q(\sqrt{ab}),\qquad a,b>0,
$$
leads to \eqref{eq:gh_density}.

The above joint-density argument also establishes the conditional distribution
directly on both boundaries.  On the inverse-gamma boundary, its first parameter is negative, its second parameter is positive, and its third parameter is nonnegative, so it is
proper.  On the gamma boundary, $\chi+t^2>0$ whenever $t\ne0$.
At $\chi=0,t=0$, the formal conditional kernel is proper only if
$\lambda>1/2$.  This exceptional point has probability zero, because a
mixture of nondegenerate normal distributions has no atoms.  The
conditional distribution can be specified arbitrarily at an exceptional
null point.  In particular, this point causes no difficulty for the exact
sampler.
\end{proof}

\begin{proof}[Proof of Proposition~\ref{prop:gh_sampler}]
The distribution of $T_n$ has no atoms, since
$\sP(T_n=t\mid V)=0$ almost surely for each $t$.
Moreover, each nonempty finite interval has strictly positive conditional
normal probability for every $V>0$, so its cumulative distribution function is strictly increasing.
Thus, its continuous inverse-transform construction in step~(i) samples
$T_n\mid (T_n\leq -d_n^\tc)$:
for $t<-d_n^\tc$,
$$
 \sP(\widetilde T\le t)
 =\frac{F_{T_n}(t)}{F_{T_n}(-d_n^\tc)}.
$$
Lemma~\ref{lem:conditional_gig}, applied with the finite-$n$ loading $\eta_n$, implies the conditional distribution in step~(ii). 
The same step recovers
$\widetilde R=(\widetilde T-\eta_n\widetilde V)/\sqrt{\widetilde V}$. Then, the remaining steps of the proof are straightforward by an argument similar to the proof of Proposition~\ref{prop:new_outer_sampler}.
\end{proof}

\begin{proof}[Proof of Proposition~\ref{prop:gh_left_tail}]
For part~(i), first suppose $\chi,\psi>0$.  Insert $t=-x$ in~\eqref{eq:gh_density}.  The fixed-order Bessel expansion~\citep[Appendix~A]{Jorgensen1982GIG}
\begin{equation}\label{eq:bessel_asymptotic_appendix}
 K_q(y)=\sqrt{\frac{\pi}{2y}}e^{-y}\{1+O(y^{-1})\}
\end{equation}
and $\sqrt{\chi+x^2}=x+O(x^{-1})$ imply
$$
 f_T(-x)
 =C x^{\lambda-1}
 \exp[-\{\sqrt{\psi+\eta^2}+\eta\}x]\{1+o(1)\}.
$$
On the gamma boundary $\chi=0$, $\psi>0$, $\lambda>0$, start from
$$
 f_V(v)=\frac{(\psi/2)^\lambda}{\Gamma(\lambda)}v^{\lambda-1}e^{-\psi v/2}.
$$
Conditioning on $V$ and applying the same Bessel integral used in the proof of Lemma~\ref{lem:conditional_gig} yields, with $\alpha=\sqrt{\psi+\eta^2}$,
$$
 f_T(-x)
 =C e^{-\eta x}
 \int_0^\infty
 v^{\lambda-3/2}
 \exp\!\left[-\frac12\left(\frac{x^2}{v}+\alpha^2v\right)\right]\rd v
 =C'e^{-\eta x}\left(\frac{x^2}{\alpha^2}\right)^{(\lambda-1/2)/2}
 K_{\lambda-1/2}(\alpha x).
$$
By~\eqref{eq:bessel_asymptotic_appendix}, we obtain $f_T(-x)\sim C''x^{\lambda-1}e^{-\tau_-x}$, where $\tau_-\coloneqq \alpha+\eta>0$ since $\psi>0$ and $\alpha>|\eta|$.

Now suppose $\psi=0$, $\eta>0$, and $\lambda<0$.  The inverse-gamma density is
$$
 f_V(v)=\frac{(\chi/2)^{-\lambda}}{\Gamma(-\lambda)}v^{\lambda-1}e^{-\chi/(2v)}.
$$
Conditioning on $V$ leads to
\begin{align*}
 f_T(-x)
 &=C e^{-\eta x}
 \int_0^\infty
 v^{\lambda-3/2}
 \exp\!\left[-\frac12\left(\frac{\chi+x^2}{v}+\eta^2v\right)\right]\rd v\\
 &=C'e^{-\eta x}
 \left(\frac{\chi+x^2}{\eta^2}\right)^{(\lambda-1/2)/2}
 K_{\lambda-1/2}(\eta\sqrt{\chi+x^2}).
\end{align*}
The Bessel expansion yields $f_T(-x)\sim C''x^{\lambda-1}e^{-\tau_- x}$ since $\tau_-=|\eta|+\eta=2\eta$.
Consequently, l'H\^opital's rule implies
$$
 \frac{\int_x^\infty f_T(-y)\,\rd y}{x^{\lambda-1}e^{-\tau_-x}}
 \to\frac{C}{\tau_-},
$$
because the derivative of the denominator is $-\tau_-x^{\lambda-1}e^{-\tau_-x}\{1+O(x^{-1})\}$.  This proves~\eqref{eq:gh_semih_tail}.

For part~(ii), $\psi=0$, $\eta=0$, and $\lambda<0$ imply $\sP(V>x)\sim C_Vx^\lambda$.  Hence $\sqrt V$ is regularly varying with survival exponent $2\lambda<0$.  Since $W\coloneqq( -R)_+$ is independent of $V$ and has moments of every positive order, Breiman's lemma \citep{Breiman1965} implies
$$
 \sP(T<-x)
 =\sP(\sqrt VW>x)
 \sim\sE[W^{-2\lambda}]\sP(\sqrt V>x)
 \sim C_0x^{2\lambda}.
$$

For part~(iii), write $\eta=-a<0$.  Then $T=-aV+\sqrt VR$.  Fix $\varepsilon\in(0,1)$.  If
$$
 V>\frac{x}{a(1-\varepsilon)}
 \quad\text{and}\quad
 R\le\varepsilon a\sqrt V,
$$
then $T\le-a(1-\varepsilon)V<-x$.  Thus, with $v_x=x/[a(1-\varepsilon)]$,
$$
 \sP(T<-x)
 \ge\sP(V>v_x)
 -\sE[\sone_{\{V>v_x\}}\bar\Phi(\varepsilon a\sqrt V)]
 \ge\sP(V>v_x)\{1-\bar\Phi(\varepsilon a\sqrt{v_x})\}.
$$
For the upper bound, if $V\le x/[a(1+\varepsilon)]$ and $T<-x$, then
$$
 R<-\frac{x-aV}{\sqrt V}
 \le-\varepsilon\sqrt{\frac{ax}{1+\varepsilon}}.
$$
Therefore,
$$
 \sP(T<-x)
 \le\sP\!\left(V>\frac{x}{a(1+\varepsilon)}\right)
 +\bar\Phi\!\left(\varepsilon\sqrt{\frac{ax}{1+\varepsilon}}\right).
$$
The Gaussian remainders are exponentially small, whereas the inverse-gamma survival function is regularly varying with index $\lambda<0$.  Divide the lower and upper bounds by $\sP(V>x/a)$, apply the uniform convergence theorem for regularly varying functions \cite[Theorem~1.5.2]{BinghamGoldieTeugels1987}, and let $\varepsilon\downarrow0$.  This leads to $\sP(T<-x)\sim\sP(V>x/a)\sim C_\eta x^\lambda$.
\end{proof}

The subsequent three lemmas are used to prove Theorem~\ref{thm:gh_efficiency} and Proposition~\ref{prop:gh_outer_gain}. To show these lemmas, we define
\begin{equation}\label{eq:duplicated_gh_projection}
 T_{j,n}
 \coloneqq
 \eta V+\sqrt{V}\big(\sqrt{\rho_n}Z_0+\sqrt{1-\rho_n}Z_j\big)~\text{for}~j=1,2,~~\text{and}~~J_n(d)\coloneqq\sP(T_{1,n}<-d,\ T_{2,n}<-d),
\end{equation}
where $Z_0,Z_1,Z_2$ are independent standard normal random variables, independent of $V$. 
\begin{lemma}\label{lem:gh_payoff_bounds}
Suppose that 
Assumptions~\ref{ass:gh_model} and~\ref{ass:gh_fixed} hold and $h_n$ is uniformly positive and bounded, with $h_L$ and $h_U$ representing its lower and upper bounds, respectively. Let $\widetilde\kappa_n^\tc\coloneqq\sP(T_n<-d_n^\tc)=F^-_T\!(d_n^\tc)$. Then,
\begin{equation}\label{eq:gh_ch_upper}
 \widetilde\sE_n^*[(\widetilde\Gamma_n^*)^2]
 \le h_U\widetilde\kappa_n^\tc\gamma_n\le h_U^2(\widetilde\kappa_n^\tc)^2,
\end{equation}
and
\begin{equation}\label{eq:gh_il_sandwich}
 h_L^2J_n(d_n^\f)
 \le\widetilde\sE_n^0[(\widetilde\Gamma_n^0)^2]
 \le h_U^2J_n(d_n^\tc).
\end{equation}
\end{lemma}

\begin{proof}
Let $\bS=(S_{1},\ldots,S_{n})$ and
$\bS_{-n}=(S_{1},\ldots,S_{n-1})$ denote the asset vector and its
non-target coordinates in the $n$th network.
For $v>0$ and $\bz\in\mathbb R^{n-1}$, define the following probabilities
under the original distribution:
$$
\lt\{\,\begin{aligned}
 \varpi_n(v,\bz)
   &=\sP(S_{n}<v_{n}(\bS_{-n})
                 \mid V=v,\bZ_{-n}=\bz),\\
 \varpi_n^\tc(v,\bz)
   &=\sP(S_{n}<v_{n}^\tc
                 \mid V=v,\bZ_{-n}=\bz),\\
 \varpi_n^\f(v,\bz)
   &=\sP(S_{n}<v_{n}^\f
                 \mid V=v,\bZ_{-n}=\bz).
\end{aligned}\rt.
$$
Then, it is easy to check that for all $(v,\bz)$,
$$
 0<\varpi_n^\f(v,\bz)\le \varpi_n(v,\bz)\le \varpi_n^\tc(v,\bz)<1,
$$
and according to~\eqref{eq:gh_bliss_estimator}, the BLISS estimator is given by
$$
 \widetilde\Gamma_n^*
   =h_n(\bS)\frac{\widetilde\kappa_n^\tc \varpi_n(v,\bZ_{-n})}{\varpi_n^\tc(V,\bZ_{-n})}.
$$
It follows pointwise that
$0\le\widetilde\Gamma_n^*\le h_U\widetilde\kappa_n^\tc$, which leads to
$
 \widetilde\sE_n^*[(\widetilde\Gamma_n^*)^2]
             \le h_U\widetilde\kappa_n^\tc\gamma_n,
$
and hence, \eqref{eq:gh_ch_upper} holds.

For the inner-only estimator, $(V,\bZ_{-n})$ retains its original density,
while the inner draw uses the same conditional default distribution. Define
$
 m_{2}(v,\bz)
 =\sE\!\left[h_n(\bS)^2\mid
       V=v,\bZ_{-n}=\bz,S_{n}<v_{n}(\bS_{-n})\right]
$
under the original distribution. Then,
$$
 \widetilde\sE_n^0[(\widetilde\Gamma_n^0)^2]
 =\int_0^\infty\int_{\mathbb R^{n-1}}
    f_V(v)\phi_{n-1}(\bz)\varpi_n(v,\bz)^2m_{2}(v,\bz)\,d\bz\,dv.
$$
The payoff bounds on $\{S_{n}<v_{n}(\bS_{-n})\}$ imply
$h_L^2\le m_{2}(v,\bz)\le h_U^2$.
Introducing two conditionally independent target-specific Gaussian coordinates given $(V,\bZ_{-n})$ therefore results in
$$
 \sE[\varpi_n^\f(V,\bZ_{-n})^2]=J_n(d^\f_{n})\quad\text{and}\quad
 \sE[\varpi_n^\tc(V,\bZ_{-n})^2]=J_n(d^\tc_{n}).
$$
Combining these identities with the pointwise probability and payoff
bounds proves \eqref{eq:gh_il_sandwich}.
\end{proof}

\begin{lemma}\label{lem:uniform_gaussian_quadrant}
Fix $\bar\rho\in[0,1)$. There exist constants $0<c_1<c_2<\infty$ such that, for every $\rho\in[0,\bar\rho]$ and $z\in\mathbb R$,
$$
 c_1(1+z_+)^{-2}\exp\!\left(-\frac{z_+^2}{1+\rho}\right)
 \le H_\rho(z)\le
 c_2(1+z_+)^{-2}\exp\!\left(-\frac{z_+^2}{1+\rho}\right),
$$
where $z_+=\max(z,0)$ and $H_\rho(z)=\sP(G_1<-z,G_2<-z)$ for standard normals $G_1,G_2$ with correlation $\rho$.
\end{lemma}
\begin{proof}
Put $a_\rho=\sqrt{2/(1+\rho)}$ and $b_\rho=\sqrt{(1-\rho)/(1+\rho)}$. Orthogonalizing the sum and difference implies
$$
 H_\rho(z)=2\int_0^\infty \phi_1(w)\bar\Phi(a_\rho z+b_\rho w)\,dw.
$$
Both $a_\rho$ and $b_\rho$ are bounded above and away from zero on $[0,\bar\rho]$. For $z\ge1$, by the upper Mills bound, we have
\begin{align*}
H_\rho(z)
&\le \frac{2}{a_\rho z}\int_0^\infty\phi_1(w)\phi_1(a_\rho z+b_\rho w)\,dw\\
&\le\frac{e^{-a_\rho^2z^2/2}}{\pi a_\rho z}
 \int_0^\infty e^{-a_\rho b_\rho zw}\,dw
 =\frac{e^{-z^2/(1+\rho)}}{\pi a_\rho^2b_\rho z^2}.
\end{align*}
For a lower bound, we restrict the integral to $0\le w\le z^{-1}$. On this interval, $\phi_1(w)\ge\phi_1(1)$ and $t=a_\rho z+b_\rho w$ is bounded above by a constant times $z$ and below by a positive constant times $z$. Moreover,
$$
 \frac{t^2}{2}=\frac{z^2}{1+\rho}+a_\rho b_\rho zw+\frac{b_\rho^2w^2}{2}
 \le\frac{z^2}{1+\rho}+C.
$$
Therefore, by the lower Mills bound $\bar\Phi(t)\ge\phi_1(t)t/(1+t^2)$, we obtain $H_\rho(z)\ge c z^{-2}e^{-z^2/(1+\rho)}$, uniformly in $\rho$. Since $z^{-2}$ and $(1+z)^{-2}$ are comparable for $z\ge1$, this proves both bounds in the statement. For $z\le1$,  $H_\rho(z)\ge H_\rho(1)$ by monotonicity, and  $\min_{\rho\in[0,\bar\rho]}H_\rho(1)$ is strictly positive. The term $(1+z_+)^{-2}e^{-z_+^2/(1+\rho)}$ is also bounded above and away from zero on $z\le1$. Adjusting the constants completes the proof.
\end{proof}

\begin{lemma}\label{lem:gh_duplicated_tail}
Suppose that Assumptions~\ref{ass:gh_model} and \ref{ass:gh_fixed} hold with the semi-heavy branch, and assume that $\rho_n\coloneqq{\|\blambda_n\|^2}/{\omega_n^2}
 \to\rho\in[0,1)$.   For any sequence $d_n\to\infty$,
\begin{equation}\label{eq:gh_duplicated_rate}
 -\frac{1}{d_n}\log J_n(d_n)
 \to\tau_2(\rho)
 \coloneqq
 \frac{2}{1+\rho}
 \left(
 \sqrt{\eta^2+\frac{\psi}{2}(1+\rho)}+\eta
 \right).
\end{equation}
\end{lemma}

\begin{proof}
In the interior and on either proper boundary, the mixing density has the form
$$
 f_V(v)=c_Vv^{\lambda-1}\exp\left(-\frac{\chi}{2v}-\frac{\psi v}{2}\right),\qquad v>0,
$$
with a fixed positive normalizing constant. Choose $\bar\rho<1$ such that $\rho_n\le\bar\rho$ eventually. Conditional on $V=v$, the joint probability is $H_{\rho_n}((d_n+\eta v)/\sqrt v)$. Lemma~\ref{lem:uniform_gaussian_quadrant} and replacing $v$ with $d_nu$ show that $J_n(d_n)$ is bounded above and below by fixed positive multiples of
$$
 I_n=d_n^\lambda\int_0^\infty
 u^{\lambda-1}e^{-\chi/(2d_nu)}
 (1+\sqrt{d_n}q(u))^{-2}e^{-d_ng_n(u)}\,du,
$$
where
$$
 q(u)=\frac{(1+\eta u)_+}{\sqrt u}\quad\text{and}\quad
 g_n(u)=\frac{\psi u}{2}+\frac{(1+\eta u)_+^2}{(1+\rho_n)u}.
$$
Let $g$ denote the same function with $\rho_n$ replaced by $\rho$. It is continuous on $(0,\infty)$ and diverges at both endpoints. Also, $g(u)\sim 1/((1+\rho)u)$ as $u\to0$. At infinity, either $\psi>0$, or $\psi=0$ and the assumption $\tau_-\coloneqq \eta+\sqrt{\psi+\eta^2}>0$ implies $\eta>0$, in which case $g(u)\sim\eta^2u/(1+\rho)$ as $u\to\infty$.

On the region $1+\eta u>0$, the first-order condition on $g$ gives the unique minimizer
$$
 u_* =\left(\eta^2+\frac{\psi(1+\rho)}2\right)^{-1/2}\quad\text{with}\quad
 g(u_*)=\tau_2(\rho)\eqqcolon m.
$$
This point is in the active region. If $\eta<0$, necessarily $\psi>0$; the minimum over the inactive region $u\ge1/|\eta|$ is $\psi/(2|\eta|)$. One can check that
$$
 g(u_*)=\frac{\psi}{\sqrt{\eta^2+\psi(1+\rho)/2}+|\eta|}
 <\frac{\psi}{2|\eta|}.
$$
Thus, $m$ is the global minimum.

There exist fixed $c,C>0$ and $a_0,b_0>0$ such that, eventually in $n$,
$$
 g_n(u)\ge c/u~\text{for}~0<u<a_0\quad\text{and}\quad g_n(u)\ge Cu~\text{for}~u>b_0.
$$
One may take $c=1/[4(1+\bar\rho)]$ and $a_0\le1/(2|\eta|)$. Also, take $C=\psi/2$ if $\psi>0$, and $C=\eta^2/(1+\bar\rho)$ otherwise. Fix $\varepsilon>0$ and choose $0<a<u_*<b$ such that $a<a_0$, $b>b_0$, $c/(2a)>m+\varepsilon$, and $Cb/2>m+\varepsilon$. For $d_n\ge1$, we get
\begin{align*}
 d_n^\lambda\int_0^a u^{\lambda-1}e^{-d_ng_n(u)}\,du
 &\le d_n^\lambda e^{-d_nc/(2a)}
       \int_0^a u^{\lambda-1}e^{-c/(2u)}\,du,\\
 d_n^\lambda\int_b^\infty u^{\lambda-1}e^{-d_ng_n(u)}\,du
 &\le d_n^\lambda e^{-d_nCb/2}
       \int_b^\infty u^{\lambda-1}e^{-Cu/2}\,du.
\end{align*}
Both fixed integrals are finite for all $\lambda$. On $[a,b]$, $g_n\to g$ uniformly, so $g_n\ge m-\varepsilon$ eventually, and the contribution is at most $\exp[-d_n(m-\varepsilon)+O(\log d_n)]$. Consequently, $\limsup d_n^{-1}\log I_n\le-m+\varepsilon$.

For the other bound, we choose a fixed closed interval $K\subset(a,b)$ of positive length containing $u_*$ on which $g\le m+\varepsilon$. By uniform convergence, we have $g_n\le m+2\varepsilon$ on $K$ eventually. On this interval, $u^{\lambda-1}$ and $e^{-\chi/(2d_nu)}$ are bounded below by positive constants for $d_n\ge1$, and
$$
 \big(1+\sqrt{d_n}q(u)\big)^{-2}\ge\lt(1+\sqrt{d_n}\sup_{u\in K}q(u)\rt)^{-2}.
$$
Hence, $I_n\ge\exp[-d_n(m+2\varepsilon)+O(\log d_n)]$. Letting $\varepsilon\downarrow0$ proves the result.
\end{proof}

\begin{proof}[Proof of Theorem~\ref{thm:gh_efficiency}]
We write $\bS(v,\bz,u)$ for the asset vector when $V=v, \bZ_{-n}=\bz,$ and $Z_n=u$. Then, by Tonelli's theorem, we obtain
\begin{align*}
 \widetilde\sE_n^*[\widetilde\Gamma_n^*]
 &=\int_0^\infty\int_{\mathbb R^{n-1}}\int_{\mathbb R}
     f_V(v)\phi_{n-1}(z)\phi_1(u)h_n(\bS(v,\bz,u))
 \sone_{\{u<-\widetilde\ell_n(v,\bz)\}}\,du\,d\bz\,dv
 =\gamma_n.
\end{align*}
Let $\widetilde\kappa_n^\f\coloneqq F^-_T\!(d_n^\f)$. By Lemma~\ref{lem:gh_payoff_bounds} and the fact that $\{S_n<v_n^\f\}\subset\{S_n<v_n(\bS_{-n})\}\subset \{S_n<v_n^\tc\}$, we observe that
$h_L\widetilde\kappa_n^\f\le\gamma_n\le h_U\widetilde\kappa_n^\tc$, $0\leq\widetilde\Gamma_n^*\le h_U\widetilde\kappa_n^\tc$ almost surely, and $\gamma_n^2\leq\widetilde\sE_n^*[(\widetilde\Gamma_n^*)^2]
 \le h_U\widetilde\kappa_n^\tc\gamma_n.
$
Therefore,
\begin{equation}\label{eq:gh_relative_bound_appendix}
 1\leq\frac{\widetilde\sE_n^*[(\widetilde\Gamma_n^*)^2]}{\gamma_n^2}
 \le\frac{h_U}{h_L}
 \frac{\widetilde\kappa_n^\tc}{\widetilde\kappa_n^\f}.
\end{equation}

Let $\tau_-\coloneqq \eta+\sqrt{\psi+\eta^2}$. By the system-wide stability regime, we have $d_n^\tc\to\infty$ and $d_n^\f/d_n^\tc\to1$.  Thus, in the semi-heavy branch, Proposition~\ref{prop:gh_left_tail}(i) yields
$
 \log\widetilde\kappa_n^\tc
 \sim
 \log\widetilde\kappa_n^\f
 \sim
 -\tau_-d_n^\tc,
$ and $\log\gamma_n\sim-\tau_-d_n^\tc$. Consequently, taking logarithms of the preceding second-moment upper bound results in asymptotic optimality. If, in addition, $d_n^\f-d_n^\tc=O(1)$, then
$$
 \frac{F_T^-(d_n^\tc)}{F_T^-(d_n^\f)}
 \sim\left(\frac{d_n^\tc}{d_n^\f}\right)^{\lambda-1}
           \exp\{\tau_-(d_n^\f-d_n^\tc)\}
$$
is bounded.  Finally, \eqref{eq:gh_relative_bound_appendix} proves bounded relative error.
In a regularly varying branch, Proposition~\ref{prop:gh_left_tail}(ii) or (iii) implies that
$F_T^-(x)\sim Cx^r$ for a fixed $r<0$ ($r=2\lambda$ or $r=\lambda$).
Thus $d_n^\f/d_n^\tc\to1$ directly implies
$F_T^-(d_n^\tc)/F_T^-(d_n^\f)\to1$. Therefore, \eqref{eq:gh_relative_bound_appendix} again proves bounded
relative error.
\end{proof}

\begin{proof}[Proof of Proposition~\ref{prop:gh_outer_gain}]
The system-wide stability regime gives
$
 d_n^\tc\to\infty$,
 $d_n^\f\to\infty$,
and
 ${d_n^\f}/{d_n^\tc}\to1.
$ Let $\tau_-\coloneqq \eta+\sqrt{\psi+\eta^2}$, $\widetilde\kappa_n^\tc\coloneqq F_T^-(d_n^\tc)$, and $\widetilde\kappa_n^\f\coloneqq F_T^-(d_n^\f)$.

\emph{(i) Semi-heavy branch.} 
The proof of Theorem~\ref{thm:gh_efficiency} implies
$$
 \log\widetilde\sE_n^*[(\widetilde\Gamma_n^*)^2]
            =-2\tau_-d_n^\tc+o(d_n^\tc).
$$
Lemmas~\ref{lem:gh_payoff_bounds} and~\ref{lem:gh_duplicated_tail} suggest that
$$
 \log\widetilde\sE_n^0[(\widetilde\Gamma_n^0)^2]
            =-\tau_2(\rho)d_n^\tc+o(d_n^\tc),
 \quad\text{where}~
 \tau_2(\rho)
 \coloneqq
 \frac{2}{1+\rho}
 \left(
 \sqrt{\eta^2+\frac{\psi}{2}(1+\rho)}+\eta
 \right)>0.
$$
Thus, we have
\begin{equation}\label{eq:gh-log-ratio-explicit}
 \lim_{n\to\infty}
 \frac{\log\widetilde\sE_n^*[(\widetilde\Gamma_n^*)^2]}
      {\log\widetilde\sE_n^0[(\widetilde\Gamma_n^0)^2]}
 =\frac{(1+\rho)\bigl(\eta+\sqrt{\eta^2+\psi}\bigr)}
         {\eta+\sqrt{\eta^2+(1+\rho)\psi/2}}.
\end{equation}
Then, the desired results follow immediately.

\emph{(ii) Regularly varying branch.}
In this case, $\widetilde\kappa_n^\f\sim\widetilde\kappa_n^\tc\to0$. Lemma~\ref{lem:gh_payoff_bounds} implies that
\begin{equation}\label{eq:gh-gain-common-bound-repaired}
 0\le
 \frac{\widetilde\sE_n^*[(\widetilde\Gamma_n^*)^2]}
      {\widetilde\sE_n^0[(\widetilde\Gamma_n^0)^2]}
 \le \left(\frac{h_U}{h_L}\right)^2
            \frac{(\widetilde\kappa_n^\tc)^2}{J_n(d_n^\f)}.
\end{equation}
First, suppose $\eta=0$. By \eqref{eq:duplicated_gh_projection}, we observe that
$$
\begin{aligned}
 J_n(d_n^\f)
 &\ge\sP(Z_0\le-1,Z_1\le-1,Z_2\le-1)\sP(V>(d_n^\f)^2)
 =\Phi(-1)^3\sP(V>(d_n^\f)^2).
\end{aligned}
$$
The inverse-gamma survival tail is proportional to $x^\lambda$ and
Proposition~\ref{prop:gh_left_tail}(ii) gives $F_T^-(d)\sim C_0d^{2\lambda}$.
Thus, for some constant $b>0$ independent of $n$,
$J_n(d_n^\f)\ge b\widetilde\kappa_n^\f$ for all sufficiently large $n$.
The bound in \eqref{eq:gh-gain-common-bound-repaired} is then
$O((\widetilde\kappa_n^\tc)^2/\widetilde\kappa_n^\f)=O(\widetilde\kappa_n^\tc)\to0$.

Now, suppose $\eta=-a<0$. Let $v_x=x/(2a)$. The
proof of Proposition~\ref{prop:gh_left_tail}(iii) implies that $T_{1,n},T_{2,n}<-x$ on the event $\{V>v_x\}\cap\{\sqrt{\rho_n}Z_0+\sqrt{1-\rho_n}Z_j\le(a/2)\sqrt V,\ j=1,2\}$. Hence, $J_n(x)\ge\sP(V>v_x)\{1-2\bar\Phi((a/2)\sqrt{v_x})\}$ uniformly in $\rho_n$.
Since $\sP(V>d_n^\f/(2a))\sim F_T^-(d_n^\f/2)\sim b'\widetilde\kappa_n^\f$ for some constant $b'>0$ independent of $n$,
we have $J_n(d_n^\f)\ge (b'/2)\widetilde\kappa_n^\f$ uniformly in $\rho_n$.
Substitution into \eqref{eq:gh-gain-common-bound-repaired} once more
bounds the ratio by $O(\widetilde\kappa_n^\tc)\to0$.
No convergence of $\rho_n$ is needed in the regularly varying branch.
\end{proof}
\end{revision}
\section{Supplementary Numerical Results}\label{app:numeric}
\setcounter{figure}{0}
\setcounter{table}{0}
\renewcommand{\thefigure}{B\arabic{figure}}
\renewcommand{\thetable}{B\arabic{table}}
    \begin{table}[H] \centering\footnotesize
        \begin{tabular}{clrrrr}
            \toprule
         No. & Bank name & Total asset & Net worth & Interbank asset & Equity vol. \\ \midrule
            
1  & HSBC Holdings Plc                               & 1,322,909 & 125,976 & 117,004 & 16.46\% \\
2  & Groupe Crédit Agricole                          & 1,047,925 & 84,292  & 97,114  & 23.34\% \\
3  & BNP Paribas                                     & 1,012,707 & 84,417  & 63,604  & 22.17\% \\
4  & ING Groep N.V.                                  & 780,776   & 50,325  & 76,469  & 17.21\% \\
5  & Deutsche Bank AG                                & 758,140   & 57,631  & 58,015  & 31.54\% \\
6  & Groupe BPCE                                     & 668,255   & 59,490  & 32,956  & *19.41\% \\
7  & Société Générale S.A.                           & 642,940   & 49,514  & 53,400  & 25.53\% \\
8  & Lloyds Banking Group Plc                        & 590,827   & 40,948  & 8,817   & 17.49\% \\
9  & Banco Santander S.A.                            & 565,109   & 77,283  & 36,878  & 22.70\% \\
10 & Barclays Plc                                    & 562,002   & 60,765  & 49,797  & 21.87\% \\
11 & Coöperatieve Rabobank U.A.                      & 547,353   & 37,204  & 14,461  & *19.41\% \\
12 & The Royal Bank of Scotland Group Plc            & 490,122   & 44,577  & 23,685  & 23.29\% \\
13 & Nordea Bank 
& 437,347   & 28,008  & 40,127  & 18.00\% \\
14 & Group Crédit Mutuel                             & 430,308   & 45,578  & 44,606  & *19.41\% \\
15 & UniCredit S.p.A.                                & 395,077   & 54,703  & 39,890  & 34.86\% \\
16 & ABN AMRO Group N.V.                             & 367,487   & 19,618  & 14,942  & 19.82\% \\
17 & Commerzbank AG                                  & 314,214   & 25,985  & 42,564  & 28.94\% \\
18 & Intesa Sanpaolo S.p.A.                          & 309,144   & 43,466  & 36,125  & 21.10\% \\
19 & Banco Bilbao Vizcaya Argentaria S.A.            & 276,960   & 46,980  & 75,226  & 22.20\% \\
20 & Danske Bank                                     & 273,199   & 20,302  & 14,926  & 16.13\% \\
21 & Svenska Handelsbanken 
& 253,639   & 12,954  & 7,339   & 16.53\% \\
22 & KBC Group NV                                    & 226,455   & 16,552  & 17,128  & 18.27\% \\
23 & Skandinaviska Enskilda Banken 
& 209,082   & 13,452  & 14,944  & 14.99\% \\
24 & Landesbank Baden-Württemberg                    & 206,824   & 12,795  & 57,434  & *19.41\% \\
25 & Swedbank 
& 202,830   & 11,356  & 6,522   & 15.90\% \\
26 & DZ BANK AG                                      & 202,301   & 19,923  & 35,800  & *19.41\% \\
27 & Bayerische Landesbank                           & 193,192   & 9,393   & 22,731  & *19.41\% \\
28 & Erste Group Bank AG                             & 147,479   & 15,368  & 10,756  & 21.08\% \\
29 & Belfius Banque SA                               & 127,440   & 8,141   & 35,597  & *19.41\% \\
30 & Landesbank Hessen-Thüringen 
& 122,115   & 8,180   & 15,767  & *19.41\% \\
31 & Banco de Sabadell S.A.                          & 108,282   & 11,111  & 1,559   & 28.60\% \\
32 & OP Financial Group                              & 91,467    & 9,973   & 7,299   & *19.41\% \\
33 & Norddeutsche Landesbank 
& 71,764    & 6,229   & 16,037  & *19.41\% \\
34 & Bank of Ireland Group plc                       & 68,264    & 7,617   & 4,537   & 29.86\% \\
35 & Raiffeisen Bank International AG                & 61,499    & 9,839   & 5,416   & 31.87\% \\
36 & Allied Irish Banks Group plc                    & 48,157    & 11,028  & 10,064  & 56.28\% \\
            \bottomrule
        \end{tabular}
        \caption{Balance sheet characteristics of 36 major European banks from the 2018 EBA stress test. Total assets, net worth, and interbank assets are in millions of euros. Equity volatility is annualized. Banks marked with * are not publicly traded. We proxy their equity volatility using the volatility of the EURO STOXX Banks index. We use this data to calibrate the model as described in Section~\ref{subsubsec:calibration}. 
        \label{tab:eba}}
\end{table}
\begin{table}[H]
    \centering
    \begin{sideways}
        \scriptsize
        \setlength{\tabcolsep}{2pt} 

        \resizebox{0.88\textheight}{!}{%
$\lt[\begin{tabular}{rrrrrrrrrrrrrrrrrrrrrrrrrrrrrrrrrrrr}
1.00 & .00 & .00 & .00 & .00 & .00 & .00 & .00 & .00 & .00 & .00 & .00 & .00 & .00 & .00 & .00 & .00 & .00 & .00 & .00 & .00 & .00 & .00 & .00 & .00 & .00 & .00 & .00 & .00 & .00 & .00 & .00 & .00 & .00 & .00 & .00 \\
.07 & 1.00 & .00 & .00 & .00 & .00 & .00 & .00 & .00 & .00 & .00 & .00 & .00 & .00 & .00 & .00 & .00 & .00 & .00 & .00 & .00 & .00 & .00 & .00 & .00 & .00 & .00 & .00 & .00 & .00 & .00 & .00 & .00 & .00 & .00 & .00 \\
.03 & \multicolumn{1}{c}{$\text{-}.01$} & 1.00 & .00 & .00 & .00 & .00 & .00 & .00 & .00 & .00 & .00 & .00 & .00 & .00 & .00 & .00 & .00 & .00 & .00 & .00 & .00 & .00 & .00 & .00 & .00 & .00 & .00 & .00 & .00 & .00 & .00 & .00 & .00 & .00 & .00 \\
.02 & .07 & .03 & 1.00 & .00 & .00 & .00 & .00 & .00 & .00 & .00 & .00 & .00 & .00 & .00 & .00 & .00 & .00 & .00 & .00 & .00 & .00 & .00 & .00 & .00 & .00 & .00 & .00 & .00 & .00 & .00 & .00 & .00 & .00 & .00 & .00 \\
.03 & .01 & \multicolumn{1}{c}{$\text{-}.01$} & .03 & 1.00 & .00 & .00 & .00 & .00 & .00 & .00 & .00 & .00 & .00 & .00 & .00 & .00 & .00 & .00 & .00 & .00 & .00 & .00 & .00 & .00 & .00 & .00 & .00 & .00 & .00 & .00 & .00 & .00 & .00 & .00 & .00 \\
.04 & .03 & .02 & .06 & \multicolumn{1}{c}{$\text{-}.02$} & 1.00 & .00 & .00 & .00 & .00 & .00 & .00 & .00 & .00 & .00 & .00 & .00 & .00 & .00 & .00 & .00 & .00 & .00 & .00 & .00 & .00 & .00 & .00 & .00 & .00 & .00 & .00 & .00 & .00 & .00 & .00 \\
.08 & .03 & .01 & \multicolumn{1}{c}{$\text{-}.02$} & .04 & .05 & .99 & .00 & .00 & .00 & .00 & .00 & .00 & .00 & .00 & .00 & .00 & .00 & .00 & .00 & .00 & .00 & .00 & .00 & .00 & .00 & .00 & .00 & .00 & .00 & .00 & .00 & .00 & .00 & .00 & .00 \\
.02 & .01 & .05 & .01 & .00 & .02 & .01 & 1.00 & .00 & .00 & .00 & .00 & .00 & .00 & .00 & .00 & .00 & .00 & .00 & .00 & .00 & .00 & .00 & .00 & .00 & .00 & .00 & .00 & .00 & .00 & .00 & .00 & .00 & .00 & .00 & .00 \\
.08 & .07 & \multicolumn{1}{c}{$\text{-}.02$} & .02 & .02 & .02 & .01 & .01 & .99 & .00 & .00 & .00 & .00 & .00 & .00 & .00 & .00 & .00 & .00 & .00 & .00 & .00 & .00 & .00 & .00 & .00 & .00 & .00 & .00 & .00 & .00 & .00 & .00 & .00 & .00 & .00 \\
.01 & .01 & .06 & .07 & .01 & .05 & .00 & .05 & .03 & .99 & .00 & .00 & .00 & .00 & .00 & .00 & .00 & .00 & .00 & .00 & .00 & .00 & .00 & .00 & .00 & .00 & .00 & .00 & .00 & .00 & .00 & .00 & .00 & .00 & .00 & .00 \\
.01 & \multicolumn{1}{c}{$\text{-}.01$} & .04 & .00 & .00 & \multicolumn{1}{c}{$\text{-}.01$} & .06 & .03 & .02 & .00 & 1.00 & .00 & .00 & .00 & .00 & .00 & .00 & .00 & .00 & .00 & .00 & .00 & .00 & .00 & .00 & .00 & .00 & .00 & .00 & .00 & .00 & .00 & .00 & .00 & .00 & .00 \\
.05 & .04 & .05 & .09 & .06 & .06 & .05 & .04 & .04 & .04 & .04 & .98 & .00 & .00 & .00 & .00 & .00 & .00 & .00 & .00 & .00 & .00 & .00 & .00 & .00 & .00 & .00 & .00 & .00 & .00 & .00 & .00 & .00 & .00 & .00 & .00 \\
.06 & .03 & .01 & .08 & .02 & .01 & .01 & \multicolumn{1}{c}{$\text{-}.01$} & .01 & \multicolumn{1}{c}{$\text{-}.01$} & .01 & .03 & .99 & .00 & .00 & .00 & .00 & .00 & .00 & .00 & .00 & .00 & .00 & .00 & .00 & .00 & .00 & .00 & .00 & .00 & .00 & .00 & .00 & .00 & .00 & .00 \\
.07 & .10 & .03 & .00 & \multicolumn{1}{c}{$\text{-}.01$} & .03 & .04 & .03 & .08 & .04 & .03 & .04 & .00 & .98 & .00 & .00 & .00 & .00 & .00 & .00 & .00 & .00 & .00 & .00 & .00 & .00 & .00 & .00 & .00 & .00 & .00 & .00 & .00 & .00 & .00 & .00 \\
.04 & .03 & \multicolumn{1}{c}{$\text{-}.01$} & .04 & .01 & .04 & .05 & .07 & .04 & .01 & .03 & .06 & .05 & .05 & .99 & .00 & .00 & .00 & .00 & .00 & .00 & .00 & .00 & .00 & .00 & .00 & .00 & .00 & .00 & .00 & .00 & .00 & .00 & .00 & .00 & .00 \\
.05 & .05 & .01 & \multicolumn{1}{c}{$\text{-}.03$} & \multicolumn{1}{c}{$\text{-}.03$} & .04 & .02 & .03 & .02 & .01 & \multicolumn{1}{c}{$\text{-}.01$} & \multicolumn{1}{c}{$\text{-}.02$} & \multicolumn{1}{c}{$\text{-}.01$} & .01 & \multicolumn{1}{c}{$\text{-}.01$} & .99 & .00 & .00 & .00 & .00 & .00 & .00 & .00 & .00 & .00 & .00 & .00 & .00 & .00 & .00 & .00 & .00 & .00 & .00 & .00 & .00 \\
.06 & .06 & .03 & \multicolumn{1}{c}{$\text{-}.04$} & .04 & .05 & .04 & .06 & .05 & .04 & \multicolumn{1}{c}{$\text{-}.01$} & \multicolumn{1}{c}{$\text{-}.03$} & \multicolumn{1}{c}{$\text{-}.01$} & .06 & .01 & .02 & .99 & .00 & .00 & .00 & .00 & .00 & .00 & .00 & .00 & .00 & .00 & .00 & .00 & .00 & .00 & .00 & .00 & .00 & .00 & .00 \\
.06 & .05 & .03 & .03 & .03 & .04 & \multicolumn{1}{c}{$\text{-}.02$} & .07 & .07 & .04 & \multicolumn{1}{c}{$\text{-}.03$} & .01 & .01 & .02 & .05 & .05 & .07 & .98 & .00 & .00 & .00 & .00 & .00 & .00 & .00 & .00 & .00 & .00 & .00 & .00 & .00 & .00 & .00 & .00 & .00 & .00 \\
.06 & .05 & \multicolumn{1}{c}{$\text{-}.01$} & .02 & .05 & .00 & .05 & .01 & .06 & \multicolumn{1}{c}{$\text{-}.01$} & .00 & .02 & \multicolumn{1}{c}{$\text{-}.03$} & .01 & .01 & \multicolumn{1}{c}{$\text{-}.02$} & .04 & .00 & .99 & .00 & .00 & .00 & .00 & .00 & .00 & .00 & .00 & .00 & .00 & .00 & .00 & .00 & .00 & .00 & .00 & .00 \\
.03 & .03 & .05 & .04 & .09 & .00 & .01 & .08 & .04 & .06 & \multicolumn{1}{c}{$\text{-}.01$} & .05 & \multicolumn{1}{c}{$\text{-}.02$} & .05 & .05 & \multicolumn{1}{c}{$\text{-}.01$} & .04 & .05 & .05 & .98 & .00 & .00 & .00 & .00 & .00 & .00 & .00 & .00 & .00 & .00 & .00 & .00 & .00 & .00 & .00 & .00 \\
.05 & .05 & .03 & .03 & .08 & .02 & .07 & .09 & .07 & .03 & .02 & .05 & .04 & .06 & .04 & .01 & .05 & .01 & .06 & .04 & .97 & .00 & .00 & .00 & .00 & .00 & .00 & .00 & .00 & .00 & .00 & .00 & .00 & .00 & .00 & .00 \\
.06 & .02 & .07 & .01 & .05 & .05 & .07 & .05 & .01 & .01 & .02 & .06 & \multicolumn{1}{c}{$\text{-}.01$} & .03 & .04 & .04 & .01 & .05 & \multicolumn{1}{c}{$\text{-}.02$} & .01 & .04 & .98 & .00 & .00 & .00 & .00 & .00 & .00 & .00 & .00 & .00 & .00 & .00 & .00 & .00 & .00 \\
.07 & .04 & .03 & .03 & .01 & .03 & .05 & .06 & .08 & .04 & .05 & .07 & .04 & .03 & .04 & .01 & .03 & .03 & .01 & .02 & .03 & .03 & .98 & .00 & .00 & .00 & .00 & .00 & .00 & .00 & .00 & .00 & .00 & .00 & .00 & .00 \\
.02 & .05 & .00 & .02 & .06 & .03 & \multicolumn{1}{c}{$\text{-}.02$} & .02 & .05 & .06 & .00 & .06 & \multicolumn{1}{c}{$\text{-}.03$} & .04 & .01 & \multicolumn{1}{c}{$\text{-}.01$} & .05 & .08 & \multicolumn{1}{c}{$\text{-}.02$} & .06 & \multicolumn{1}{c}{$\text{-}.01$} & .03 & .02 & .98 & .00 & .00 & .00 & .00 & .00 & .00 & .00 & .00 & .00 & .00 & .00 & .00 \\
.06 & .05 & .06 & .09 & .07 & .02 & .04 & .03 & .05 & .07 & .05 & .07 & .04 & .04 & .00 & .01 & .00 & .00 & .02 & .04 & .06 & .02 & .02 & \multicolumn{1}{c}{$\text{-}.01$} & .97 & .00 & .00 & .00 & .00 & .00 & .00 & .00 & .00 & .00 & .00 & .00 \\
.03 & .05 & .04 & .07 & .05 & .02 & .07 & .00 & .00 & .04 & .03 & .05 & .02 & .03 & .05 & \multicolumn{1}{c}{$\text{-}.03$} & .00 & .00 & .04 & .03 & .01 & .03 & .02 & .00 & .05 & .98 & .00 & .00 & .00 & .00 & .00 & .00 & .00 & .00 & .00 & .00 \\
.04 & .06 & .07 & .06 & .00 & .03 & .01 & .06 & .06 & .02 & .07 & .03 & .01 & .06 & .02 & \multicolumn{1}{c}{$\text{-}.02$} & .03 & .02 & .04 & .03 & .02 & \multicolumn{1}{c}{$\text{-}.02$} & .03 & \multicolumn{1}{c}{$\text{-}.01$} & .03 & .03 & .98 & .00 & .00 & .00 & .00 & .00 & .00 & .00 & .00 & .00 \\
.05 & .06 & .01 & .02 & .01 & .03 & .06 & .01 & .02 & \multicolumn{1}{c}{$\text{-}.01$} & .00 & .01 & .05 & .05 & .01 & .03 & .02 & \multicolumn{1}{c}{$\text{-}.04$} & .03 & .05 & .05 & \multicolumn{1}{c}{$\text{-}.01$} & .04 & \multicolumn{1}{c}{$\text{-}.03$} & .05 & .03 & .02 & .98 & .00 & .00 & .00 & .00 & .00 & .00 & .00 & .00 \\
.05 & .06 & .04 & .00 & .03 & .04 & .05 & .02 & .09 & .05 & .00 & .04 & .01 & .09 & .03 & .02 & .04 & .01 & .00 & .02 & .01 & .04 & .05 & .02 & .02 & .03 & .04 & .04 & .98 & .00 & .00 & .00 & .00 & .00 & .00 & .00 \\
.09 & .06 & .08 & .05 & .04 & .05 & .06 & .03 & .02 & .04 & .01 & .05 & .03 & .06 & .03 & .04 & .02 & .04 & .00 & .03 & .03 & .04 & .02 & .00 & .03 & .04 & .07 & .01 & .03 & .97 & .00 & .00 & .00 & .00 & .00 & .00 \\
\multicolumn{1}{c}{$\text{-}.01$} & .02 & .04 & .08 & .04 & .06 & .03 & .05 & .01 & .03 & .04 & .07 & .03 & .02 & .05 & \multicolumn{1}{c}{$\text{-}.03$} & .00 & .03 & \multicolumn{1}{c}{$\text{-}.01$} & .04 & .03 & .04 & .05 & .01 & .03 & .03 & .01 & .02 & .02 & \multicolumn{1}{c}{$\text{-}.01$} & .98 & .00 & .00 & .00 & .00 & .00 \\
\multicolumn{1}{c}{$\text{-}.04$} & \multicolumn{1}{c}{$\text{-}.01$} & .05 & .06 & \multicolumn{1}{c}{$\text{-}.01$} & .07 & .03 & .01 & .02 & .02 & .04 & .03 & .01 & \multicolumn{1}{c}{$\text{-}.01$} & .05 & .00 & \multicolumn{1}{c}{$\text{-}.03$} & .02 & .00 & .00 & .01 & .02 & .04 & \multicolumn{1}{c}{$\text{-}.03$} & .04 & .04 & .03 & .03 & .01 & .01 & .08 & .98 & .00 & .00 & .00 & .00 \\
\multicolumn{1}{c}{$\text{-}.02$} & .07 & .02 & .02 & \multicolumn{1}{c}{$\text{-}.02$} & \multicolumn{1}{c}{$\text{-}.01$} & \multicolumn{1}{c}{$\text{-}.03$} & .08 & .03 & .05 & .06 & .00 & \multicolumn{1}{c}{$\text{-}.03$} & .08 & .00 & .03 & .03 & .01 & \multicolumn{1}{c}{$\text{-}.02$} & .04 & .03 & .02 & .01 & .03 & .02 & \multicolumn{1}{c}{$\text{-}.03$} & .02 & .00 & .01 & \multicolumn{1}{c}{$\text{-}.02$} & .02 & \multicolumn{1}{c}{$\text{-}.02$} & .98 & .00 & .00 & .00 \\
.04 & .04 & .03 & .03 & .00 & .02 & .02 & \multicolumn{1}{c}{$\text{-}.02$} & .02 & .04 & .04 & .01 & .02 & .05 & \multicolumn{1}{c}{$\text{-}.01$} & \multicolumn{1}{c}{$\text{-}.01$} & .00 & .01 & \multicolumn{1}{c}{$\text{-}.01$} & .03 & .01 & .02 & .01 & .04 & .04 & \multicolumn{1}{c}{$\text{-}.05$} & .01 & .03 & .00 & .05 & .02 & .03 & .04 & .99 & .00 & .00 \\
.02 & .02 & .06 & .00 & .01 & .04 & .03 & .04 & .07 & .07 & .02 & .02 & \multicolumn{1}{c}{$\text{-}.03$} & .07 & .03 & .03 & .03 & .03 & .01 & .06 & .01 & .00 & .05 & .03 & .01 & .02 & .07 & .05 & .08 & .07 & .01 & .05 & .01 & .02 & .97 & .00 \\
.08 & .00 & .00 & .02 & .02 & .02 & .06 & .02 & .05 & .04 & .04 & .02 & .00 & .04 & .06 & \multicolumn{1}{c}{$\text{-}.01$} & .01 & .01 & .06 & .04 & .04 & .03 & .04 & .01 & .01 & .01 & .01 & .02 & .00 & .02 & .04 & .03 & \multicolumn{1}{c}{$\text{-}.01$} & .05 & .01 & .98 \\
\end{tabular}\rt]$
        }
    \end{sideways}
    \caption{The matrix $\bR$ for the \edit{non-diagonal-loading case} in Section~\ref{subsec:numeric_2}\label{tab:correlation}}
\end{table}
\begin{revision}
\newcolumntype{C}[1]{>{\centering\arraybackslash}p{#1}}
\begin{table}[!t]
    \centering\footnotesize
\edit{\begin{tabular}{C{3.3cm}C{1.5cm}C{1.5cm}C{1.5cm}cC{1.5cm}C{1.5cm}C{1.5cm}}
\toprule
Price-impact factor&\multicolumn{3}{c}{Non-diagonal-loading}&&\multicolumn{3}{c}{Diagonal-loading}\\\cmidrule{2-4}\cmidrule{6-8}
$\zeta$ ($\times10^{-8}$)&Gaussian&$t$&VG&&Gaussian&$t$&VG \\\midrule
2.0 & 86 & 111 & 102 && 83 & 107 & 99 \\
1.6 & 58 & 88  & 77 && 56 & 85  & 75 \\
1.2 & 36 & 69  & 58 && 35 & 67  & 56 \\
0.8 & 21 & 54  & 42 && 21 & 52  & 41 \\
0.4 & 11 & 41  & 29 && 11 & 40  & 29 \\
0 & 5  & 31  & 20 && 5  & 31  & 19 \\\bottomrule 
\end{tabular}
\begin{tabular}{C{3.3cm}C{1.5cm}C{1.5cm}C{1.5cm}cC{1.5cm}C{1.5cm}C{1.5cm}}
\toprule
Bankruptcy-cost factor&\multicolumn{3}{c}{Non-diagonal-loading}&&\multicolumn{3}{c}{Diagonal-loading}\\\cmidrule{2-4}\cmidrule{6-8}
$\delta$&Gaussian&$t$&VG&&Gaussian&$t$&VG \\\midrule
0.5 & 47 & 102 & 83 && 45 & 99 & 81 \\
0.4 & 43 & 93  & 76 && 41 & 90 & 74 \\
0.3 & 39 & 84  & 69 && 37 & 82 & 67 \\
0.2 & 35 & 77  & 62 && 34 & 74 & 60 \\
0.1 & 31 & 69  & 55 && 30 & 67 & 54 \\
0   & 28 & 61  & 49 && 27 & 59 & 48\\\bottomrule
\end{tabular}

    \caption{\label{tab:12}
    Yield estimates (in basis points) for different price-impact and bankruptcy-cost factors with the exponential inverse demand function. The upper and lower tables show the estimates under the experimental settings of Figures~\ref{fig:nu} and~\ref{fig:delta}, respectively. VG stands for variance-gamma.}}
\end{table}

\begin{table}[!t]
    \centering\footnotesize
\edit{\begin{tabular}{C{3.3cm}C{1.5cm}C{1.5cm}C{1.5cm}cC{1.5cm}C{1.5cm}C{1.5cm}}
\toprule
&\multicolumn{3}{c}{Non-diagonal-loading}&&\multicolumn{3}{c}{Diagonal-loading}\\\cmidrule{2-4}\cmidrule{6-8}
Scenario&Gaussian&$t$&VG&&Gaussian&$t$&VG \\\midrule
Beta     & 32  & 63  & 51 && 30  & 61  & 50 \\
Uniform  & 34  & 64  & 53 && 33  & 62  & 52 \\
Priority & 39 & 72 & 75 && 27 & 62 & 66 \\\bottomrule 
\end{tabular}
\caption{\label{tab:34}
Yield estimates (in basis points) across different random price impact scenarios and external-liability priority scenarios with the exponential inverse demand function. VG stands for variance-gamma.}}
    
\end{table}
\begin{figure}[!t]
\centering
	\begin{tikzpicture}[font=\sffamily\small]
	\begin{semilogyaxis}[name=plot1,height=1.8in,width=0.315\textwidth,
    title = {Gaussian},
	ylabel={Standard error},
	every tick label/.append style={font=\scriptsize},
	axis on top,
 	scaled x ticks = false,    x dir=reverse,
	xticklabel style={/pgf/number format/fixed},
	/pgf/number format/1000 sep={}]
	\addplot+[black, densely dotted, thick, mark = o, mark size = 2pt, mark options=solid]	table[x index=0,y index=2, col sep=comma]{figures/ex2GlinC.csv};
	\addplot+[blue, thick, dashed, mark = square, mark size = 2pt, mark options=solid]	table[x index=0,y index=5, col sep=comma]{figures/ex2GlinC.csv};
	\addplot+[red, thick, solid, mark = star, mark options=solid]	table[x index=0,y index=8, col sep=comma]{figures/ex2GlinC.csv};
	\end{semilogyaxis}
	\begin{semilogyaxis}[name=plot2,height=1.8in,width=0.315\textwidth,at={($(plot1.east)+(0.4in,0in)$)},anchor=west,
    title = {$t$},
	every tick label/.append style={font=\scriptsize},
	axis on top,
 	scaled x ticks = false,    x dir=reverse,
	xticklabel style={/pgf/number format/fixed},
	/pgf/number format/1000 sep={}]
	\addplot+[black, densely dotted, thick, mark = o, mark size = 2pt, mark options=solid]	table[x index=0,y index=2, col sep=comma]{figures/ex2TlinC.csv};
	\addplot+[blue, thick, dashed, mark = square, mark size = 2pt, mark options=solid]	table[x index=0,y index=5, col sep=comma]{figures/ex2TlinC.csv};
	\addplot+[red, thick, solid, mark = star, mark options=solid]	table[x index=0,y index=8, col sep=comma]{figures/ex2TlinC.csv};
	\end{semilogyaxis}
	\begin{semilogyaxis}[name=plot3,height=1.8in,width=0.315\textwidth,at={($(plot2.east)+(0.4in,0in)$)},anchor=west,
    title = {Variance-gamma},
	every tick label/.append style={font=\scriptsize},
	axis on top,
 	scaled x ticks = false,    x dir=reverse,
	xticklabel style={/pgf/number format/fixed},
	/pgf/number format/1000 sep={},
    legend pos=outer north east, legend entries={MC, ILIS, BLISS}]
	\addplot+[black, densely dotted, thick, mark = o, mark size = 2pt, mark options=solid]	table[x index=0,y index=2, col sep=comma]{figures/ex2VGlinC.csv};
	\addplot+[blue, thick, dashed, mark = square, mark size = 2pt, mark options=solid]	table[x index=0,y index=5, col sep=comma]{figures/ex2VGlinC.csv};
	\addplot+[red, thick, solid, mark = star, mark options=solid]	table[x index=0,y index=8, col sep=comma]{figures/ex2VGlinC.csv};
	\end{semilogyaxis}
	\begin{semilogyaxis}[name=plot4,height=1.8in,width=0.315\textwidth,at={($(plot1.south)-(0in,0.2in)$)},anchor=north,
	xlabel={$\zeta~(\times10^{-8})$},
	ylabel={Efficiency ratio},
	every tick label/.append style={font=\scriptsize},
	axis on top,
 	scaled x ticks = false,    x dir=reverse,
	xticklabel style={/pgf/number format/fixed},
	/pgf/number format/1000 sep={}]
	\addplot+[black, densely dotted, thick, mark = o, mark size = 2pt, mark options=solid]	table[x index=0,y index=3, col sep=comma]{figures/ex2GlinC.csv};
	\addplot+[blue, thick, dashed, mark = square, mark size = 2pt, mark options=solid]	table[x index=0,y index=6, col sep=comma]{figures/ex2GlinC.csv};
	\addplot+[red, thick, solid, mark = star, mark options=solid]	table[x index=0,y index=9, col sep=comma]{figures/ex2GlinC.csv};
	\end{semilogyaxis}
	\begin{semilogyaxis}[name=plot5,height=1.8in,width=0.315\textwidth,at={($(plot4.east)+(0.4in,0in)$)},anchor=west,
	xlabel={$\zeta~(\times10^{-8})$},
	every tick label/.append style={font=\scriptsize},
	axis on top,
 	scaled x ticks = false,    x dir=reverse,
	xticklabel style={/pgf/number format/fixed},
	/pgf/number format/1000 sep={}]
	\addplot+[black, densely dotted, thick, mark = o, mark size = 2pt, mark options=solid]	table[x index=0,y index=3, col sep=comma]{figures/ex2TlinC.csv};
	\addplot+[blue, thick, dashed, mark = square, mark size = 2pt, mark options=solid]	table[x index=0,y index=6, col sep=comma]{figures/ex2TlinC.csv};
	\addplot+[red, thick, solid, mark = star, mark options=solid]	table[x index=0,y index=9, col sep=comma]{figures/ex2TlinC.csv};
	\end{semilogyaxis}
	\begin{semilogyaxis}[name=plot6,height=1.8in,width=0.315\textwidth,at={($(plot5.east)+(0.4in,0in)$)},anchor=west,
	xlabel={$\zeta~(\times10^{-8})$},
	every tick label/.append style={font=\scriptsize},
	axis on top,
 	scaled x ticks = false,    x dir=reverse,
	xticklabel style={/pgf/number format/fixed},
	/pgf/number format/1000 sep={}]
	\addplot+[black, densely dotted, thick, mark = o, mark size = 2pt, mark options=solid]	table[x index=0,y index=3, col sep=comma]{figures/ex2VGlinC.csv};
	\addplot+[blue, thick, dashed, mark = square, mark size = 2pt, mark options=solid]	table[x index=0,y index=6, col sep=comma]{figures/ex2VGlinC.csv};
	\addplot+[red, thick, solid, mark = star, mark options=solid]	table[x index=0,y index=9, col sep=comma]{figures/ex2VGlinC.csv};
	\end{semilogyaxis}
    \node[yshift=29pt] 
    at ($(plot1.north west)!0.5!(plot3.north east)$) 
    {(a) Non-diagonal-loading case};
    
	\begin{semilogyaxis}[name=plot7,height=1.8in,width=0.315\textwidth,at={($(plot4.south)-(0in,1in)$)},anchor=north,
    title = {Gaussian},
	ylabel={Standard error},
	every tick label/.append style={font=\scriptsize},
	axis on top,
 	scaled x ticks = false,    x dir=reverse,
	xticklabel style={/pgf/number format/fixed},
	/pgf/number format/1000 sep={}]
	\addplot+[black, densely dotted, thick, mark = o, mark size = 2pt, mark options=solid]	table[x index=0,y index=2, col sep=comma]{figures/ex2GlinD.csv};
	\addplot+[blue, thick, dashed, mark = square, mark size = 2pt, mark options=solid]	table[x index=0,y index=5, col sep=comma]{figures/ex2GlinD.csv};
	\addplot+[red, thick, solid, mark = star, mark options=solid]	table[x index=0,y index=8, col sep=comma]{figures/ex2GlinD.csv};
	\end{semilogyaxis}
	\begin{semilogyaxis}[name=plot8,height=1.8in,width=0.315\textwidth,at={($(plot7.east)+(0.4in,0in)$)},anchor=west,
    title = {$t$},
	every tick label/.append style={font=\scriptsize},
	axis on top,
 	scaled x ticks = false,    x dir=reverse,
	xticklabel style={/pgf/number format/fixed},
	/pgf/number format/1000 sep={}]
	\addplot+[black, densely dotted, thick, mark = o, mark size = 2pt, mark options=solid]	table[x index=0,y index=2, col sep=comma]{figures/ex2TlinD.csv};
	\addplot+[blue, thick, dashed, mark = square, mark size = 2pt, mark options=solid]	table[x index=0,y index=5, col sep=comma]{figures/ex2TlinD.csv};
	\addplot+[red, thick, solid, mark = star, mark options=solid]	table[x index=0,y index=8, col sep=comma]{figures/ex2TlinD.csv};
	\end{semilogyaxis}
	\begin{semilogyaxis}[name=plot9,height=1.8in,width=0.315\textwidth,at={($(plot8.east)+(0.4in,0in)$)},anchor=west,
    title = {Variance-gamma},
	every tick label/.append style={font=\scriptsize},
	axis on top,
 	scaled x ticks = false,    x dir=reverse,
	xticklabel style={/pgf/number format/fixed},
	/pgf/number format/1000 sep={},
    legend pos=outer north east, legend entries={MC, ILIS, BLISS}]
	\addplot+[black, densely dotted, thick, mark = o, mark size = 2pt, mark options=solid]	table[x index=0,y index=2, col sep=comma]{figures/ex2VGlinD.csv};
	\addplot+[blue, thick, dashed, mark = square, mark size = 2pt, mark options=solid]	table[x index=0,y index=5, col sep=comma]{figures/ex2VGlinD.csv};
	\addplot+[red, thick, solid, mark = star, mark options=solid]	table[x index=0,y index=8, col sep=comma]{figures/ex2VGlinD.csv};
	\end{semilogyaxis}
	\begin{semilogyaxis}[name=plot10,height=1.8in,width=0.315\textwidth,at={($(plot7.south)-(0in,0.2in)$)},anchor=north,
	xlabel={$\zeta~(\times10^{-8})$},
	ylabel={Efficiency ratio},
	every tick label/.append style={font=\scriptsize},
	axis on top,
 	scaled x ticks = false,    x dir=reverse,
	xticklabel style={/pgf/number format/fixed},
	/pgf/number format/1000 sep={}]
	\addplot+[black, densely dotted, thick, mark = o, mark size = 2pt, mark options=solid]	table[x index=0,y index=3, col sep=comma]{figures/ex2GlinD.csv};
	\addplot+[blue, thick, dashed, mark = square, mark size = 2pt, mark options=solid]	table[x index=0,y index=6, col sep=comma]{figures/ex2GlinD.csv};
	\addplot+[red, thick, solid, mark = star, mark options=solid]	table[x index=0,y index=9, col sep=comma]{figures/ex2GlinD.csv};
	\end{semilogyaxis}
	\begin{semilogyaxis}[name=plot11,height=1.8in,width=0.315\textwidth,at={($(plot10.east)+(0.4in,0in)$)},anchor=west,
	xlabel={$\zeta~(\times10^{-8})$},
	every tick label/.append style={font=\scriptsize},
	axis on top,
 	scaled x ticks = false,    x dir=reverse,
	xticklabel style={/pgf/number format/fixed},
	/pgf/number format/1000 sep={}]
	\addplot+[black, densely dotted, thick, mark = o, mark size = 2pt, mark options=solid]	table[x index=0,y index=3, col sep=comma]{figures/ex2TlinD.csv};
	\addplot+[blue, thick, dashed, mark = square, mark size = 2pt, mark options=solid]	table[x index=0,y index=6, col sep=comma]{figures/ex2TlinD.csv};
	\addplot+[red, thick, solid, mark = star, mark options=solid]	table[x index=0,y index=9, col sep=comma]{figures/ex2TlinD.csv};
	\end{semilogyaxis}
	\begin{semilogyaxis}[name=plot12,height=1.8in,width=0.315\textwidth,at={($(plot11.east)+(0.4in,0in)$)},anchor=west,
	xlabel={$\zeta~(\times10^{-8})$},
	every tick label/.append style={font=\scriptsize},
	axis on top,
 	scaled x ticks = false,    x dir=reverse,
	xticklabel style={/pgf/number format/fixed},
	/pgf/number format/1000 sep={}]
	\addplot+[black, densely dotted, thick, mark = o, mark size = 2pt, mark options=solid]	table[x index=0,y index=3, col sep=comma]{figures/ex2VGlinD.csv};
	\addplot+[blue, thick, dashed, mark = square, mark size = 2pt, mark options=solid]	table[x index=0,y index=6, col sep=comma]{figures/ex2VGlinD.csv};
	\addplot+[red, thick, solid, mark = star, mark options=solid]	table[x index=0,y index=9, col sep=comma]{figures/ex2VGlinD.csv};
	\end{semilogyaxis}
    \node[yshift=29pt] 
    at ($(plot7.north west)!0.5!(plot9.north east)$) 
    {(b) Diagonal-loading case};
	\end{tikzpicture}
    \edit{\caption{\label{fig:lin_nu} Performance comparison between MC, ILIS, and BLISS for the EBA example in Section~\ref{subsec:numeric_2} across different values of the price-impact factor $(\zeta)$ with the linear inverse demand function. The three columns present the results with Gaussian, $t$, and variance-gamma log-return models, respectively. The upper six panels correspond to the \edit{non-diagonal-loading} case, while the lower six panels describe the diagonal-loading case. In each of the two cases, the top three panels compare standard errors of the three methods, and the bottom three panels exhibit their efficiency ratios.}}
\end{figure}
\begin{figure}[!t]
\centering
	\begin{tikzpicture}[font=\sffamily\small]
	\begin{semilogyaxis}[name=plot1,height=1.8in,width=0.315\textwidth,
    title = {Gaussian},
	ylabel={Standard error},
	every tick label/.append style={font=\scriptsize},
	axis on top,
 	scaled x ticks = false,    x dir=reverse,
	xticklabel style={/pgf/number format/fixed},
	/pgf/number format/1000 sep={}]
	\addplot+[black, densely dotted, thick, mark = o, mark size = 2pt, mark options=solid]	table[x index=10,y index=12, col sep=comma]{figures/ex2GlinC.csv};
	\addplot+[blue, thick, dashed, mark = square, mark size = 2pt, mark options=solid]	table[x index=10,y index=15, col sep=comma]{figures/ex2GlinC.csv};
	\addplot+[red, thick, solid, mark = star, mark options=solid]	table[x index=10,y index=18, col sep=comma]{figures/ex2GlinC.csv};
	\end{semilogyaxis}
	\begin{semilogyaxis}[name=plot2,height=1.8in,width=0.315\textwidth,at={($(plot1.east)+(0.4in,0in)$)},anchor=west,
    title = {$t$},
	every tick label/.append style={font=\scriptsize},
	axis on top,
 	scaled x ticks = false,    x dir=reverse,
	xticklabel style={/pgf/number format/fixed},
	/pgf/number format/1000 sep={}]
	\addplot+[black, densely dotted, thick, mark = o, mark size = 2pt, mark options=solid]	table[x index=10,y index=12, col sep=comma]{figures/ex2TlinC.csv};
	\addplot+[blue, thick, dashed, mark = square, mark size = 2pt, mark options=solid]	table[x index=10,y index=15, col sep=comma]{figures/ex2TlinC.csv};
	\addplot+[red, thick, solid, mark = star, mark options=solid]	table[x index=10,y index=18, col sep=comma]{figures/ex2TlinC.csv};
	\end{semilogyaxis}
	\begin{semilogyaxis}[name=plot3,height=1.8in,width=0.315\textwidth,at={($(plot2.east)+(0.4in,0in)$)},anchor=west,
    title = {Variance-gamma},
	every tick label/.append style={font=\scriptsize},
	axis on top,
 	scaled x ticks = false,    x dir=reverse,
	xticklabel style={/pgf/number format/fixed},
	/pgf/number format/1000 sep={},
    legend pos=outer north east, legend entries={MC, ILIS, BLISS}]
	\addplot+[black, densely dotted, thick, mark = o, mark size = 2pt, mark options=solid]	table[x index=10,y index=12, col sep=comma]{figures/ex2VGlinC.csv};
	\addplot+[blue, thick, dashed, mark = square, mark size = 2pt, mark options=solid]	table[x index=10,y index=15, col sep=comma]{figures/ex2VGlinC.csv};
	\addplot+[red, thick, solid, mark = star, mark options=solid]	table[x index=10,y index=18, col sep=comma]{figures/ex2VGlinC.csv};
	\end{semilogyaxis}
	\begin{semilogyaxis}[name=plot4,height=1.8in,width=0.315\textwidth,at={($(plot1.south)-(0in,0.2in)$)},anchor=north,
	xlabel={$\delta$},
	ylabel={Efficiency ratio},
	every tick label/.append style={font=\scriptsize},
	axis on top,
 	scaled x ticks = false,    x dir=reverse,
	xticklabel style={/pgf/number format/fixed},
	/pgf/number format/1000 sep={}]
	\addplot+[black, densely dotted, thick, mark = o, mark size = 2pt, mark options=solid]	table[x index=10,y index=13, col sep=comma]{figures/ex2GlinC.csv};
	\addplot+[blue, thick, dashed, mark = square, mark size = 2pt, mark options=solid]	table[x index=10,y index=16, col sep=comma]{figures/ex2GlinC.csv};
	\addplot+[red, thick, solid, mark = star, mark options=solid]	table[x index=10,y index=19, col sep=comma]{figures/ex2GlinC.csv};
	\end{semilogyaxis}
	\begin{semilogyaxis}[name=plot5,height=1.8in,width=0.315\textwidth,at={($(plot4.east)+(0.4in,0in)$)},anchor=west,
	xlabel={$\delta$},
	every tick label/.append style={font=\scriptsize},
	axis on top,
 	scaled x ticks = false,    x dir=reverse,
	xticklabel style={/pgf/number format/fixed},
	/pgf/number format/1000 sep={}]
	\addplot+[black, densely dotted, thick, mark = o, mark size = 2pt, mark options=solid]	table[x index=10,y index=13, col sep=comma]{figures/ex2TlinC.csv};
	\addplot+[blue, thick, dashed, mark = square, mark size = 2pt, mark options=solid]	table[x index=10,y index=16, col sep=comma]{figures/ex2TlinC.csv};
	\addplot+[red, thick, solid, mark = star, mark options=solid]	table[x index=10,y index=19, col sep=comma]{figures/ex2TlinC.csv};
	\end{semilogyaxis}
	\begin{semilogyaxis}[name=plot6,height=1.8in,width=0.315\textwidth,at={($(plot5.east)+(0.4in,0in)$)},anchor=west,
	xlabel={$\delta$},
	every tick label/.append style={font=\scriptsize},
	axis on top,
 	scaled x ticks = false,    x dir=reverse,
	xticklabel style={/pgf/number format/fixed},
	/pgf/number format/1000 sep={}]
	\addplot+[black, densely dotted, thick, mark = o, mark size = 2pt, mark options=solid]	table[x index=10,y index=13, col sep=comma]{figures/ex2VGlinC.csv};
	\addplot+[blue, thick, dashed, mark = square, mark size = 2pt, mark options=solid]	table[x index=10,y index=16, col sep=comma]{figures/ex2VGlinC.csv};
	\addplot+[red, thick, solid, mark = star, mark options=solid]	table[x index=10,y index=19, col sep=comma]{figures/ex2VGlinC.csv};
	\end{semilogyaxis}
    \node[yshift=29pt] 
    at ($(plot1.north west)!0.5!(plot3.north east)$) 
    {(a) Non-diagonal-loading case};
    
	\begin{semilogyaxis}[name=plot7,height=1.8in,width=0.315\textwidth,at={($(plot4.south)-(0in,1in)$)},anchor=north,
    title = {Gaussian},
	ylabel={Standard error},
	every tick label/.append style={font=\scriptsize},
	axis on top,
 	scaled x ticks = false,    x dir=reverse,
	xticklabel style={/pgf/number format/fixed},
	/pgf/number format/1000 sep={}]
	\addplot+[black, densely dotted, thick, mark = o, mark size = 2pt, mark options=solid]	table[x index=10,y index=12, col sep=comma]{figures/ex2GlinD.csv};
	\addplot+[blue, thick, dashed, mark = square, mark size = 2pt, mark options=solid]	table[x index=10,y index=15, col sep=comma]{figures/ex2GlinD.csv};
	\addplot+[red, thick, solid, mark = star, mark options=solid]	table[x index=10,y index=18, col sep=comma]{figures/ex2GlinD.csv};
	\end{semilogyaxis}
	\begin{semilogyaxis}[name=plot8,height=1.8in,width=0.315\textwidth,at={($(plot7.east)+(0.4in,0in)$)},anchor=west,
    title = {$t$},
	every tick label/.append style={font=\scriptsize},
	axis on top,
 	scaled x ticks = false,    x dir=reverse,
	xticklabel style={/pgf/number format/fixed},
	/pgf/number format/1000 sep={}]
	\addplot+[black, densely dotted, thick, mark = o, mark size = 2pt, mark options=solid]	table[x index=10,y index=12, col sep=comma]{figures/ex2TlinD.csv};
	\addplot+[blue, thick, dashed, mark = square, mark size = 2pt, mark options=solid]	table[x index=10,y index=15, col sep=comma]{figures/ex2TlinD.csv};
	\addplot+[red, thick, solid, mark = star, mark options=solid]	table[x index=10,y index=18, col sep=comma]{figures/ex2TlinD.csv};
	\end{semilogyaxis}
	\begin{semilogyaxis}[name=plot9,height=1.8in,width=0.315\textwidth,at={($(plot8.east)+(0.4in,0in)$)},anchor=west,
    title = {Variance-gamma},
	every tick label/.append style={font=\scriptsize},
	axis on top,
 	scaled x ticks = false,    x dir=reverse,
	xticklabel style={/pgf/number format/fixed},
	/pgf/number format/1000 sep={},
    legend pos=outer north east, legend entries={MC, ILIS, BLISS}]
	\addplot+[black, densely dotted, thick, mark = o, mark size = 2pt, mark options=solid]	table[x index=10,y index=12, col sep=comma]{figures/ex2VGlinD.csv};
	\addplot+[blue, thick, dashed, mark = square, mark size = 2pt, mark options=solid]	table[x index=10,y index=15, col sep=comma]{figures/ex2VGlinD.csv};
	\addplot+[red, thick, solid, mark = star, mark options=solid]	table[x index=10,y index=18, col sep=comma]{figures/ex2VGlinD.csv};
	\end{semilogyaxis}
	\begin{semilogyaxis}[name=plot10,height=1.8in,width=0.315\textwidth,at={($(plot7.south)-(0in,0.2in)$)},anchor=north,
	xlabel={$\delta$},
	ylabel={Efficiency ratio},
	every tick label/.append style={font=\scriptsize},
	axis on top,
 	scaled x ticks = false,    x dir=reverse,
	xticklabel style={/pgf/number format/fixed},
	/pgf/number format/1000 sep={}]
	\addplot+[black, densely dotted, thick, mark = o, mark size = 2pt, mark options=solid]	table[x index=10,y index=13, col sep=comma]{figures/ex2GlinD.csv};
	\addplot+[blue, thick, dashed, mark = square, mark size = 2pt, mark options=solid]	table[x index=10,y index=16, col sep=comma]{figures/ex2GlinD.csv};
	\addplot+[red, thick, solid, mark = star, mark options=solid]	table[x index=10,y index=19, col sep=comma]{figures/ex2GlinD.csv};
	\end{semilogyaxis}
	\begin{semilogyaxis}[name=plot11,height=1.8in,width=0.315\textwidth,at={($(plot10.east)+(0.4in,0in)$)},anchor=west,
	xlabel={$\delta$},
	every tick label/.append style={font=\scriptsize},
	axis on top,
 	scaled x ticks = false,    x dir=reverse,
	xticklabel style={/pgf/number format/fixed},
	/pgf/number format/1000 sep={}]
	\addplot+[black, densely dotted, thick, mark = o, mark size = 2pt, mark options=solid]	table[x index=10,y index=13, col sep=comma]{figures/ex2TlinD.csv};
	\addplot+[blue, thick, dashed, mark = square, mark size = 2pt, mark options=solid]	table[x index=10,y index=16, col sep=comma]{figures/ex2TlinD.csv};
	\addplot+[red, thick, solid, mark = star, mark options=solid]	table[x index=10,y index=19, col sep=comma]{figures/ex2TlinD.csv};
	\end{semilogyaxis}
	\begin{semilogyaxis}[name=plot12,height=1.8in,width=0.315\textwidth,at={($(plot11.east)+(0.4in,0in)$)},anchor=west,
	xlabel={$\delta$},
	every tick label/.append style={font=\scriptsize},
	axis on top,
 	scaled x ticks = false,    x dir=reverse,
	xticklabel style={/pgf/number format/fixed},
	/pgf/number format/1000 sep={}]
	\addplot+[black, densely dotted, thick, mark = o, mark size = 2pt, mark options=solid]	table[x index=10,y index=13, col sep=comma]{figures/ex2VGlinD.csv};
	\addplot+[blue, thick, dashed, mark = square, mark size = 2pt, mark options=solid]	table[x index=10,y index=16, col sep=comma]{figures/ex2VGlinD.csv};
	\addplot+[red, thick, solid, mark = star, mark options=solid]	table[x index=10,y index=19, col sep=comma]{figures/ex2VGlinD.csv};
	\end{semilogyaxis}
    \node[yshift=29pt] 
    at ($(plot7.north west)!0.5!(plot9.north east)$) 
    {(b) Diagonal-loading case};
	\end{tikzpicture}
    \vspace{-12pt}
    \edit{\caption{\label{fig:lin_delta} Performance comparison between MC, ILIS, and BLISS for the EBA example in Section~\ref{subsec:numeric_2} across different values of the bankruptcy-cost factor $(\delta)$ with the linear inverse demand function. The three columns present the results with Gaussian, $t$, and variance-gamma log-return models, respectively. The upper six panels correspond to the \edit{non-diagonal-loading} case, while the lower six panels describe the diagonal-loading case. In each of the two cases, the top three panels compare standard errors of the three methods, and the bottom three panels exhibit their efficiency ratios.}}
\end{figure}

\begin{figure}[!t]
\centering
	\begin{tikzpicture}[font=\sffamily\small]
	\begin{semilogyaxis}[name=plot1,height=1.8in,width=0.315\textwidth,
    title = {Gaussian},
	ylabel={Standard error},
    xtick={0,1,2,3},
    xticklabels={Beta,Uniform,Priority},
	every tick label/.append style={font=\scriptsize},
	axis on top]
	\addplot+[black, only marks, thick, mark = o, mark size = 2pt, mark options=solid]	table[x expr=\coordindex,y index=2, col sep=comma, skip first n=1]{figures/ex2GlinC2.csv};
	\addplot+[blue, only marks, thick, mark = square, mark size = 2pt, mark options=solid]	table[x expr=\coordindex,y index=5, col sep=comma, skip first n=1]{figures/ex2GlinC2.csv};
	\addplot+[red, only marks, thick, mark = star, mark options=solid]	table[x expr=\coordindex,y index=8, col sep=comma, skip first n=1]{figures/ex2GlinC2.csv};
	\end{semilogyaxis}
	\begin{semilogyaxis}[name=plot2,height=1.8in,width=0.315\textwidth,at={($(plot1.east)+(0.4in,0in)$)},anchor=west,
    title = {$t$},
    xtick={0,1,2,3},
    xticklabels={Beta,Uniform,Priority},
	every tick label/.append style={font=\scriptsize},
	axis on top]
	\addplot+[black, only marks, thick, mark = o, mark size = 2pt, mark options=solid]	table[x expr=\coordindex,y index=2, col sep=comma, skip first n=1]{figures/ex2TlinC2.csv};
	\addplot+[blue, only marks, thick, mark = square, mark size = 2pt, mark options=solid]	table[x expr=\coordindex,y index=5, col sep=comma, skip first n=1]{figures/ex2TlinC2.csv};
	\addplot+[red, only marks, thick, mark = star, mark options=solid]	table[x expr=\coordindex,y index=8, col sep=comma, skip first n=1]{figures/ex2TlinC2.csv};
	\end{semilogyaxis}
	\begin{semilogyaxis}[name=plot3,height=1.8in,width=0.315\textwidth,at={($(plot2.east)+(0.4in,0in)$)},anchor=west,
    title = {Variance-gamma},
    xtick={0,1,2,3},
    xticklabels={Beta,Uniform,Priority},
	every tick label/.append style={font=\scriptsize},
	axis on top,
    legend pos=outer north east, legend entries={MC, ILIS, BLISS}]
	\addplot+[black, only marks, thick, mark = o, mark size = 2pt, mark options=solid]	table[x expr=\coordindex,y index=2, col sep=comma, skip first n=1]{figures/ex2VGlinC2.csv};
	\addplot+[blue, only marks, thick, mark = square, mark size = 2pt, mark options=solid]	table[x expr=\coordindex,y index=5, col sep=comma, skip first n=1]{figures/ex2VGlinC2.csv};
	\addplot+[red, only marks, thick, mark = star, mark options=solid]	table[x expr=\coordindex,y index=8, col sep=comma, skip first n=1]{figures/ex2VGlinC2.csv};
	\end{semilogyaxis}
	\begin{semilogyaxis}[name=plot4,height=1.8in,width=0.315\textwidth,at={($(plot1.south)-(0in,0.2in)$)},anchor=north,
	ylabel={Efficiency ratio},
    xtick={0,1,2,3},
    xticklabels={Beta,Uniform,Priority},
	every tick label/.append style={font=\scriptsize},
	axis on top]
	\addplot+[black, only marks, thick, mark = o, mark size = 2pt, mark options=solid]	table[x expr=\coordindex,y index=3, col sep=comma, skip first n=1]{figures/ex2GlinC2.csv};
	\addplot+[blue, only marks, thick, mark = square, mark size = 2pt, mark options=solid]	table[x expr=\coordindex,y index=6, col sep=comma, skip first n=1]{figures/ex2GlinC2.csv};
	\addplot+[red, only marks, thick, mark = star, mark options=solid]	table[x expr=\coordindex,y index=9, col sep=comma, skip first n=1]{figures/ex2GlinC2.csv};
	\end{semilogyaxis}
	\begin{semilogyaxis}[name=plot5,height=1.8in,width=0.315\textwidth,at={($(plot4.east)+(0.4in,0in)$)},anchor=west,
    xtick={0,1,2,3},
    xticklabels={Beta,Uniform,Priority},
	every tick label/.append style={font=\scriptsize},
	axis on top]
	\addplot+[black, only marks, thick, mark = o, mark size = 2pt, mark options=solid]	table[x expr=\coordindex,y index=3, col sep=comma, skip first n=1]{figures/ex2TlinC2.csv};
	\addplot+[blue, only marks, thick, mark = square, mark size = 2pt, mark options=solid]	table[x expr=\coordindex,y index=6, col sep=comma, skip first n=1]{figures/ex2TlinC2.csv};
	\addplot+[red, only marks, thick, mark = star, mark options=solid]	table[x expr=\coordindex,y index=9, col sep=comma, skip first n=1]{figures/ex2TlinC2.csv};
	\end{semilogyaxis}
	\begin{semilogyaxis}[name=plot6,height=1.8in,width=0.315\textwidth,at={($(plot5.east)+(0.4in,0in)$)},anchor=west,
    xtick={0,1,2,3},
    xticklabels={Beta,Uniform,Priority},
	every tick label/.append style={font=\scriptsize},
	axis on top]
	\addplot+[black, only marks, thick, mark = o, mark size = 2pt, mark options=solid]	table[x expr=\coordindex,y index=3, col sep=comma, skip first n=1]{figures/ex2VGlinC2.csv};
	\addplot+[blue, only marks, thick, mark = square, mark size = 2pt, mark options=solid]	table[x expr=\coordindex,y index=6, col sep=comma, skip first n=1]{figures/ex2VGlinC2.csv};
	\addplot+[red, only marks, thick, mark = star, mark options=solid]	table[x expr=\coordindex,y index=9, col sep=comma, skip first n=1]{figures/ex2VGlinC2.csv};
	\end{semilogyaxis}
    \node[yshift=29pt] 
    at ($(plot1.north west)!0.5!(plot3.north east)$) 
    {(a) Non-diagonal-loading case};
    
	\begin{semilogyaxis}[name=plot7,height=1.8in,width=0.315\textwidth,at={($(plot4.south)-(0in,1in)$)},anchor=north,
    title = {Gaussian},
	ylabel={Standard error},
    xtick={0,1,2,3},
    xticklabels={Beta,Uniform,Priority},
	every tick label/.append style={font=\scriptsize},
	axis on top]
	\addplot+[black, only marks, thick, mark = o, mark size = 2pt, mark options=solid]	table[x expr=\coordindex,y index=2, col sep=comma, skip first n=1]{figures/ex2GlinD2.csv};
	\addplot+[blue, only marks, thick, mark = square, mark size = 2pt, mark options=solid]	table[x expr=\coordindex,y index=5, col sep=comma, skip first n=1]{figures/ex2GlinD2.csv};
	\addplot+[red, only marks, thick, mark = star, mark options=solid]	table[x expr=\coordindex,y index=8, col sep=comma, skip first n=1]{figures/ex2GlinD2.csv};
	\end{semilogyaxis}
	\begin{semilogyaxis}[name=plot8,height=1.8in,width=0.315\textwidth,at={($(plot7.east)+(0.4in,0in)$)},anchor=west,
    title = {$t$},
    xtick={0,1,2,3},
    xticklabels={Beta,Uniform,Priority},
	every tick label/.append style={font=\scriptsize},
	axis on top]
	\addplot+[black, only marks, thick, mark = o, mark size = 2pt, mark options=solid]	table[x expr=\coordindex,y index=2, col sep=comma, skip first n=1]{figures/ex2TlinD2.csv};
	\addplot+[blue, only marks, thick, mark = square, mark size = 2pt, mark options=solid]	table[x expr=\coordindex,y index=5, col sep=comma, skip first n=1]{figures/ex2TlinD2.csv};
	\addplot+[red, only marks, thick, mark = star, mark options=solid]	table[x expr=\coordindex,y index=8, col sep=comma, skip first n=1]{figures/ex2TlinD2.csv};
	\end{semilogyaxis}
	\begin{semilogyaxis}[name=plot9,height=1.8in,width=0.315\textwidth,at={($(plot8.east)+(0.4in,0in)$)},anchor=west,
    title = {Variance-gamma},
    xtick={0,1,2,3},
    xticklabels={Beta,Uniform,Priority},
	every tick label/.append style={font=\scriptsize},
	axis on top,
    legend pos=outer north east, legend entries={MC, ILIS, BLISS}]
	\addplot+[black, only marks, thick, mark = o, mark size = 2pt, mark options=solid]	table[x expr=\coordindex,y index=2, col sep=comma, skip first n=1]{figures/ex2VGlinD2.csv};
	\addplot+[blue, only marks, thick, mark = square, mark size = 2pt, mark options=solid]	table[x expr=\coordindex,y index=5, col sep=comma, skip first n=1]{figures/ex2VGlinD2.csv};
	\addplot+[red, only marks, thick, mark = star, mark options=solid]	table[x expr=\coordindex,y index=8, col sep=comma, skip first n=1]{figures/ex2VGlinD2.csv};
	\end{semilogyaxis}
	\begin{semilogyaxis}[name=plot10,height=1.8in,width=0.315\textwidth,at={($(plot7.south)-(0in,0.2in)$)},anchor=north,
	ylabel={Efficiency ratio},
    xtick={0,1,2,3},
    xticklabels={Beta,Uniform,Priority},
	every tick label/.append style={font=\scriptsize},
	axis on top]
	\addplot+[black, only marks, thick, mark = o, mark size = 2pt, mark options=solid]	table[x expr=\coordindex,y index=3, col sep=comma, skip first n=1]{figures/ex2GlinD2.csv};
	\addplot+[blue, only marks, thick, mark = square, mark size = 2pt, mark options=solid]	table[x expr=\coordindex,y index=6, col sep=comma, skip first n=1]{figures/ex2GlinD2.csv};
	\addplot+[red, only marks, thick, mark = star, mark options=solid]	table[x expr=\coordindex,y index=9, col sep=comma, skip first n=1]{figures/ex2GlinD2.csv};
	\end{semilogyaxis}
	\begin{semilogyaxis}[name=plot11,height=1.8in,width=0.315\textwidth,at={($(plot10.east)+(0.4in,0in)$)},anchor=west,
    xtick={0,1,2,3},
    xticklabels={Beta,Uniform,Priority},
	every tick label/.append style={font=\scriptsize},
	axis on top]
	\addplot+[black, only marks, thick, mark = o, mark size = 2pt, mark options=solid]	table[x expr=\coordindex,y index=3, col sep=comma, skip first n=1]{figures/ex2TlinD2.csv};
	\addplot+[blue, only marks, thick, mark = square, mark size = 2pt, mark options=solid]	table[x expr=\coordindex,y index=6, col sep=comma, skip first n=1]{figures/ex2TlinD2.csv};
	\addplot+[red, only marks, thick, mark = star, mark options=solid]	table[x expr=\coordindex,y index=9, col sep=comma, skip first n=1]{figures/ex2TlinD2.csv};
	\end{semilogyaxis}
	\begin{semilogyaxis}[name=plot12,height=1.8in,width=0.315\textwidth,at={($(plot11.east)+(0.4in,0in)$)},anchor=west,
    xtick={0,1,2,3},
    xticklabels={Beta,Uniform,Priority},
	every tick label/.append style={font=\scriptsize},
	axis on top]
	\addplot+[black, only marks, thick, mark = o, mark size = 2pt, mark options=solid]	table[x expr=\coordindex,y index=3, col sep=comma, skip first n=1]{figures/ex2VGlinD2.csv};
	\addplot+[blue, only marks, thick, mark = square, mark size = 2pt, mark options=solid]	table[x expr=\coordindex,y index=6, col sep=comma, skip first n=1]{figures/ex2VGlinD2.csv};
	\addplot+[red, only marks, thick, mark = star, mark options=solid]	table[x expr=\coordindex,y index=9, col sep=comma, skip first n=1]{figures/ex2VGlinD2.csv};
	\end{semilogyaxis}
    \node[yshift=29pt] 
    at ($(plot7.north west)!0.5!(plot9.north east)$) 
    {(b) Diagonal-loading case};
	\end{tikzpicture}
    \edit{\caption{\label{fig:lin_scenario} Performance comparison between MC, ILIS, and BLISS for the EBA example in Section~\ref{subsec:numeric_2} across different random price impact scenarios and external-liability priority scenarios with the linear inverse demand function. The three columns present the results with Gaussian, $t$, and variance-gamma log-return models, respectively. The upper six panels correspond to the \edit{non-diagonal-loading} case, while the lower six panels describe the diagonal-loading case. In each of the two cases, the top three panels compare standard errors of the three methods, and the bottom three panels exhibit their efficiency ratios.}}
\end{figure}
    
\end{revision}

\newcolumntype{C}[1]{>{\centering\arraybackslash}p{#1}}
\begin{table}[!t]
    \centering\footnotesize
\edit{\begin{tabular}{C{3.3cm}C{1.5cm}C{1.5cm}C{1.5cm}cC{1.5cm}C{1.5cm}C{1.5cm}}
\toprule
Price-impact factor&\multicolumn{3}{c}{Non-diagonal-loading}&&\multicolumn{3}{c}{Diagonal-loading}\\\cmidrule{2-4}\cmidrule{6-8}
$\zeta$ ($\times10^{-8}$)&Gaussian&$t$&VG&&Gaussian&$t$&VG \\\midrule
2.0   & 99 & 120 & 112 &  & 96 & 116 & 110 \\
1.6 & 64 & 92  & 83 &  & 62 & 90  & 81 \\
1.2 & 39 & 71  & 60 &  & 37 & 69  & 58 \\
0.8 & 22 & 54  & 42 &  & 21 & 53  & 41 \\
0.4 & 11 & 41  & 29 &  & 11 & 40  & 29 \\
0   & 5  & 31  & 20 &  & 5  & 31  & 19 \\\bottomrule 
\end{tabular}
\begin{tabular}{C{3.3cm}C{1.5cm}C{1.5cm}C{1.5cm}cC{1.5cm}C{1.5cm}C{1.5cm}}
\toprule
Bankruptcy-cost factor&\multicolumn{3}{c}{Non-diagonal-loading}&&\multicolumn{3}{c}{Diagonal-loading}\\\cmidrule{2-4}\cmidrule{6-8}
$\delta$&Gaussian&$t$&VG&&Gaussian&$t$&VG \\\midrule
0.5 & 50 & 104 & 85 &  & 47 & 101 & 83 \\
0.4 & 45 & 95  & 78 &  & 43 & 92  & 76 \\
0.3 & 41 & 86  & 70 &  & 39 & 84  & 69 \\
0.2 & 37 & 78  & 64 &  & 36 & 76  & 62 \\
0.1 & 33 & 70  & 57 &  & 32 & 68  & 55 \\
0   & 29 & 62  & 50 &  & 28 & 60  & 49\\\bottomrule
\end{tabular}

    \caption{\label{tab:lin_12}Yield estimates (in basis points) for different price-impact and bankruptcy-cost factors with the linear inverse demand function. The upper and lower tables show the estimates under the experimental settings of Figures~\ref{fig:lin_nu} and~\ref{fig:lin_delta}, respectively. VG stands for variance-gamma.}}
    
\end{table}

\begin{table}[!t]
    \centering\footnotesize
\edit{\begin{tabular}{C{3.3cm}C{1.5cm}C{1.5cm}C{1.5cm}cC{1.5cm}C{1.5cm}C{1.5cm}}
\toprule
&\multicolumn{3}{c}{Non-diagonal-loading}&&\multicolumn{3}{c}{Diagonal-loading}\\\cmidrule{2-4}\cmidrule{6-8}
Scenario&Gaussian&$t$&VG&&Gaussian&$t$&VG \\\midrule
Beta     & 34  & 65  & 54 && 33  & 63  & 52 \\
Uniform  & 37  & 66  & 56 && 36  & 65  & 54 \\
Priority & 43 & 74 & 79 && 30 & 64 & 69 \\\bottomrule 
\end{tabular}
\caption{\label{tab:lin_34}
Yield estimates (in basis points) across different random price impact scenarios and external-liability priority scenarios with the linear inverse demand function. VG stands for variance-gamma.}}
\end{table}
\begin{revision}
\section{Further Generalized Hyperbolic Examples}\label{app:GH-ex}
In this section, we describe the conditional sampling construction in Proposition~\ref{prop:gh_sampler} for two distributions other than those discussed in Section~\ref{subsec:gh_examples}.

\textbf{Normal-inverse-Gaussian.} The normal-inverse-Gaussian (NIG) distribution also allows for skewness and semi-heavy tails and has been used to model financial returns and stochastic volatility \citep{BarndorffNielsen1997NIG}. This case corresponds to $\lambda=-1/2$ and $\chi,\psi>0$. Then, $T_n$ follows a univariate NIG distribution, and the conditional mixing distribution is
$$
\widetilde V\mid(\widetilde T=t)
\sim
\operatorname{GIG}\!\left(
-1,\,\chi+t^2,\,\psi+\eta_n^2
\right).
$$

\textbf{Hyperbolic.} The hyperbolic distribution provides another model with exponential tails that has been applied to financial returns \citep{EberleinKeller1995}. For $\lambda=1$ and $\chi,\psi>0$, $T_n$ follows a univariate hyperbolic distribution. In this case, the conditional mixing distribution becomes
$$
\widetilde V\mid(\widetilde T=t)
\sim
\operatorname{GIG}\!\left(
\frac12,\,\chi+t^2,\,\psi+\eta_n^2
\right).
$$
\end{revision}

\bibliography{reference}
\end{document}